\documentclass[11pt]{article}
\usepackage{latexsym}
\usepackage{amssymb}
\usepackage{graphicx}
\usepackage{mathtools,mathrsfs}
\usepackage{amssymb}
\usepackage{amsmath}
\usepackage{color}
\usepackage{tikz}
\usepackage{hyperref}
\hypersetup{hypertex=true,
	colorlinks=true,
	linkcolor=blue,
	anchorcolor=blue,
	citecolor=blue}
\usepackage{amsthm}
\usepackage{enumerate}
\usepackage{cleveref}
\crefformat{section}{\S#2#1#3} 
\crefformat{subsection}{\S#2#1#3}
\crefformat{subsubsection}{\S#2#1#3}
\newcommand{\C}{\mathbb C}
\newcommand{\R}{\mathbb R}
\newcommand{\Z}{\mathbb Z}
\newcommand{\F}{\mathbf{F}}

\newcommand{\cH}{\mathcal H}
\newcommand{\cE}{\mathcal E}
\newcommand{\cT}{\mathscr T}
\newcommand{\Str}{\operatorname{Str}}

\newcommand{\Tr}{\operatorname{Tr}}
\newcommand{\FP}{\operatorname{FP}}
\def\str{\mathrm{str}}

\newcommand{\ddbar}{\partial_B\bar\partial_B}
\newcommand{\ph}{\widehat p}
\newcommand{\qh}{\widehat q}
\def\Re{{\mathrm{Re}\,}}
\newcommand{\dd}{\mathrm d}
\def\WP{{\mathrm{WP}}}
\def\CY{{\mathrm{CY}}}
\def\LG{{\mathrm{LG}}}

\def\p{\partial}
\def\Im{{\mathrm{Im}\,}}
\newcommand{\Id}{\mathrm{Id}}
\newcommand{\prim}{\mathrm{prim}}
\newcommand{\be}{\begin{equation}}
	\newcommand{\ee}{\end{equation}}
\newcommand{\ba}{\begin{aligned}}
	\newcommand{\ea}{\end{aligned}}

\newcommand{\End}{\operatorname{End}}
\newtheorem{theorem}{Theorem}[section]
\newtheorem{proposition}[theorem]{Proposition}
\newtheorem{lemma}[theorem]{Lemma}
\newtheorem{corollary}[theorem]{Corollary}
\theoremstyle{definition}

\newtheorem{definition}{Definition}[section]
\newtheorem{remark}[theorem]{Remark}

\numberwithin{equation}{section}
\allowdisplaybreaks[1]
\title{The Genus-One B-Model Calabi–Yau/Landau–Ginzburg Correspondence}
\author{Xinxing Tang
		\footnote{Beijing Institute of Mathematical Sciences and Applications, Beijing, China, tangxinxing@bimsa.cn.
		}
		\and Junrong Yan
		\footnote{Department of Mathematics, Northeastern University, Boston, United States, j.yan@northeastern.edu. 
		} 
	}
\date{}
\begin{document}
\maketitle
\vspace{-1.5em}

\begin{abstract}
    Let $f:\mathbb C^n\to\mathbb C$ be a homogeneous polynomial of degree $n$ with an isolated critical point at the origin. The Calabi--Yau/Landau--Ginzburg correspondence predicts a relation between invariants of the Landau--Ginzburg orbifold $(\mathbb C^n/\mathbb Z_n,f)$ and those of the associated Calabi--Yau hypersurface $X_f=\{f=0\}\subset\mathbb{P}^{n-1}$.
In this paper, we establish this correspondence for the genus-one B-model.
We define the Landau--Ginzburg genus-one B-model invariant intrinsically
from the $f$-twisted Laplacian on $\C^n/\Z_n$ as a modification of the
torsion-type invariant introduced by Fan and Fang, and prove that it
agrees with the corresponding BCOV invariant of $X_f$ up to an additive
constant depending only on the dimension. The proof depends on a holomorphic anomaly formula for families of LG models and a uniform estimate for the torsion near a one-node degeneration of $X_f$.

\end{abstract}
\tableofcontents
\section{Introduction}
\label{sec:introduction}
\subsection{Overview}

In the 1990s, physicists discovered mirror symmetry, a relationship between the symplectic geometry (A-model) of a Calabi-Yau manifold $X$ and the complex geometry (B-model) of its mirror $X^{\vee}$. Since its formulation more than three decades ago \cite{COGP1991,Greene1990DualityI}, the mirror symmetry conjecture has had a profound influence on mathematics. This paper focuses only on closed-string theory.

The A-model on a Calabi--Yau (CY) manifold $X$ is the Gromov--Witten
theory of $X$, whose invariants are notoriously difficult to compute. The B-model is often more amenable to explicit geometric and analytic descriptions. Its genus-zero theory is governed by variations of Hodge structure, whereas the higher genus B-model theory is known as the BCOV theory \cite{BCOV,CostelloLiBCOV}.
At genus one, Bershadsky--Cecotti--Ooguri--Vafa
\cite{BCOV1993,BCOV} related the amplitude to holomorphic
analytic torsion and derived its holomorphic anomaly equation.
The underlying spectral construction goes back to Ray--Singer
\cite{RaySinger1973}.
Fang--Lu--Yoshikawa \cite{FLY} constructed a metric-independent
BCOV invariant for CY threefolds and studied its variation
and degeneration. Eriksson--Freixas i Montplet--Mourougane
\cite{EFM} extended this construction and its degeneration theory
to arbitrary dimension. Zhang \cite{Zhang2022Extension,Zhang2023Blowup} extended the BCOV
invariant to CY pairs and established blow-up formulas.
Building on this work, Fu--Zhang \cite{FuZhang2023BCOV} proved
the birational invariance of the BCOV invariants. At higher genus, the construction of the B-model has taken a different direction. More recently, building on Wang--Yan's results
\cite{WangYan2025Graphs,WangYan2026Renormalizations} on the ultraviolet
(UV) finiteness of Feynman graph integrals in holomorphic field theories,
Li--Wang \cite{LiWang2026BCOV} announced an explicit all-genus quantization of BCOV theory on compact Calabi–Yau manifolds.

Simultaneously, physicists extended the above discussion to singularity theory, leading to what became known as the Landau–Ginzburg (LG) model. The LG model is defined on a pair $(X, f)$, where $X$ is a complete noncompact Kähler manifold and $f$ is a holomorphic function on $X$. For the A-side, the LG A-model is constructed by Fan-Jarvis-Ruan \cite{FJR2013} following Witten's proposal \cite{Witten1993}, which is known as the FJRW-theory.  Chang–Li–Li \cite{chang2015witten} subsequently gave an algebraic construction of Witten's top Chern class via cosection localization. On the B-side, Saito's theories of primitive form  \cite{1981Primitive} and  the higher residue pairing \cite{Saitohigher} generate the structure of a Frobenius manifold on the singularity's universal deformation space \cite{SaitoFM}, which yield the genus-0 theory of LG B-model.
A polyvector-field formulation was developed by Li--Li--Saito
\cite{LiLiSaito}.
Also, some Hodge theoretical aspect in the LG B-model, called $tt^*$ geometry, is discovered by Cecotti-Vafa \cite{cecotti1991topological}.
On the analytic side, Fan \cite{Fan} introduced an $L^2$-theoretic approach in the spirit of $N=2$ geometry by studying the spectral theory of a twisted Laplacian operator.
Besides, Li-Wen \cite{li2019l2} employed $L^2$ Hodge theory to construct a Frobenius manifold structure for the case where $f$ has a compact critical locus.
To construct genus-one invariants, Fan--Fang \cite{FF} introduced
torsion-type invariants using the heat kernel of the twisted Laplacian. Tang
\cite{Tang2017Torsion} further studied the second torsion-type invariant. Our
construction of the LG genus-one B-model invariant is obtained by modifying
this second torsion-type invariant.

We now turn to a brief review of mirror symmetry. On the Calabi--Yau side,
Givental \cite{Givental1998Mirror} and Lian--Liu--Yau
\cite{LianLiuYau1997} established genus-zero mirror symmetry for
Calabi--Yau hypersurfaces and related complete intersections.  At genus one, Getzler \cite{Getzler1997} derived a universal relation among genus-one Gromov–Witten invariants from the intersection theory of $\overline{\mathcal M}_{1,4}$, providing an important constraint on elliptic Gromov–Witten theory. 
Zinger \cite{Zinger2009}
computed the reduced Gromov--Witten invariants of Calabi--Yau hypersurfaces and established the
genus-one mirror formula for the quintic. Eriksson--Freixas i
Montplet--Mourougane \cite{EFMMirror} proved the mirror symmetry formula for
genus-one BCOV invariants; combined with Zinger's results, their work
established genus-one mirror symmetry for Calabi--Yau hypersurfaces in
projective space and their mirror families.
On the LG side, He--Li--Shen--Webb \cite{HeLiShenWebb2022}
proved an all-genus mirror theorem relating FJRW theory with
maximal diagonal symmetry to the Saito--Givental theory of the
transpose polynomial, for invertible potentials with no variable
of weight $1/2$.
For the quintic LG theory, Guo--Ross \cite{GuoRoss2019LG}
proved a genus-one mirror formula for the FJRW potential.

The Calabi-Yau/Landau-Ginzburg (CY/LG) correspondence is another relation that connects nonlinear sigma models on CY manifolds with LG models. Witten's gauged linear sigma model \cite{Witten1993Phases} explains
the CY and LG descriptions as different phases of a
common physical theory.
On the A-model side, Chiodo--Ruan proved the CY/LG state-space
correspondence \cite{ChiodoRuan2011State} and the genus-zero
correspondence for the quintic \cite{ChiodoRuan2010}.
At genus one, Guo--Ross \cite{GuoRoss2019Global} proved the
CY/LG correspondence for the quintic.
The mixed-spin-$P$ fields of Chang--Li--Li--Liu
\cite{ChangLiLiLiu2019MSP} provide a geometric framework relating
Gromov--Witten and FJRW theories. Subsequent developments include
the polynomial structure of higher-genus quintic Gromov--Witten
potentials proved by Chang--Guo--Li \cite{ChangGuoLi2021}. On the B-side, for the genus-zero theory, Tang--Yan \cite{TY}
established the CY/LG comparison of the $tt^*$ structures and the
Weil--Petersson metrics.
Fan--Lan--Yang \cite{FanLanYang2021,FanLanYang2022} obtained the same comparison
by a different approach. 

	\begin{center}
		\sffamily
		\footnotesize
		\begin{tikzpicture}[auto, block_center/.style ={rectangle, draw=black, thick, fill=white,
				text width=12em, text centered,
				minimum height=3em,inner sep=6pt,outer sep=4pt},
			block_left/.style ={rectangle, draw=black, thick, fill=white,
				text width=12em, text centered, minimum height=3em, inner sep=6pt,outer sep=4pt}]
			\matrix [column sep=25mm,row sep=12mm] {
				\node [block_left] (CYA) {Calabi-Yau A-model};
				& \node [block_center] (LGA) {Landau-Ginzburg A-model};\\
				\node [block_left] (CYB) {Calabi-Yau B-model};
				& \node [block_center] (LGB) {Landau-Ginzburg B-model};\\
			};
			\path[draw=black,latex-latex,thick] (CYA)  -- node[above,font=\tiny,outer sep=2pt]{CY/LG}(LGA);
			\path[draw=black,latex-latex,thick] (LGA) -- node[left, font=\tiny, align=left]{Mirror Symmetry}(LGB);
			\path[draw=black,latex-latex,thick] (CYB) -- node[below,font=\tiny,outer sep=2pt] {CY/LG} (LGB);
			\path[draw=black,latex-latex,thick] (CYA) -- node[right,font=\tiny,align=left] {Mirror Symmetry} (CYB);
		\end{tikzpicture}
	\end{center}

The diagram summarizes the two comparisons involved in mirror symmetry and the CY/LG correspondence. In this paper, we study the lower horizontal comparison at genus one.
\smallskip

Let $n\geq4$, and let
$f:\mathbb C^n\to\mathbb C$ be a homogeneous polynomial of degree $n$
with an isolated critical point at the origin. Then
\[
X_f=\{[z_1:\cdots:z_n]\in\mathbb P^{n-1}:f(z_1,\ldots,z_n)=0\}
\]
is a smooth Calabi--Yau hypersurface of complex dimension $n-2$. The genus-zero correspondence identifies the $tt^*$ geometric structures on the CY and LG sides. Our main result below shows that, at genus one, the comparison extends to the corresponding spectral invariants. 
For the associated Landau--Ginzburg orbifold $(\mathbb C^n/\mathbb Z_n,f)$, we define the genus one term $F_1^{\mathrm{LG}}(f)$ in B-model as some modification of analytic torsions associated with $f$-twisted Laplacian, see \Cref{subsec: setting cy lg}.

Let $F_1^{\mathrm{BCOV}}(X_f)$ denote the  genus-one BCOV invariant \cite{EFM,FLY} on $X_f$. Our main theorem is the following.
\paragraph{Main Theorem.}
    For every $n\geq4$, there is a constant $C_n$, depending only on $n$,
such that
\begin{equation}\label{eq:intro-correspondence}
F_1^{\mathrm{LG}}(f)=F_1^{\mathrm{BCOV}}(X_f)+C_n.
\end{equation}

Our main result also answers the conjecture stated in \cite[page~384]{FF}.

 \begin{remark}[Comparison with the Saito--Givental theory]

Givental's formalism packages genus-zero Gromov--Witten theory into a
Frobenius manifold and organizes higher-genus invariants, via Givental
quantization, into the expansion (cf.
\cite{Givental2001Quantization,Teleman2012})
\[
\log\mathcal D=\sum_{g\geq0}\hbar^{g-1}\mathcal F_g.
\]
Under a
generic semisimplicity condition, He--Li--Shen--Webb \cite{HeLiShenWebb2022} use this procedure to define
higher-genus LG B-model potentials from Saito's genus-zero theory
\cite{1981Primitive,SaitoFM} and prove an all-genus mirror theorem. Thus, unlike Gromov--Witten invariants,
their higher-genus B-model potentials are produced through Givental
quantization rather than through an independent intrinsic higher-genus
B-model construction.

The construction here is different: The genus-one LG B-model is defined intrinsically through the spectral invariants of the $f$-twisted Laplacian, without requiring semisimplicity or reconstructing the invariant from genus-zero data. 
\end{remark}

\begin{remark}
Beyond the comparison of two genus-one invariants, the result provides an analytic construction of the LG genus-one B-model. It may also offer a useful starting point for transferring computations between LG and CY geometries, particularly in situations where the LG description is more explicit.     
\end{remark}

\noindent The proof of our main theorem has two parts. 
\begin{itemize}
\setlength{\itemsep}{-1pt}
    \item We first establish a family anomaly formula for the modified LG torsion. Based
on our previous work on the genus-zero CY/LG correspondence for
$tt^*$-structures and Weil--Petersson metrics \cite{TY}, this formula agrees
with the BCOV anomaly formula. Hence the difference of the two genus-one invariants is pluriharmonic on the smooth parameter space.
    \item We then study this difference near the generic nodal locus of the discriminant (the boundary of the moduli). The required estimate for the LG torsion is obtained from a spectral-gap bound and a heat-kernel variation formula. Together with the known degeneration formula for the BCOV invariant, this gives the growth condition needed to extend the pluriharmonic difference across the discriminant. The topology of the parameter space then implies that the extended function is constant.
\end{itemize}

For an application of our main result, we note that certain structures are easier to be seen on the LG side. Our
main results allow us to transfer them directly to the CY side.
We provide only one simple example here, leaving more interesting
applications and further discussions to a separate paper.

\begin{corollary}\label{cor:bcov-split}
Let $r+s=n\geq 4$, and let $A_0,A_1$ on $\mathbb C^r$ and
$B_0,B_1$ on $\mathbb C^s$ be nondegenerate homogeneous polynomials
of degree $n$. Let
\[
X_{A,B}
:=\bigl\{[x:y]\in\mathbb P^{n-1}: A(x)+B(y)=0\bigr\}.
\]
We have\[
\begin{aligned}
&F_1^{\mathrm{BCOV}}(X_{A_1,B_1})
 -F_1^{\mathrm{BCOV}}(X_{A_0,B_1})=F_1^{\mathrm{BCOV}}(X_{A_1,B_0})
 -F_1^{\mathrm{BCOV}}(X_{A_0,B_0}).
\end{aligned}
\]
 In particular, if $r=s$, then
\[
\begin{aligned}
&F_1^{\mathrm{BCOV}}(X_{A,B})
 =\frac{1}{2}\Big(F_1^{\mathrm{BCOV}}(X_{A,A})
 +F_1^{\mathrm{BCOV}}(X_{B,B})\Big).
\end{aligned}
\]

\end{corollary}
The proof relies on the product structure of the $(A+B)$-twisted Laplacian.
We will provide the details after introducing some additional notation.

\subsection{The settings}
\label{sec:setup}
\subsubsection{CY/LG correspondence}\label{subsec: setting cy lg}

Let $f:\mathbb C^n\to\mathbb C$ be a homogeneous polynomial of degree $n$ with an isolated critical point at the origin. 
We study the genus-one B-model comparison for the pair $(\mathbb C^n/\mathbb Z_n,f)$ and the Calabi-Yau hypersurface 
\[
X_f=\{[z_1:\cdots:z_n]\in\mathbb P^{n-1}:f(z_1,\ldots,z_n)=0\}.
\]

In this paper, we identify $\Z_n$ with $\{\gamma\in\C:\gamma^n=1\}.$

For the analytic construction, rather than working directly on the quotient orbifold $\mathbb C^n/\mathbb Z_n$, we work on $\mathbb C^n$ and restrict  to $\mathbb Z_n$-invariant differential forms. Thus the relevant Hilbert space is
\[
\cE:=L^2\Omega^{*,*}(\mathbb C^n)^{\mathbb Z_n}.
\]
For a trace class operator $T$ on $\Z_2$-graded Hilbert space $\cH$, 
\[\Str_\cH(T)\]
denotes its supertrace.
We will use $$\Str$$ to denote $\Str_\cE$.

We consider Witten's $f$-twisted differential  and  Laplacian:
\begin{equation}\label{eq:differential}
{\bar{\partial}_{f}}=\bar\partial+df\wedge,
 \qquad
{\Delta_{f}}={\bar{\partial}_{f}}{\bar{\partial}_{f}}^*+{\bar{\partial}_{f}}^*{\bar{\partial}_{f}}\footnote{\text{We also have the twisted differential operator } $\partial_f:=\p+d\bar{f}\wedge$,\text{ and the K\"ahler identity holds}
\[\p_f\p_f^*+\p_f^*\p_f=\Delta_f=\bar\p_f\bar\p_f^*+\bar\p_f^*\bar\p_f.\]}.
\end{equation}
Here the adjoint $\bar\p_f^*$ is taken with respect to the Euclidean metric on $\C^n$. It is standard that (see \cite[Proposition~4.6]{zhang2001lectures} or
\cite[(10)]{FF})
\begin{equation}\label{eq:fixed}
\Delta_{f}
=
\Delta_{\bar\partial}
+
|\partial f|^2
+
L_{f},
\end{equation}
where
\begin{equation}\label{eq:selfadjointLF}
L_{f}
=
\sum_{a,c}
2\left(
\frac{\partial^2f}{\partial z_a\partial z_c}\,
dz_a\wedge\iota_{\partial/\partial\bar z_c}
+
\overline{
\frac{\partial^2f}{\partial z_a\partial z_c}
}\,
d\bar z_a\wedge\iota_{\partial/\partial z_c}
\right),
\end{equation}
and $\Delta_{\bar\partial}$ is the usual Dolbeault Hodge Laplacian.

 Let $P_f$ be the orthogonal projection
onto the space of $\Z_n$-invariant harmonic forms $$\cH_f:=\ker\Delta_{f}\subset\cE.$$ 
It carries a natural Hermitian metric \be\label{eq: harmonic metric} h_f\ee given by the $L^2$ inner products
 on differential forms induced by the Euclidean metrics.

Let $N^{1,0}$ and $N^{0,1}$ act on a $(p,q)$-form by multiplication
by $p$ and $q$, respectively. We define two modified number operators (see also \cite[\S3.4]{TY}) by
\begin{align}
 \ph&=N^{1,0}-\frac1nE,
&\qh&=N^{0,1}+\frac1nE.\label{eq:pq}
\end{align}
Here
\begin{equation}\label{eq:E}
E:=
\sum_{a=1}^n
\left(
z_a\partial_{z_a}
-\bar z_a\partial_{\bar z_a}
\right)
+N^{1,0}-N^{0,1}.
\end{equation}
Thus, for
$
\omega=\omega_{I,J}(z,\bar z)\,dz_I\wedge d\bar z_J,
$
one has
\[
E\omega
=
\Big[
\sum_{a=1}^n
\left(
z_a\partial_{z_a}
-\bar z_a\partial_{\bar z_a}
\right)\omega_{I,J}
+
(|I|-|J|)\omega_{I,J}
\Big]
dz_I\wedge d\bar z_J.
\]

The modified number operators satisfy\footnote{There are also commutative relations
\[[\qh,\p_f]=0,\quad [\qh,\p_f^*]=0,\]
which will be used in the proof of Lemma \ref{lem:N2reduction}.}
\begin{equation}\label{eq:pqcommutators}
\begin{aligned}
[\ph,\bar\partial_{f}]&=0,
&
[\ph,\bar\partial_{f}^{\,*}]&=0,\\
[\qh,\bar\partial_{f}]&=\bar\partial_{f},
&
[\qh,\bar\partial_{f}^{\,*}]&=-\bar\partial_{f}^{\,*}.
\end{aligned}
\end{equation}
Furthermore,
\begin{equation}\label{eq:pqDelta}
[\ph,\Delta_{f}]
=
[\qh,\Delta_{f}]
=
[\ph,P_f]
=
[\qh,P_f]
=0.
\end{equation}

Consider the graded heat trace
\begin{equation}\label{eq:reducedtrace}
\Str\!\left(\ph\qh(e^{-s{\Delta_{f}}}-P_f)\right)\footnote{By Lemma~\ref{lem:N2reduction}, we also have
\[
\Str\!\left(\ph\qh\big(e^{-s\Delta_f}-P_f\big)\right)
=
\frac{1}{2}\,
\Str\left(N^2\big(e^{-s\Delta_f}-P_f\big)\right),
\qquad\text{where}\quad N:=N^{1,0}+N^{0,1},
\]
which was studied as the second LG torsion in
\cite{FF,Tang2017Torsion}.},
 \qquad s>0.
\end{equation}
For $\Re w$ sufficiently large, define
\begin{equation}\label{eq:zeta}
\zeta(w)=\zeta(w;f):
=\frac1{\Gamma(w)}\int_0^\infty s^{w-1}\Str\!\left(\ph\qh(e^{-s{\Delta_{f}}}-P_f)\right)\,ds.
\end{equation}
Then $ \zeta(\cdot)$ extends to a meromorphic function on $\C$ and is regular at $w=0$
(Proposition~\ref{prop:familyheat}).

Our LG torsion is defined as
\begin{equation}\label{eq:torsion}
 \cT^{\LG}(f):=-\frac{\partial}{\partial w}
                         \zeta(w)\Big|_{w=0}.
\end{equation}
By \Cref{lem:N2reduction}, $\cT^{\LG}(f)$ is essentially the same as the
second torsion-type invariant introduced in \cite{FF}. As on the Calabi--Yau side, the BCOV torsion is not itself a B-model
invariant, since it depends on the choice of metric. The genus-one BCOV
invariant is obtained by modifying the BCOV torsion. Similarly, on the LG
side, the genus-one B-model invariant can be obtained by modifying the torsion
introduced above.

We next introduce the data needed to modify the LG torsion.

Let \(\Omega^k_{\mathrm{poly}}\) denote the space of polynomial holomorphic
\(k\)-forms. For any \(\alpha\in\Omega^n_{\mathrm{poly}}\), one has
$
\bar\p_f\alpha=0.
$
Although \(\alpha\in\Omega^n_{\mathrm{poly}}\) is not \(L^2\)-integrable,
there is an isomorphism
\be\label{eq: defn Sf}
S_f:
\left(
\Omega^n_{\mathrm{poly}}/
\big(df\wedge\Omega^{n-1}_{\mathrm{poly}}\big)
\right)^{\Z_n}
\longrightarrow
\cH_f
\ee
characterized by the following ``Hodge decomposition'': \(S_f(\alpha)\) is
the harmonic form for which
\[
S_f(\alpha)-\alpha
=
\bar{\partial}_f\nu
\]
is solvable for some
$
\nu\in\Omega^{n-1}(\C^n)
$
of polynomial growth.
See \cite[\S3.3, Proposition~3.5]{TY} for more details.

Set
\[w_0=w_0(f):=S_{f}(dz_1\wedge\cdots\wedge dz_n).
\]

Our modified LG torsion is defined by
\[
\cT_{\mathrm{mod}}^{\LG}(f)
:=
\cT^{\LG}(f)
{+}
c\log\big(h_f(w_0,w_0)\big),
\]
where the constant $c$ will be determined later.

For $A\in GL(n,\C)$, we have
\[
 \Phi_A:X_{f\circ A}\xrightarrow{\sim}X_f,\quad[z]\mapsto[Az],\qquad \Phi_A^*\mathcal O_{X_f}(1)=\mathcal O_{X_{f\circ A}}(1),
\]
with the polarization preserved. In particular, for $A=aI$, $a\in\C^*$, $X_{f\circ aI}=X_{a^nf}=X_f$. 
We will show that the transformation law of $\cT^\LG$ under $f\to f\circ A$ is compensated by that of the vacuum norm. This determines the coefficient $c$ in the modified torsion. More explicitly, Proposition \ref{prop:GLinvariance} gives the transformation laws, 
\[
 \begin{aligned}
 \cT^\LG(f\circ A)-\cT^\LG(f)
   &=\frac{\chi(X_f)}{12}\log|\det A|^2,\\
 \log h_{f\circ A}(w_0,w_0)-\log h_f(w_0,w_0)
   &=-\log|\det A|^2.
 \end{aligned}
\]
Thus the two terms cancel precisely when $c=\chi(X_f)/12$.

\vskip 0.1cm

Then we arrive at the following definition.
\begin{definition}
\label{def:modifiedearly}
The modified LG torsion and LG genus-one term for the LG model $(\C^n/\Z_n,f)$ is defined as
\begin{align}
 \cT_{\mathrm{mod}}^{\LG}(f)
   &:=\cT^{\LG}(f)+\frac{\chi(X_f)}{12}\log h_f(w_0,w_0),
       \label{eq:Tmodearly}\\
 F_{1}^{\LG}(f)
   &=-\frac12\cT_{\mathrm{mod}}^{\LG}(f).
       \label{eq:Fmodearly}
\end{align}

\end{definition}

Let
\[
F_{1}^{\mathrm{BCOV}}(X_f)=-\frac{1}{2}\log\tau^{\mathrm{BCOV}}(X_f)
\]
denote the BCOV invariant of $X_f$. See
\cite[Definition~5.7]{EFM} for the definition in arbitrary dimension and
\cite[Definition 4.13]{FLY} for the case of Calabi--Yau threefolds. For completeness, we also review its definition in
\Cref{app: defn BCOV}.
It is a B-model invariant of Calabi--Yau manifolds: it depends only on the complex structure and is independent of the choice of K\"ahler metric. 

\smallskip
Our main result is the following.
\begin{theorem}\label{thm: main}
For $n\ge4$, there exists a constant $C_n$, depending only on $n$, such that
\[
F_{1}^{\LG}(f)
=
F_{1}^{\mathrm{BCOV}}(X_f)
+
C_n.
\]
\end{theorem}

Now we can prove \Cref{cor:bcov-split}. 

\begin{proof}[Proof of \Cref{cor:bcov-split}]
One suffices to prove the statement on LG's side by our main result. The proof relies only on the product structure of the $(A+B)$-twisted Laplacian.

For $m=r,s,n$, let $\ph_m,\qh_m$ denote the modified degree operators
on forms on $\mathbb C^m$, defined by
\[
\ph_m=N^{1,0}-\frac1n E_m,
\qquad
\qh_m=N^{0,1}+\frac1n E_m,
\]
where
\[
E_m:=
\sum_{a=1}^m
\big(
z_a\partial_{z_a}
-\bar z_a\partial_{\bar z_a}
\big)
+N^{1,0}-N^{0,1}.
\]

For $\gamma\in\mathbb Z_n$, write
$\Str_\gamma(T):=\Str_{L^2\Omega^{*,*}(\C^m)}(\gamma^*T)$.

\def\f{{\mathfrak{f}}}

For a non-degenerate homogeneous polynomial $\f$ of degree $n$ on $\C^m$ $(m\leq n)$, set
\[
K_\f^\gamma(t)
=\Str_\gamma\Big(\ph_m\qh_m(e^{-t\Delta_\f}-\Pi_\f)\Big).
\]
Let $\Pi_\f$ denote the orthogonal projection onto the space of harmonic forms, not only onto the space of \(\Z_n\)-invariant harmonic forms.
We have
\be\label{eq: str trace for first two terms}
\Str_\gamma(e^{-t\Delta_\f})=I_m^\gamma,
\qquad
\Str_\gamma\Big(\widehat{p}_m(e^{-t\Delta_\f}-\Pi_\f)\Big)=0,
\ee
where
\[
I_m^\gamma=(-1)^m(\gamma+\cdots+\gamma^{n-1})^m.
\]
Here the first identity follows from the equivariant Euler formula.
For the second, note that
\[
\Str_\gamma\Big(\widehat{p}_m(e^{-t\Delta_\f}-\Pi_\f)\Big)
=
\Str_\gamma\Big(\widehat{q}_m(e^{-t\Delta_\f}-\Pi_\f)\Big)
=0.
\]
Indeed, for the vanishing of $\Str_\gamma\Big(\widehat{p}_m(e^{-t\Delta_\f}-\Pi_\f)\Big)$, note that $\bar\partial_\f+\bar\partial_\f^*$ commutes with $\widehat{p}_m$,
and for any eigenform $u$ of $\Delta_\f$ with eigenvalue $\lambda\neq0$,
the forms $u$ and
$\frac{1}{\sqrt{\lambda}}(\bar\partial_\f+\bar\partial_\f^*)u$
are paired eigenforms with the same eigenvalue $\lambda$ but
opposite $\Z_2$-gradings. The vanishing of $\Str_\gamma\Big(\widehat{q}_m(e^{-t\Delta_\f}-\Pi_\f)\Big)$ follows simiarly.

Hence, by \eqref{eq: str trace for first two terms},
\begin{equation}\label{eq:split-equivariant-heat}
K_{A+B}^\gamma(t)
=I_s^\gamma K_A^\gamma(t)+I_r^\gamma K_B^\gamma(t).
\end{equation}

Finally, since the operator is of product type, $w_0(A+B)=w_0(A)\wedge w_0(B)$. Thus we have
$\|w_0(A+B)\|^2=\|w_0(A)\|^2\|w_0(B)\|^2$.
The result then follows directly from \eqref{eq:split-equivariant-heat}.
\end{proof}

\subsubsection{The family setting and anomaly formula}
As we claimed before, the proof of main theorem depends on the family anomaly formula for the modified LG torsion. To do it, we consider the family case.

Let $B$ be a complex manifold and
\[
 F:\C^n\times B\longrightarrow\C,\quad n\ge4
\]
a holomorphic family of homogeneous polynomials $\{F_b\}_{b\in B}$ of degree $n$ with an isolated critical point at $0$. Associated with $F$, we obtain a holomorphic family of smooth Calabi-Yau manifolds
\[\pi:\mathcal{X}\longrightarrow B,\]
whose fiber over $b\in B$ is $X_b:=X_{F_b}$. We denote by \(\omega_{\mathrm{WP}}^{\mathrm{CY}}\in\Omega^{1,1}(B)\) the
Weil--Petersson form obtained by pulling back the Weil--Petersson form on the
moduli space of Calabi--Yau manifolds to \(B\).

\paragraph{The $tt^*$ structure on LG's side.} We first review the \(tt^*\)-structures on the LG side, which will appear in
the anomaly formula for the genus-one terms. The spaces $\cH_b:=\cH_{F_b}$ form a vector bundle $\cH\to B$.
It carries a Hermitian metric $h_F=(h_{F_b})_{b\in B}$. For local holomorphic coordinates $b^i$ on $B$, write
$\partial_i=\partial/\partial b^i$ and
$\partial_{\bar j}=\partial/\partial\bar b^j$, set
\begin{equation}\label{eq:higgs}
\phi_i=\partial_{b^i}F,
 \qquad C_i=P_bM_{\phi_i}P_b\big|_{\cH_b},
 \qquad \bar{C}_j=P_bM_{\overline{\phi}_j}P_b|_{\cH_b}
\end{equation}
and  
\[C=\sum_i C_i\,db^i,
 \qquad \bar{C}=\sum_j \bar{C}_j\,d\bar b^j.
\]
Here $M_\phi $ denotes multiplication by a scalar function $\phi$ and $P_b:=P_{F_b}$. $C$ and $\bar{C}$ are called the Higgs field in the LG B-model.

Let $d_B$ denote differentiation in the parameter directions. That is, for a smooth
 forms $u(b,z)$,
\[
d_Bu
=
\sum_i db^i\,\partial_{b^i}u
+
\sum_j d\bar b^j\,\partial_{\bar b^j}u.
\]
The orthogonal projections define a connection $\nabla^{\cH}$ on $\cH$ preserving $h_F$:
\begin{equation}\label{eq:harmonicconnection}
\nabla^\cH:=P\,d_B.
\end{equation}
That is, for $u(b)\in\cH_b$,
\[
\nabla^\cH_{\partial_{b^i}}u
=
P_b(\partial_{b^i}u),
\qquad
\nabla^\cH_{\partial_{\bar b^j}}u
=
P_b(\partial_{\bar b^j}u).
\]
We write its curvature and $(0,1)$-part as follows:
\[
R^\cH=(\nabla^\cH)^2,\qquad
\bar\partial^\cH:=(\nabla^\cH)^{0,1}.
\]
The data
\[
\bigl(
\mathcal H,\kappa^{\LG},h_F,\nabla^{\cH},C,\bar C
\bigr)
\]
satisfies the $tt^*$ structure equations, where $\kappa^\LG$
denotes the natural real structure on $\mathcal H$;
see \cite[Definition~1.1]{TY}.
For the purposes of this paper, we only need the following
consequences of the $tt^*$ structure equations:
\be\label{eq: tt lg}
(\bar\partial^{\mathcal H})^2=0,
\qquad
R^{\mathcal H}_{i\bar j}
=-[C_i,\bar C_j].
\ee
Thus $\mathcal{H}$ is a holomorphic bundle on $B$ and $\nabla^{\cH}$ is the Chern connection of $h_F$. This is the harmonic-bundle construction of
\cite[Theorem~4.14]{Fan}; see also
\cite[Proposition~3.3]{TY} in the present homogeneous setting.

\paragraph{The $tt^*$ structure on CY's side.}
Let
$\pi: \mathcal{X}\rightarrow B$ be a holomorphic family of Calabi-Yau manifolds of dimension $n-2$ parametrized by $B$.

Let $\mathcal{H}_{\prim}^{n-2}$ be the primitive Hodge bundle over $B$, for any $b\in B$, the fiber $\mathcal{H}_{\prim}^{n-2}\big|_{X_b}$ is the primitive $(n-2)$-form on $X_b$. Then we have the Hodge decomposition and Hodge filtration
\begin{align*}
&\mathcal{H}^{n-2}_{\prim}=\mathcal{H}^{n-2,0}_{\prim}\oplus \mathcal{H}^{n-3,1}_{\prim}\oplus\cdots\oplus \mathcal{H}^{0,n-2}_{\prim},\qquad \overline{\cH^{n-2-p,p}_{\prim}}=\cH^{p,n-2-p}_{\prim};\\
&\F^{n-2}\mathcal{H}^{n-2}_{\prim}\subset\cdots\subset \F^0\cH^{n-2}_{\prim}=\cH^{n-2}_{\prim},\qquad \F^k\cH^{n-2}_{\prim}=\oplus_{p=0}^{n-2-k}\cH_{\prim}^{n-2-p,p}.
\end{align*}
The theory of variation of Hodge structure tells us such a bundle is equipped with
\begin{itemize}
  \item a flat Gauss-Manin connection $\nabla^{\CY}$ which satisfies the Griffiths transversality
  $$\nabla^{\CY}_{v}\F^p\subset \F^{p-1}, \quad v\in {T}^{1,0}_B.$$
  \item a pairing $g: \cH^{n-2}_{\prim}\otimes\cH^{n-2}_{\prim}\rightarrow C^{\infty}(B)$
  $$g(\alpha,\beta)=i^{(n-2)^2+2p}\int_{X}\alpha\wedge\bar{\beta},\quad \alpha,\beta\in\cH^{n-2-p,p}_{\prim}.$$
  We call it the $tt^*$ metric in the CY B-model.
\end{itemize}

In particular, $\cH^{n-2,0}_{\prim}$ forms a line bundle over $B$. Let $\Omega_b$ be the holomorphic volume form on $X_{b}$, i.e. a local holomorphic section of $\cH^{n-2,0}_{\prim}$. The K\"{a}hler potential $K$ is defined via
$$e^{-K}=i^{-n-2}\int_{X}\Omega_b\wedge\overline{\Omega_b}.$$
Moreover, the Weil-Peterson form $\omega_{\WP}^{\CY}$ can be computed by
$$\omega_{\mathrm{WP}}^{\mathrm{CY}}=\partial_B\bar{\partial}_BK\in\Omega^{1,1}(B).$$

Now consider the $tt^*$ geometry structure induced by the variation of the Hodge structure on the CY B model. To begin, the $tt^*$ connection $\nabla^{\cH_\prim}$ is the Chern connection of the $tt^*$ metric $g$ on $\cH_{\prim}$, denote the $(1,0)$ and $(0,1)$ part of the $tt^*$ connection by $D^{\mathrm{CY}}$ and $\bar{D}^{\mathrm{CY}}$ respectively.

At the same time, the Kodaira-Spencer map
$$T^{1,0}B\longrightarrow \bigoplus_p\mathrm{Hom}(\cH^{n-2-p,p}_{\prim},\cH_{\prim}^{n-2-p-1,p+1})\subset\End(\cH^{n-2}_{\prim})$$
gives us a section of $\Omega^{1,0}(B,\End(\cH^{n-2}_{\prim}))$, which we denote by $C^{\mathrm{CY}}$. Its complex conjugate gives a section of $\Omega^{0,1}(\End(\cH^{n-2}_{\prim}))$, we denote by $\bar{C}^{\mathrm{CY}}$.

By Griffiths transversality and Hodge decomposition, we have
$$\nabla^{\CY,(1,0)}=D^{\mathrm{CY}}+C^{\mathrm{CY}},\quad \nabla^{\CY,(0,1)}=\bar{D}^{\mathrm{CY}}+\bar{C}^{\mathrm{CY}}.$$

The data 
\[
\bigl(
\mathcal H_{\prim},g,\kappa^\CY,\nabla^{\cH_\prim},C^\CY,\bar{C}^\CY
\bigr)
\] then
satisfies the $tt^*$ structure equation. Here $\kappa^\CY$ is the natural real structure on $\mathcal{H}_\prim$.

Now we recall the results in \cite{TY} that will be used throughout this paper.

\begin{proposition}[Genus-zero B-model CY/LG correspondence, \cite{TY}]
\label{prop:TY_summary}
For the family $F$ above, we have the following results.
\begin{enumerate}[(i)]
\item The $L^2$-Hodge decomposition holds:
\begin{align*}
L^2\Omega^{*,*}(\C^n)^{\Z_n}&=\mathcal H_b\oplus\Im\bar{\p}_{F_b}\oplus\Im\bar{\p}_{F_b}^*\\
\mathcal H_b
&\simeq\left(
\Omega_{\mathrm{poly}}^n(\mathbb C^n)
\big/
dF_b\wedge
\Omega_{\mathrm{poly}}^{n-1}(\mathbb C^n)
\right)^{\mathbb Z_n}.    
\end{align*}
The harmonic forms are concentrated in degree $n$, form a space of
constant dimension, and decay exponentially at infinity.
The isomorphism is realized by the canonical map $S_{F_b}$ in \eqref{eq: defn Sf}. That is, for every polynomial $n$-form $\alpha$,
\[
S_{F_b}(\alpha)-\alpha
=\bar\partial_{F_b}\nu
\]
for some form $\nu$ of polynomial growth. See \cite[Theorem~3.1, \S3.3, Propositions~3.5 and~3.7, Lemma~A.3]{TY}.
\item The section
\[w_0(F_b)=S_{F_b}(dz_1\wedge\cdots\wedge dz_n)\]
generates a holomorphic line bundle \be\label{eq: defn of L0}\mathcal{L}_0\to B\ee preserved by the induced Chern connection. Its Hermitian metric determines a $(1,1)$-form on $B$:
\[
\omega_{\mathrm{WP}}^{\mathrm{LG}}
=-\ddbar\log h_F(w_0,w_0)\in \Omega^{1,1}(B).
\]
Let $\Omega_{F_b}$ be the Poincar\'e form on $X_b:=\{F_b=0\}\subset\mathbb{P}^{n-1}$:
\[\Omega_{F_b}
=\operatorname{Res}_{X_b}\Big(
\frac{\sum_{a=1}^n(-1)^{a-1}z_a\,
dz_1\wedge\cdots\wedge\widehat{dz_a}\wedge\cdots\wedge dz_n}{F_b}
\Big).\]
Then 
\begin{equation}
\label{eq:LGresiduemetric}    
h_{F_b}(w_0,w_0)
=c_n
\int_{X_{F_b}}i^{(n-2)^2}
\Omega_{F_b}\wedge\overline{\Omega_{F_b}},\qquad c_n>0,
\end{equation}
where $c_n$ depends only on the dimension and normalization\footnote{We normalize the Poincar\'e residue by
$\operatorname{Res}((dF_b/F_b)\wedge\alpha)=\alpha|_{X_b}$.}. In particular, the $(1,1)$-form $\omega^\LG_\WP$ coincides with the
Weil--Petersson form on $B$:
\be\label{eq: weil peterson metric lgcy}
\omega_{\mathrm{WP}}^{\mathrm{LG}}
=\omega_{\mathrm{WP}}^{\mathrm{CY}}.
\ee
See \cite[\S3.5, Theorems~1.1 and~1.3]{TY}.
\item  The two $tt^*$ structures
\[
\bigl(
\mathcal H,\kappa^{\LG},h_F,\nabla^{\cH},C,\bar C
\bigr)
\quad\text{and}\quad
\bigl(
\mathcal H_{\prim},\kappa^\CY,g,
\nabla^{\cH_\prim},C^\CY,\bar C^\CY
\bigr)
\]
are isomorphic; see \cite[Theorem~1.3]{FanLanYang2022}
and \cite[Theorem~1.4]{TY}.
For the purposes of this paper, we only need the following
consequence of this isomorphism:
\be\label{eq: equal of trace C}
\Tr_{\mathcal H}(C\wedge\bar C)
=
\Tr_{\mathcal H_\prim}(C^\CY\wedge\bar C^\CY).
\ee
\end{enumerate}
    
\end{proposition}

\paragraph{Family of modified LG torsion and our main results}
One can define a family of LG torsion as
\begin{equation}
 \cT^{\LG}(b):= \cT^{\LG}(F_b).
\end{equation}

As before, we can define family version of genus-one term as
\begin{definition}
\label{def:modifiedearly family}
For the LG family $F$, define the modified LG torsion
$\cT_{\mathrm{mod}}^{\LG}(F)$ and the LG genus-one term
$F_{1}^{\LG}(F)$ to be the functions in $C^\infty(B)$ given by
\[
\cT_{\mathrm{mod}}^{\LG}(F)(b):=\cT_{\mathrm{mod}}^{\LG}(F_b),\qquad
F_{1}^{\LG}(F)(b):=F_{1}^{\LG}(F_b)=-\frac{1}{2}\cT_{\mathrm{mod}}^{\LG}(F_b).
\]
\end{definition}

We will prove the following genus-one LG holomorphic anomaly formula.

\begin{theorem}
\label{thm:modifiedanomalyearly}
We have
\begin{equation}\label{eq:main}
\ddbar F_{1}^\LG
=
\frac{(-1)^{n+1}}2\Tr_\cH(C\wedge \bar C)
+\frac{\chi(X_F)}{24}\omega_{\mathrm{WP}}^{\mathrm{LG}}.
\end{equation}
Equivalently,
\begin{equation}\label{eq:modifiedmainT}
\ddbar\cT^\LG_{\mathrm{mod}}
=(-1)^n\Tr_\cH(C\wedge \bar C)
-\frac{\chi(X_F)}{12}\omega_{\mathrm{WP}}^{\mathrm{LG}}.
\end{equation}
\end{theorem}

Similarly, we have BCOV invariant $F_{1}^{\mathrm{BCOV}}(\mathcal{X})\in C^\infty(B)$ associated with the family  $\mathcal{X}\to B$,
which is given by
\[F_{1}^{\mathrm{BCOV}}(\mathcal{X})(b):=F_{1}^{\mathrm{BCOV}}(X_b).\]

Comparing the anomaly formulas on two sides gives the following.

\begin{corollary}\label{cor: harmonic}
On $B$,
\[
\ddbar\Big(
F_{1}^\LG(F)
-F_{1}^{\mathrm{BCOV}}(\mathcal{X})
\Big)=0.
\]
\end{corollary}

\begin{proof}
For $0\le p,q\le n-2$, let $\mathcal H^{p,q}\to B$ be the
holomorphic Hodge bundle whose fiber at $b\in B$ is
$H^q(X_b,\Omega^p_{X_b})$, where $\Omega^p_{X_b}$ is the
bundle of holomorphic $p$-forms on $X_b$. Equip
$\mathcal H^{p,q}$ with the Hodge metric associated with the
hyperplane polarization, and denote its Chern connection
by $\nabla^{p,q}$.

For $0\le k\le 2n-4$, define the weighted Hodge curvature form
\[
\omega_{H^k}
:=
\sum_{\substack{p+q=k\\0\le p,q\le n-2}}
p\,\Tr_{\mathcal H^{p,q}}
\bigl((\nabla^{p,q})^2\bigr)
\in\Omega^{1,1}(B).
\]

The anomaly formula for BCOV invariants, see \cite[{\S}5.4]{EFM}, becomes
\[
\partial_B\bar\partial_B F_1^{\mathrm{BCOV}}
=
-\frac12\sum_{k=0}^{2n-4}(-1)^k\omega_{H^k}
+\frac{\chi(X_F)}{24}\omega_{\mathrm{WP}}^\CY.
\]

By the weak Lefschetz theorem and Poincar\'e duality,
all cohomology of $X_b$ outside the primitive middle-degree
part is generated by powers of the hyperplane class.
These classes are flat for the Gauss--Manin connection and
have constant Hodge norms with respect to the hyperplane
polarization. Their Hodge Chern curvatures therefore vanish.
Consequently,
\[
\sum_{k=0}^{2n-4}(-1)^k\omega_{H^k}
=
(-1)^{n-2}\omega_{H^{n-2}},
\]
and
\[
\omega_{H^{n-2}}
=
\sum_{p=0}^{n-2}(n-2-p)
\Tr_{\mathcal H_{\prim}^{n-2-p,p}}
\bigl((\nabla^{\mathcal H_{\prim}})^2\bigr).
\]

Note that the $tt^*$ equation in CY's setting gives
\[
\bigl((\nabla^{\mathcal H_{\prim}})^2\bigr)_{i\bar j}
=
-[C_i^\CY,\bar C_j^\CY].
\]
Then
\begin{align*}
(\omega_{H^{n-2}})_{i\bar j}
={}&
-\sum_{p=0}^{n-2}(n-2-p)
\Tr_{\mathcal H_{\prim}^{n-2-p,p}}
\bigl(C_i^\CY\bar C_j^\CY\bigr)
+
\sum_{p=0}^{n-2}(n-2-p)
\Tr_{\mathcal H_{\prim}^{n-2-p,p}}
\bigl(\bar C_j^\CY C_i^\CY\bigr)\\
={}&
-\sum_{p=1}^{n-2}(n-2-p)
\Tr_{\mathcal H_{\prim}^{n-2-p,p}}
\bigl(C_i^\CY\bar C_j^\CY\bigr)
+
\sum_{p=0}^{n-3}(n-2-p)
\Tr_{\mathcal H_{\prim}^{n-3-p,p+1}}
\bigl(C_i^\CY\bar C_j^\CY\bigr)\\
={}&
\sum_{p=0}^{n-2}
\Tr_{\mathcal H_{\prim}^{n-2-p,p}}
\bigl(C_i^\CY\bar C_j^\CY\bigr)
=
\Tr_{\mathcal H_{\prim}}
\bigl(C_i^\CY\bar C_j^\CY\bigr).
\end{align*}

Substitution into the anomaly formula above gives
\[
\partial_B\bar\partial_B F_1^{\mathrm{BCOV}}
=
\frac{(-1)^{n+1}}{2}
\Tr_{\mathcal H_{\prim}}
\bigl(C^\CY\wedge\bar C^\CY\bigr)
+\frac{\chi(X_F)}{24}\omega_{\mathrm{WP}}^\CY.
\]
By \Cref{thm:modifiedanomalyearly} and \Cref{prop:TY_summary}(ii)--(iii),
$F_1^\LG$ satisfies the same anomaly formula. Hence
\[
\partial_B\bar\partial_B
\bigl(F_1^\LG-F_1^{\mathrm{BCOV}}\bigr)=0.
\]
\end{proof}

\subsection{Conventions}

Throughout this paper, for $z=(z_a)_{a=1}^n\in\C^n$, we consider standard metric $g_0$ on $\C^n$
\[
 g_0=\sum_{a=1}^n\bigl((dx_a)^2+(dy_a)^2\bigr),\quad\text{where} \quad  z_a=x_a+iy_a.
\]
The volume measure is $ d^{2n}z:=\prod_{a=1}^n dx_a\,dy_a.$
All integral kernels are taken with respect to this measure. We distinguish the
coordinate column vector $\partial_zF$ from the differential one-form $\partial F$:
\[
 \partial_zF:=\bigl(\partial_{z_1}F,\ldots,\partial_{z_n}F\bigr)^T,
 \qquad
 \partial F:=\sum_a\partial_{z_a}F\,dz_a.
\]
Since $|dz_a|_{g_0}^2=|d\bar z_a|_{g_0}^2=2$,
\[
 |\partial_zF|^2:=\overline{\partial_zF}^T\partial_zF=\sum_a|\partial_{z_a}F|^2,
 \qquad |\partial F|^2:=|\partial F|_{g_0}^2
       =2|\partial_zF|^2.
\]
The Hessian $\partial_z^2F$ is the matrix
$\bigl(\partial_{z_a}\partial_{z_c}F\bigr)_{a,c}$.
The conjugate transpose of a matrix $A$ is always written $\bar A^T$.
The superscript $*$ is reserved for the metric adjoint of an operator.
In particular,
\[
 (dz_a\wedge)^*=2\iota_{\partial/\partial z_a},\qquad
 (d\bar z_a\wedge)^*=2\iota_{\partial/\partial\bar z_a},\qquad
 \bar\partial^*=-2\sum_a\iota_{\partial/\partial\bar z_a}\partial_{z_a}.
\]
Then
\[
 \Delta_{\bar\partial}
 =-2\sum_a\partial_{z_a}\partial_{\bar z_a}
 =-\frac12\sum_{a=1}^n(\partial_{x_a}^2+\partial_{y_a}^2).
\]

For complex-valued functions $f,g$, the notation
$\nabla f\cdot\nabla g:=\sum_{\alpha=1}^{2n}
(\partial_\alpha f)(\partial_\alpha g)$ is complex bilinear,
not Hermitian. The real gradient and all distances refer to $g_0$.
Thus,
\[
 \Delta_{\bar\partial}(fg)
 =f\Delta_{\bar\partial}g+g\Delta_{\bar\partial}f
       -\nabla f\cdot\nabla g.
\]

\subsection{Organization of the paper}

The paper is organized as follows.
In \Cref{sec:inputs}, we develop the family heat calculus,
perform the small-time analysis, and obtain the holomorphic
anomaly formula, i.e. Theorem~\ref{thm:modifiedanomalyearly}. The proof in \Cref{sec:inputs} relies on the Brownian bridge technique.
An alternative proof that does not use Brownian bridges is provided in
\Cref{app:getzler-small-time}.
In \Cref{sec:modified}, we study the modified LG torsion
and prove its invariance under scalar multiplication and linear
changes of coordinates. 
Finally, in \Cref{sec: proof main}, we give the necessary
boundary estimates and complete the proof of the genus-one
CY/LG correspondence in Theorem~\ref{thm: main}.

{\paragraph{Acknowledgments.}
The authors thank Huijun Fan for raising this question and for encouraging them
to pursue this problem. J.~Yan also thanks Maxim Braverman, Xianzhe Dai, Minghao Wang and Ken-ichi Yoshikawa
for helpful discussions on related questions. J.~Yan is supported by the
Zelevinsky Postdoctoral Fellowship.
X. Tang is supported by 
the BNSFC Youth Project (Grant No.~1254041), the NSFC
Youth Project (Grant No.~12501079) and BNSFC (Grant No.~JR25001).}

\section{Family heat calculus and Brownian bridges}
\label{sec:inputs}

\subsection{Family heat calculus}
For $u\in\C^*$, we consider the following Bismut-type superconnection
\begin{equation}\label{eq:superconnection}
 \mathbb B_u=d_B+u\bar\partial_{F_b}+\bar u\bar\partial_{F_b}^*.
\end{equation}
Here $d_B$ is the de Rham differential on $B.$
Then we compute:
\[\mathbb B_u^2=|u|^2\Delta_{F_b}+u\sum_idb^i\partial\phi_i\wedge+\bar{u}\sum_jd\bar{b}^j(\partial\phi_j\wedge)^*.\]

Now the situation is that $\C^n$ is non-compact and the potential part of the twisted Laplacian is polynomially growing. 
The following proposition establishes, in our setting, the validity of
Duhamel's principle and the cyclicity of the supertrace (supertrace vanishes on supercommutator), as well as the
small-time heat kernel expansion and the large-time heat kernel estimates.

\begin{proposition}\label{prop:familyheat}
Let $K\Subset B$. We have the following results uniformly on $K$.
\begin{enumerate}[(i)]
\item\label{item:calculus}
Let
$
L\in C^\infty\!\left(
\C^n,\End(\Lambda^{*,*}\C^n)
\right),
$
whose pointwise norm has polynomial growth of $O(|z|^{2n-2-\delta})$ for some $\delta>0$. Then the Duhamel's principal holds:
\be\label{eq: Duhamel}
\left.
\frac{d}{d\varepsilon}
\Str\Big(e^{-s\Delta_{F_b}+\varepsilon L}\Big)\right|_{\varepsilon=0}
=\int_0^1
\Str\Big(e^{-(1-t)s\Delta_{F_b}}\,
L\,
e^{-ts\Delta_{F_b}}\Big)\,dt=\Str\!\left(L e^{-s\Delta_{F_b}}\right).
\ee
And the operators above as well as $L e^{-s\Delta_{F_b}}$ and $ e^{-s\Delta_{F_b}}L$ are trace class, locally uniformly in $b$, and depend smoothly on
$b$. 

Moreover, for any $
L\in C^\infty\!\left(
\C^n,\End(\Lambda^{*,*}\C^n)
\right),
$
whose pointwise norm has polynomial growth,
\[
\Str\!\left(L e^{-s\Delta_{F_b}}\right)
=
\Str\!\left(e^{-s\Delta_{F_b}}L\right).
\]
Actually, the corresponding cyclicity identities hold for the supertraces of
products of $e^{-s\Delta_{F_b}}$ and polynomial-growth endomorphisms in $ C^\infty\!\left(
\C^n,\End(\Lambda^{*,*}\C^n)
\right)$.

\item\label{item:expansion}
There are log-free expansions in $C^\infty(K)$:
\[
 \begin{aligned}
 \Str\!\left(\ph\qh(e^{-s{\Delta_{F_b}}}-P_b)\right)
     &\sim\sum_{\lambda\in\Lambda}a_\lambda s^\lambda,\\
\Str(\ph e^{-\mathbb B_{\sqrt s}^{\,2}})^{(1,1)}
     &\sim\sum_{\lambda\in\Lambda}A_\lambda s^\lambda,
 \qquad \#\bigl(\Lambda\cap(-\infty,M]\bigr)<\infty\quad(\forall M\in\R).
 \end{aligned}
\]
Here $\Lambda\subset\R$ is the set of exponents, $a_\lambda$ are
functions on $B$ and $A_\lambda$ are $(1,1)$-forms on $B$; missing
coefficients are set to zero. The expansions are stable under
$s\partial_s$, with arbitrarily high remainder order.

\item\label{item:infinity}
The projection $P_b$ is smooth in $b$. Moreover, for $s\geq1$
\begin{equation}\label{eq:infinityinput}
 \begin{aligned}
 \Str\!\left(\ph\qh(e^{-s{\Delta_{F_b}}}-P_b)\right)
     &=O_{C^\infty(K)}(e^{-c_Ks}),\\
 \lim_{s\to\infty}
      \Str(\ph e^{-\mathbb B_{\sqrt s}^{\,2}})^{(1,1)}
     &=-\Str_{\cH}\!\left(\ph R^{\cH}\right).
 \end{aligned}
\end{equation}

\end{enumerate}
Here $w^{(1,1)}$ denotes the $(1,1)$-component of a differential form $w$.
\end{proposition}

\begin{proof}
Fix $K\Subset B$. Homogeneity and the isolated-critical-point condition imply
\be\label{eq: weight}
|\partial F_b(z)|\geq c_K|z|^{n-1},
\qquad b\in K
\ee
for some constant $c_K>0.$

For a trace-class operator $H$, consider the trace norm
\[
\|H\|_{\mathrm{tr}}
=\Tr\bigl((H^*H)^{1/2}\bigr)
=\sum_{j=1}^{\infty}s_j(H),
\]
Here $s_j(H)$ are the singular
values of $H$, counted with multiplicity.

It follows from \cite[Proposition~2.4]{DY} that
$\|e^{-s\Delta_{F_b}}\|_{\mathrm{tr}}<\infty$. It is straightforward to check that there exists a constant $l_K>0$, such that
\[\Delta_{F_b}+l_K\geq \Delta_{\bar\p}+c_K|z|^{2n-2}+1.\]
This implies if $|L|=O(|z|^k)$ for some integral $k$ (where $|L|$ denotes the pointwise operator norm),
\[\|L e^{-s\Delta_{F_b}}\|_{\mathrm{tr}}\leq Cs^{-\frac{k}{2n-2}}\| e^{-s\Delta_{F_b}/2}\|_{\mathrm{tr}}.\]
 Similarly, one can show if $L=O(|z|^{2n-2-\delta})$, then $e^{-s\Delta_{F_b}+\varepsilon L}$ and $\int_0^1
e^{-(1-t)s\Delta_{F_b}}\,
L\,
e^{-ts\Delta_{F_b}}\,dt$ are of trace class. Thus, \eqref{eq: Duhamel} follows from the Duhamel's principal.

Furthermore, we have
\begin{align*}
    \Str\!\left(L e^{-s\Delta_{F_b}}\right)=\,&\Str\!\left(\big(L e^{-s\Delta_{F_b}/2}\big) e^{-s\Delta_{F_b}/2}\right)\\
    =\,&\Str\!\left(e^{-s\Delta_{F_b}/2}\big(L e^{-s\Delta_{F_b}/2}\big) \right)=\Str\!\left(\big(e^{-s\Delta_{F_b}/2}L\big) e^{-s\Delta_{F_b}/2}\right)\\
    =\,&\Str\!\left(e^{-s\Delta_{F_b}/2} e^{-s\Delta_{F_b}/2} L\right)=\Str\!\left(e^{-s\Delta_{F_b}}L\right).
\end{align*}
This proves~(i).

The statement \textup{(ii)} follows from the same argument of \cite[Theorem~4.1 and \S6, especially Theorem~6.2]{DY}.

\vskip 0.1cm
For \textup{(iii)},  constancy of the kernel dimension (see Proposition~\ref{prop:TY_summary} (i)) give a uniform gap of nonzero spectrum of $\Delta_{F_b}$ from zero on $K$.  The spectral gap also gives the  large-time decay; the same argument as in \cite[Theorem~3.4]{BK}
then gives the  limit in \eqref{eq:infinityinput}.
\end{proof}

\begin{proof}[Proof of Theorem \ref{thm:modifiedanomalyearly}]
The relations $[\ph,\bar\partial_{F_b}]=0$,
$[\qh,\bar\partial_{F_b}]=\bar\partial_{F_b}$ and their adjoints,
together with  the fact that supertrace vanishes on supercommutator and Duhamel's principal in Proposition \ref{prop:familyheat}, give 
\begin{equation}\label{eq:transgression}
\ddbar\Str\!\left(\ph\qh(e^{-s\Delta_{F_b}}-P_b)\right)
=
s\frac{d}{ds}
\Str(\ph e^{-\mathbb B_{\sqrt s}^{\,2}})^{(1,1)}.
\end{equation}
This is the same superconnection calculation as in
\cite[Theorem~2.9]{BGS2}.

By Proposition~\ref{prop:familyheat} again, a standard computation (see
\cite[\S 3]{BK}) shows that
\begin{equation}\label{eq:boundaryformula}
\ddbar\cT^\LG
=-\FP_{s\to0}
\bigl[\Str(\ph e^{-\mathbb B_{\sqrt s}^{\,2}})\bigr]^{(1,1)}
+\Str_\cH(\ph R^{\cH,(1,1)}).
\end{equation}
Here $\FP_{s\to0}$ denotes the constant term, i.e., the coefficient of $s^0$ in an asymptotic expansion.

We will show in 
Proposition~\ref{prop:smalltime} that 
\[\FP_{s \to0}\Str(\ph e^{-\mathbb B_{\sqrt s}^{\,2}})^{(1,1)}=0.\]
The $tt^*$ equation, see \eqref{eq: tt lg}, gives
$R^\cH_{i\bar j}=-[C_i,\bar{C}_j]$.
Homogeneity  and $[\widehat{p},P_F]=0$ give $[\ph,C_i]=-C_i$; by the fact that supertrace vanishes on supercommutator,
$\Str_\cH([C_i,\ph \bar{C}_j])=0$. Thus,
as $\cH^k=0$ for $k\ne n$,
\begin{equation}\label{eq:rawanomaly}
\partial_i\partial_{\bar j}\cT^\LG
=\Str_\cH(\ph R^\cH_{i\bar j})
=-\Str_\cH([\ph,C_i]\bar{C}_j)
=(-1)^n\Tr_\cH(C_i\bar{C}_j).
\end{equation}
Finally, Proposition~\ref{prop:TY_summary} (ii) gives
$\ddbar\log h_F(w_0,w_0)=-\omega_{\mathrm{WP}}^\LG$;
thus Definition~\ref{def:modifiedearly} yields
\[
\ddbar F_{1}^\LG
=
\frac{(-1)^{n+1}}2\Tr_\cH(C\wedge \bar{C})
+\frac{\chi(X_F)}{24}\omega_{\mathrm{WP}}^\LG.
\qedhere
\]
\end{proof}

 To evaluate
\[
\FP_{s\to0}
\bigl[\Str(\ph e^{-\mathbb B_{\sqrt{s}}^{\,2}})\bigr]^{(1,1)},
\]
we consider, for $h>0$,  
\begin{align}
D_h&=|\partial F_b|^2+h^2\Delta_{\bar\partial}+hL_{F_b}
 +\sum_i db^i\,\partial\phi_i\wedge
 +\sum_jd\bar b^j\,(\partial\phi_j\wedge)^*.
\end{align}

For
\[
\omega=\sum_{I,J}\omega_{I,J}(z)\,dz_I\wedge d\bar z_J,
\]
 define the operator $U_s$:
\be\label{eq: defn Us}
U_s\omega
=
s^{-\frac{n}{2(n-1)}}
\sum_{I,J}
\omega_{I,J}\!\left(s^{-\frac1{2(n-1)}}z\right)
dz_I\wedge d\bar z_J.
\ee
Then $U_s$ is unitary and $[U_s,\ph\,]=0$. If $h=s^{\frac{n}{2(n-1)}}$, we have,
\[U_s(\sqrt{s}\bar\p_{F_b})U_s^{-1}=h\bar\p+\p F_b\wedge,\qquad U_s(\sqrt{s}\bar\p_{F_b}^*)U_s^{-1}=h\bar\p^*+(\p F_b\wedge)^*.\]
Furthermore
\begin{equation}\label{eq:exactscale1}
U_s\mathbb B_{\sqrt{s}}^{2}U_s^{-1}=D_h.
\end{equation}

Thus, we have:
\begin{proposition}\label{prop:exactscale}
\[
\FP_{s \to0}
\Str(\ph e^{-\mathbb B_{\sqrt{s}}^{2}})^{(1,1)}
=
\FP_{h \to0}
\Str(\ph e^{-D_h})^{(1,1)}.
\]
\end{proposition}

Now, to compute
\[
\FP_{s \to0}
\Str(\ph e^{-\mathbb B_{\sqrt{s}}^{2}})^{(1,1)},
\]
it suffices to compute
\[
\FP_{h \to0}
\Str(\ph e^{-D_h})^{(1,1)}.
\]
We will carry out this computation using the Brownian bridge technique in
Proposition~\ref{prop:smalltime}. An alternative proof without using Brownian bridge is provided in
\Cref{app:getzler-small-time}.

\subsection{Brownian bridges }
\label{sec:matrixremainder}
In this subsection, we introduce the Brownian-bridge representation of heat kernel of $D_h$ (see Lemma \ref{lem:matrixremainder}). The resulting formula will be used to compute the constant term in the small $h$ expansion.

\vskip 0.1cm
We use the standard Brownian bridge $\mathcal B_u$ in $\R^{2n}$:
\begin{equation}\label{eq:bridgecovariance}
 \mathcal{B}_u=W_u-uW_1,\qquad
 \mathbb E \mathcal{B}_u^\alpha=0,\qquad
 \mathbb E(\mathcal{B}_u^\alpha \mathcal{B}_v^\beta)
   =\delta_{\alpha\beta}\bigl(\min(u,v)-uv\bigr),\qquad 0\le u,v\le1.
\end{equation}
Here $W_u$ is standard real Brownian motion, with covariance
$\mathbb E(W_u^\alpha W_v^\beta)=\delta_{\alpha\beta}\min(u,v)$;
$\mathcal{B}_u$ denotes the bridge.
 $\alpha,\beta$ are indices range from $1$ to $2n$.
 
Consider the path
\begin{equation}\label{eq:bridgepath}
 \xi_u(x,y)=(1-u)x+uy+ h\mathcal{B}_u.
\end{equation}

For a continuous matrix-valued function $M$, our ordering convention is
\begin{equation}\label{eq:pathordering}
 \begin{aligned}
 &\mathcal P\exp\left(-\int_0^1M(\xi_u)\,du\right)=\Id+\sum_{q=1}^\infty(-1)^q
   \int_{0<u_1<\cdots<u_q<1}
       M(\xi_{u_1})\cdots M(\xi_{u_q})\,du_1\cdots du_q.
 \end{aligned}
\end{equation}

Then we have:
\begin{lemma}\label{lem:matrixremainder}
Let $K\Subset B$. For every pair of integers $L,k\ge0$, the following
statements hold.

\begin{enumerate}[(i)]
\item
The $\Omega^\bullet(B)$-valued kernel of $e^{-D_h}$ is
\begin{equation}\label{eq:bridgekernel}
 K_{D_h}(1;x,y)
 =
 (2\pi h^2)^{-n}
 \exp\!\left(-\frac{|x-y|^2}{2h^2}\right)
 \mathcal K_h(x,y),
\end{equation}
where
\begin{equation}\label{eq:bridgeamplitude}
\begin{aligned}
 \mathcal K_h(x,y)
 =\mathbb E\bigg[
 &\exp\!\Big(-\int_0^1 |\partial F_b\bigl(\xi_u(x,y)\bigr)|^2\,du
 \Big)
\\[-0.2em]
 &~\cdot\mathcal P\exp\!\Big(
-\int_0^1\bigl(hL_{F_b}+\sum_i db^i\,\partial\phi_i\wedge
 +\sum_j d\bar b^j\,(\partial\phi_j\wedge)^*\bigr)
   \bigl(\xi_u(x,y)\bigr)\,du
 \Big)
 \bigg].
\end{aligned}
\end{equation}
In each fixed degree of $\Omega^\bullet(B)$, the amplitude $\mathcal K_h(x,y)$ is smooth for $0\le h\le1$.
\item
The amplitude has the Taylor expansion
\begin{equation*}\label{eq:amplitudeTaylor}
\mathcal K_h(z,w)=
\sum_{j=0}^{L}h^j\,\mathcal K_{j,b}(z,w)+
R_{L,b}(h;z,w),\qquad\mathcal K_{j,b}
=\frac1{j!}\left.\partial_h^j\mathcal K_h\right|_{h=0}.
\end{equation*}
For every $0\le j\le L$,
\begin{equation}\label{eq:coefficientintegrability}
 \sup_{\substack{b\in K\\|\beta|\le k}}
 \int_{\C^n}
 \Bigg(
\bigl\|\partial_b^\beta\mathcal K_{j,b}(z,z)\bigr\|+\langle z\rangle\left\|
\left.\partial_b^\beta\nabla_1\mathcal K_{j,b}(z,w)\right|_{w=z}\right\|
\Bigg)\,d^{2n}z
 <\infty.
\end{equation}
\item
The Taylor remainder satisfies
\begin{equation}\label{eq:familyremainder}
\begin{aligned}
 \max_{0\le\ell\le1}\;
 \sup_{\substack{b\in K\\|\beta|\le k}}
 \int_{\C^n}
 \Bigg(
 &
 \left\|
 \partial_b^\beta(h\partial_h)^\ell
 R_{L,b}(h;z,z)
 \right\|
\\
 &\quad+
 \langle z\rangle
 \left\|
 \left.
 \partial_b^\beta\nabla_1(h\partial_h)^\ell
 R_{L,b}(h;z,w)
 \right|_{w=z}
 \right\|
 \Bigg)\,d^{2n}z
 \le C_{K,L,k}h^{L+1}.
\end{aligned}
\end{equation}
\end{enumerate}
Here $\langle z\rangle=(1+|z|^2)^{1/2}$, $\beta$ is a multiindex in local coordinates on $B$, and $\nabla_1$ denotes differentiation
with respect to the first spatial variable. The constant $C_{K,L,k}$ is independent of $b$ and $0<h\le1$. Furthermore, the same estimate holds
after applying $h\partial_h$.
\end{lemma}

\begin{proof}
Fix $K\Subset B$. 
\textup{(i)} The Brownian-bridge formula for $D_h$ follows from \cite[Lemma~1.12 and Theorem~1.14]{GueneysuMatrix}. 

For \textup{(ii)} and \textup{(iii)}, the same arguments as in \cite[Theorem~4.1 and \S6]{DY} imply that integrands in (ii) and (iii) are bounded above by
\[(1+|z|)^Me^{-c|z|^{2n-2}}\]
for some $M,c>0$.
\end{proof}

\subsection{Small-time calculation}
\label{sec:smalltime}

 We suppress $b$ in $F$ and $L_F$ when no confusion can arise.
\begin{proposition}\label{prop:smalltime}
For the family in \S\ref{sec:setup}, one has the following asymptotic expansion:
\begin{equation}\label{eq:thetaexp}
 \Str(\ph e^{-D_h}) ^{(1,1)}
=
h^{-2}\mathcal A_{-1}(b)
+O_{C^\infty}(h^2).
\end{equation}
In particular,
\[
\FP_{h \to0}
\Str(\ph e^{-D_h})^{(1,1)}=0.
\]
\end{proposition}

\begin{proof}[Outline of the proof]
Since $D_h$ and $\ph$ commute with the $\mathbb Z_n$-action,
\begin{equation}\label{eq:invariant-trace-decomposition}
\bigl[\Str_{\cE}(\ph e^{-D_h})\bigr]^{(1,1)}
=
\frac{1}{n}\sum_{g\in\mathbb Z_n}
\bigl[
\Str_{L^2\Omega^{*,*}(\mathbb C^n)}
(g^*\ph e^{-D_h})
\bigr]^{(1,1)}.
\end{equation}
We first consider the identity sector $g=1$. Decompose
\[
\ph=\ph_{\mathrm{num}}+\ph_{\mathrm{diff}},
\]
where
\[
\ph_{\mathrm{num}}
=
\frac{n-1}{n}N^{1,0}+\frac{1}{n}N^{0,1},
\qquad
\ph_{\mathrm{diff}}
=
-\frac{1}{n}\sum_{a=1}^n
\left(
z_a\partial_{z_a}
-\bar z_a\partial_{\bar z_a}
\right).
\]

For the $db^i\wedge d\bar{b}^j$-coefficient of the
$\ph_{\mathrm{num}}$-term, a nonzero pointwise supertrace
containing $k$ factors of $L_F$ requires
\[
k\geq 2n-2,
\qquad
k\equiv 0\pmod 2.
\]
If the Taylor expansion in $z+h\mathcal B_u$ has total degree
$\ell$, the
corresponding term has order
\[
h^{k+\ell-2n}.
\]
The symmetry $\mathcal B\mapsto-\mathcal B$ eliminates odd
$\ell$. Hence the only contributions through order $h^0$ are
\[
(k,\ell)
=(2n-2,0),\qquad
(2n-2,2),\qquad
(2n,0).
\]

Step 3 evaluates the first case $(k,\ell)=(2n-2,0)$. After including the factor $1/n$ in \eqref{eq:invariant-trace-decomposition}, it gives the coefficient of $h^{-2}$:
\[
(\mathcal A_{-1})_{i\bar j}
=
\frac{1}{n(2\pi)^n}
\int_{\mathbb C^n}
e^{-|\partial F|^2}
Q_{i\bar j}\,d^{2n}z,
\]
where
\begin{equation}\label{eq:Q}
Q_{i\bar j}
=\frac{(-1)^n2^{2n-1}}{n}
\overline{\partial_z\phi_j}^T
\operatorname{adj}((\partial_z^2F)(\overline{\p_z^2F})^T)
\partial_z\phi_i.
\end{equation}
Step 3 also together with Step 4-5 shows that the sum of the two constant-order contributions,
\[
(k,\ell)=(2n-2,2)
\qquad\text{and}\qquad
(k,\ell)=(2n,0),
\]
is the integral of
\[
-\frac{1}{12}
\Delta_{\bar\partial}
\left(
e^{-|\partial F|^2}Q_{i\bar j}
\right).
\]
Step~6 shows that $\ph_{\mathrm{diff}}$ contributes neither an
$h^{-2}$-term nor a constant term. Therefore the constant
coefficient of the identity sector is
\[
(\mathcal A_0)_{i\bar j}
=
-\frac{1}{12n(2\pi)^n}
\int_{\mathbb C^n}
\Delta_{\bar\partial}
\left(
e^{-|\partial F|^2}Q_{i\bar j}
\right)d^{2n}z.
\]
Since $Q_{i\bar j}$ has polynomial growth and
\[
|\partial F(z)|^2\geq c|z|^{2n-2},
\]
the integrand is rapidly decreasing. Hence Stokes' theorem gives $\mathcal A_0=0$.

\vskip 0.1cm
For the nonidentity sector $g\neq 1$, Step~7 gives
\[
\bigl[
\Str_{L^2\Omega^{*,*}(\mathbb C^n)}
(g^*\ph e^{-D_h})
\bigr]^{(1,1)}
=
O_{C^\infty}(h^{2n-2}),
\]
so the nonidentity sectors have no constant term.

Finally, apply Lemma~\ref{lem:matrixremainder}\textup{(iii)} with
$L=2n+1$. After multiplication by the factor
$(2\pi h^2)^{-n}$, the integrated remainder is
$O_{C^\infty}(h^2)$. 
Thus
\[
\bigl[\Str_{\cE}(\ph e^{-D_h})\bigr]^{(1,1)}
=h^{-2}\mathcal A_{-1}
+O_{C^\infty}(h^2),
\]
and hence
\[
\FP_{h\to0}
\bigl[\Str_{\cE}(\ph e^{-D_h})\bigr]^{(1,1)}=0.\qedhere
\]
\end{proof}

\begin{proof}[Proof of \Cref{prop:smalltime}]
Recall that we decompose
\[
\ph=\ph_{\mathrm{num}}+\ph_{\mathrm{diff}},
\]
where
\[
\ph_{\mathrm{num}}
=
\frac{n-1}{n}N^{1,0}+\frac{1}{n}N^{0,1},
\qquad
\ph_{\mathrm{diff}}
=
-\frac{1}{n}\sum_{a=1}^n
\left(
z_a\partial_{z_a}
-\bar z_a\partial_{\bar z_a}
\right).
\]
We compute $\ph_{\mathrm{num}}$-contribution with $g=1$ in step 2-5, the $\ph_{\mathrm{diff}}$-contribution with $g=1$ in step 6, and $g\neq1$-contribution in step 7.

\subsubsection*{Step 1. Extraction of the $db^i\wedge d\bar b^j$ coefficient}
Fix $i,j$.
To extract the $db^i\wedge d\bar b^j$ coefficient, take odd formal
variables $\eta,\bar\eta$  and set
\begin{equation}\label{eq:formalD}
\mathscr D_h
=
|\partial F|^2
+h^2\Delta_{\bar\partial}
+hL_F
+\eta\,\partial\phi_i\wedge
+\bar\eta\,(\partial\phi_j\wedge)^*.
\end{equation}
The formal variables above satisfy
\[
\eta^2=\bar\eta^2=0,
\qquad
\eta\bar\eta=-\bar\eta\eta.
\]
We also assume that
\[
\eta C=-C\eta,
\qquad
\bar\eta C=-C\bar\eta,
\qquad
\forall\,
C\in
\{dz_a\wedge,d\bar z_a\wedge,
  \iota_{\partial_{z_a}},\iota_{\partial_{\bar z_a}}\}.
\]
Note also that the symbol $\bar\eta$ is an independent formal variable and
is not related to complex conjugation.
If
$
K_{\mathscr D_h}(1;x,y;\eta,\bar\eta)
$
denotes the integral kernel of $e^{-\mathscr D_h}$, then
\[
[\eta\bar\eta]
K_{\mathscr D_h}(1;x,y;\eta,\bar\eta)
\]
denotes the coefficient of $\eta\bar\eta$ in its expansion in
$\eta,\bar\eta$. This notation is also used after taking the pointwise
or operator supertrace.

By Lemma~\ref{lem:matrixremainder},
\begin{equation}\label{eq:sourcebridge}
\begin{aligned}
K_{\mathscr D_h}(1;z,z;\eta,\bar\eta)
=
{ (2\pi h^2)^{-n}}&\mathbb E\biggl[
e^{-\int_0^1
|\partial F(z+h\mathcal B_u)|^2\,du}
\\
&\quad\cdot
\mathcal P\exp\left(
-\int_0^1
\Bigl(
hL_F
+\eta\,\partial\phi_i\wedge
+\bar\eta\,(\partial\phi_j\wedge)^*
\Bigr)
(z+h\mathcal B_u)\,du
\right)
\biggr].
\end{aligned}
\end{equation}

Expanding the path-ordered exponential in
\eqref{eq:sourcebridge}, we have
\be\label{eq: expansion time}
\begin{aligned}
&\mathcal P\exp\left(
-\int_0^1
\Big(
hL_F
+\eta\,\partial\phi_i\wedge
+\bar\eta\,(\partial\phi_j\wedge)^*
\Big)(z+h\mathcal B_u)\,du
\right)
\\
&\qquad=
\sum_{m=0}^{\infty}(-1)^m
\int_{0<u_1<\cdots<u_m<1}
\prod_{\nu=1}^m
\Big(
hL_F
+\eta\,\partial\phi_i\wedge
+\bar\eta\,(\partial\phi_j\wedge)^*
\Big)(z+h\mathcal B_{u_\nu})
\,du_1\cdots du_m .
\end{aligned}
\ee
Therefore, the coefficient of $\eta\bar\eta$ in a term containing
$k$ factors of $hL_F$ with $m=k+2$. More explicitly, it is
\[
(-1)^{k+2}h^k
\int_{0<u_1<\cdots<u_{k+2}<1}
\sum_{\substack{1\leq r,s\leq k+2\\ r\neq s}}
\epsilon_{r,s}
\prod_{\nu=1}^{k+2}Q_\nu^{r,s}
\,du_1\cdots du_{k+2},
\]
where
\[
Q_\nu^{r,s}
=
\begin{cases}
(\partial\phi_i\wedge)(z+h\mathcal B_{u_\nu}),
    & \nu=r,\\[2mm]
(\partial\phi_j\wedge)^*(z+h\mathcal B_{u_\nu}),
    & \nu=s,\\[2mm]
L_F(z+h\mathcal B_{u_\nu}),
    & \nu\neq r,s,
\end{cases}
\]
and
\[
\epsilon_{r,s}
=
\begin{cases}
-1, & r<s,\\
\phantom{-}1, & s<r.
\end{cases}
\]
The sign comes from moving the odd formal variables
$\eta,\bar\eta$ to the left in the order $\eta\bar\eta$.

We next expand all coefficients at $z$. For any scalar or
operator-valued smooth coefficient $G$, Taylor's formula gives, with
repeated real indices summed,
\[
\begin{aligned}
G(z+h\mathcal B_u)
={}&
G(z)
+h\,\partial_\alpha G(z)\mathcal B_u^\alpha
+\frac{h^2}{2}
\partial_\alpha\partial_\beta G(z)
\mathcal B_u^\alpha\mathcal B_u^\beta
+\cdots .
\end{aligned}
\]
In particular,
\[
\begin{aligned}
L_F(z+h\mathcal B_u)
={}&
L_F(z)
+h\,\partial_\alpha L_F(z)\mathcal B_u^\alpha
+\frac{h^2}{2}
\partial_\alpha\partial_\beta L_F(z)
\mathcal B_u^\alpha\mathcal B_u^\beta
+\cdots ,
\\
(\partial\phi_i\wedge)(z+h\mathcal B_u)
={}&
(\partial\phi_i\wedge)(z)
+h\,\partial_\alpha(\partial\phi_i\wedge)(z)
\mathcal B_u^\alpha
+\cdots ,
\\
(\partial\phi_j\wedge)^*(z+h\mathcal B_u)
={}&
(\partial\phi_j\wedge)^*(z)
+h\,\partial_\alpha(\partial\phi_j\wedge)^*(z)
\mathcal B_u^\alpha
+\cdots .
\end{aligned}
\]
Similarly, writing $V=|\partial F|^2$,
\[
V(z+h\mathcal B_u)
=
V(z)
+h\,\partial_\alpha V(z)\mathcal B_u^\alpha
+\frac{h^2}{2}
\partial_\alpha\partial_\beta V(z)
\mathcal B_u^\alpha\mathcal B_u^\beta
+\cdots.
\]

Thus, a term in the coefficient of $\eta\bar\eta$ containing $k$
factors of $L_F$ and having total order $\ell$ of Taylor expansion of operators above carries the
factor
\[
h^{-2n}h^kh^\ell
=
h^{k+\ell-2n},
\]
where $h^{-2n}$ comes from the factor
$(2\pi h^2)^{-n}$ in \eqref{eq:sourcebridge}.

\subsubsection*{Step 2. Admissible powers of $h$}

Recall that
\[
N^{1,0}
=
\sum_{a=1}^n
dz_a\wedge\iota_{\partial/\partial z_a},
\qquad
N^{0,1}
=
\sum_{a=1}^n
d\bar z_a\wedge\iota_{\partial/\partial\bar z_a}.
\]
For exponents $\varepsilon_a,\delta_a,\bar\varepsilon_a,
\bar\delta_a\in\{0,1\}$,
\[
\operatorname{str}\!\left(
\prod_{a=1}^n
(dz_a\wedge)^{\varepsilon_a}
\iota_{\partial/\partial z_a}^{\,\delta_a}
(d\bar z_a\wedge)^{\bar\varepsilon_a}
\iota_{\partial/\partial\bar z_a}^{\,\bar\delta_a}
\right)
=0
\]
unless
\[
\varepsilon_a=\delta_a=\bar\varepsilon_a=\bar\delta_a=1
\qquad
\text{for every }a=1,\ldots,n.
\]
In this case, the supertrace is $1$.

Thus, in the coefficient of $\eta\bar\eta$, the appearance of
$N^{1,0}$ or $N^{0,1}$ requires at least $2n-2$ factors of $L_F$,
whereas without them at least $2n$ factors of $L_F$ are required.

Also, counting the holomorphic degree implies that the number $k$ of factors of
$L_F$ is even to ensure nonzero super trace for some $k\geq 2n-2$.
Recall that a term in the coefficient of $\eta\bar\eta$ containing $k$
factors of $L_F$ and having total order $\ell$ of Taylor expansion of operators in Step 1 has factor
\[
h^{k+\ell-2n}.
\]
Since terms with odd $\ell$ have zero expectation under the symmetry
$\mathcal B\mapsto-\mathcal B$, the only terms through order $h^0$ are given by
\[
(k,\ell)
=
(2n-2,0),\qquad
(2n-2,2),\qquad
(2n,0).
\]

\subsubsection*{Step 3. The terms with $\ell=0$}

By Step 2, we only need to consider
\[
(k,\ell)=(2n-2,0)
\qquad\text{and}\qquad
(k,\ell)=(2n,0).
\]

Since
\[
\partial_z^2F(z)
=
\left(
\frac{\partial^2F}{\partial z_a\partial z_c}(z)
\right)_{a,c=1}^n
\]
is complex symmetric, choose a Takagi factorization
\[
\partial_z^2F(z)
=
U\operatorname{diag}(\lambda_1,\ldots,\lambda_n)U^T,
\qquad
\lambda_a\ge0,
\]
where $U$ is unitary and $\lambda_1,\ldots,\lambda_n$ are the singular values of
$\partial_z^2F(z)$. Thus
\[
(\partial_z^2F){ \overline{\partial_z^2F}^{\,T}}
=
U\operatorname{diag}(\lambda_1^2,\ldots,\lambda_n^2){ \bar U^{\,T}}.
\]

After the corresponding unitary change as above, the pointwise calculation
factors into calculation on $n$ four-dimensional eigenspaces. In one factor, let $e$
and $\bar e$ denote the corresponding holomorphic and antiholomorphic
one-forms. In the $a$-th factor, set $\lambda=\lambda_a$. With respect to
the ordered basis $(1,e,\bar e,e\wedge\bar e)$, the action of $L_F$ on
this factor is represented by
\[
L_\lambda
=
\begin{pmatrix}
0&0&0&0\\
0&0&2\lambda&0\\
0&2\lambda&0&0\\
0&0&0&0
\end{pmatrix}.
\]
Hence
\begin{equation}\label{eq:ordinaryblock}
\operatorname{str}(e^{-hL_\lambda})
=
2-2\cosh(2h\lambda)
=
-4h^2\lambda^2
\left(
1+\frac{h^2\lambda^2}{3}+O(h^4)
\right).
\end{equation}

On the same eigenspace, we have
\[
\frac{n-1}{n}N^{1,0}+\frac1nN^{0,1}
=
\operatorname{diag}
\left(
0,\frac{n-1}{n},\frac1n,1
\right),
\]
while exterior multiplication by $e$ and its adjoint are represented by
\[
e\wedge
=
\begin{pmatrix}
0&0&0&0\\
1&0&0&0\\
0&0&0&0\\
0&0&1&0
\end{pmatrix},
\qquad
(e\wedge)^*
=
\begin{pmatrix}
0&2&0&0\\
0&0&0&0\\
0&0&0&2\\
0&0&0&0
\end{pmatrix}.
\]
Duhamel's principal therefore gives
\[
\begin{aligned}
&[\eta\bar\eta]\,
\operatorname{str}\Bigg[
\left(
\frac{n-1}{n}N^{1,0}
+\frac1nN^{0,1}
\right)
e^{
-hL_\lambda
-\eta ({ \bar U^{\,T}}\partial_z\phi_i)_a\,e\wedge
-\bar\eta\,
\overline{({ \bar U^{\,T}}\partial_z\phi_j)_a}\,
(e\wedge)^*
}
\Bigg]
\\
&\qquad
=
-{ \frac2n}
({ \bar U^{\,T}\partial_z\phi_i})_a
\overline{({ \bar U^{\,T}}\partial_z\phi_j)_a}
\left(
1+{ \frac{h^2\lambda^2}{3}}+O(h^4)
\right).
\end{aligned}
\]
While without the appearance of $\frac{n-1}{n}N^{1,0}+\frac1nN^{0,1}$ in the super-trace,
\begin{equation}\label{eq:unweightedzero-one}
[\eta\bar\eta]\,
\operatorname{str}
\left(
e^{
-hL_\lambda
-\eta ({ \bar U^{\,T}}\partial_z\phi_i)_a\,e\wedge
-\bar\eta\,
\overline{({ \bar U^{\,T}}\partial_z\phi_j)_a}\,
(e\wedge)^*
}
\right)
=0.
\end{equation}
Recall \eqref{eq:Q}, we have
\begin{align}
&[\eta\bar\eta]\,
\operatorname{str}\Bigg[
\left(
\frac{n-1}{n}N^{1,0}
+\frac1nN^{0,1}
\right)
e^{
-hL_F
-\eta\,\partial\phi_i\wedge
-\bar\eta\,(\partial\phi_j\wedge)^*
}
\Bigg]
\notag\\
&\qquad
=
h^{2n-2}Q_{i\bar j}
+
{ \frac{h^{2n}}{3}}
\operatorname{tr}
\left(
(\partial_z^2F){ \overline{\partial_z^2F}^{\,T}}
\right)
Q_{i\bar j}
+
O(h^{2n+2}).
\label{eq:frozenclifford}
\end{align}
Thus the term $(k,\ell)=(2n-2,0)$ gives the 
$h^{-2} Q_{i\bar j}$, while the term $(k,\ell)=(2n,0)$ gives
\begin{equation}\label{eq:ellzero-constant}
{ \frac{1}{3}}
e^{-|\partial F|^2}
Q_{i\bar j}
\operatorname{tr}
\left(
(\partial_z^2F){ \overline{\partial_z^2F}^{\,T}}
\right).
\end{equation}

\subsubsection*{Step 4. The term with $(k,\ell)=(2n-2,2)$}

We can expand \eqref{eq:Q}:
\begin{equation}\label{eq:Qmonomials}
Q_{i\bar j}(z)
=
\sum_\rho c_\rho
\prod_{\nu=1}^{2n}f_{\rho,\nu}(z),
\end{equation}
where $\rho$ indexes the monomials in the determinant expansion of
\[
{ \overline{\partial_z\phi_j}^{\,T}}
\operatorname{adj}\!\left(
(\partial_z^2F){ \overline{\partial_z^2F}^{\,T}}
\right)
\partial_z\phi_i.
\]
Here $c_\rho$ is the corresponding
numerical coefficient and each
$f_{\rho,\nu}$ is one of the indicated matrix or vector entries.
Each product contains $n-1$ entries of
$\partial_z^2F$, $n-1$ entries of $\overline{\partial_z^2F}$,
one entry of $\partial_z\phi_i$, and one entry of
$\overline{\partial_z\phi_j}$. Fix a monomial indexed by $\rho$ in
\eqref{eq:Qmonomials}. Each factor $f_{\rho,\nu}$ depends either
only on $z$ or only on $\bar z$. Therefore,
\begin{equation}\label{eq:harmonicfactor}
\Delta_{\bar\partial}f_{\rho,\nu}=0,
\qquad
\nu=1,\ldots,2n.
\end{equation}

By the Clifford relation, we can replace the
ordered simplices in \eqref{eq: expansion time} by the cube $[0,1]^{m}$. Then ignoring the factor $(2\pi)^{-n}h^{-2n}$, the contribution of
$c_\rho\prod_{\nu=1}^{2n}f_{\rho,\nu}$ is
\begin{equation}\label{eq:monomialbridge}
c_\rho
\int_{[0,1]^{2n}}
\mathbb E\left[
\exp\left(
-\int_0^1
|\partial F(z+h\mathcal B_u)|^2\,du
\right)
\prod_{\nu=1}^{2n}
f_{\rho,\nu}(z+h\mathcal B_{u_\nu})
\right]
\,du_1\cdots du_{2n}.
\end{equation}

Let $\partial_\alpha$, $\alpha=1,\ldots,2n$, denote differentiation with
respect to the real coordinates on $\C^n\simeq\R^{2n}$. We use Einstein's
summation convention below.  Up to second order in $h$,
\begin{equation}\label{eq:factorTaylor}
\begin{aligned}
f_{\rho,\nu}
\left(z+ h\mathcal B_{u_\nu}
\right)={}& f_{\rho,\nu}+ h
\partial_\alpha f_{\rho,\nu}\,
\mathcal B_{u_\nu}^\alpha
+\frac{h^2}{2}
\partial_\alpha\partial_\beta f_{\rho,\nu}\,
\mathcal B_{u_\nu}^\alpha
\mathcal B_{u_\nu}^\beta
+O(h^3).
\end{aligned}
\end{equation}
 Also,
\begin{equation}\label{eq:potentialTaylor}
\begin{aligned}
&\quad e^{-\int_0^1
|\partial F(z+h\mathcal B_u)|^2\,du}
=
e^{-|\partial F|^2}
\Big[
1
{ -h}
\int_0^1
\partial_\alpha|\partial F|^2\,
\mathcal B_u^\alpha\,du
\\
&
{ -\frac{h^2}{2}}
\int_0^1
\partial_\alpha\partial_\beta|\partial F|^2\,
\mathcal B_u^\alpha\mathcal B_u^\beta\,du
{ +\frac{h^2}{2}}
\int_0^1\!\!\int_0^1
\partial_\alpha|\partial F|^2
\partial_\beta|\partial F|^2\,
\mathcal B_u^\alpha\mathcal B_v^\beta\,du\,dv
\Big]
+O(h^3).
\end{aligned}
\end{equation}

The covariance formula
\begin{equation}\label{eq:bridgecovariance-secondorder}
\mathbb E
\left(
\mathcal B_u^\alpha\mathcal B_v^\beta
\right)
=
\bigl(\min(u,v)-uv\bigr)\delta_{\alpha\beta}
\end{equation}
implies
\begin{equation}\label{eq:bridgeintegrals-secondorder}
\int_0^1u(1-u)\,du=\frac16,
\qquad
\int_0^1\!\!\int_0^1
\bigl(\min(u,v)-uv\bigr)\,du\,dv
=
\frac1{12}.
\end{equation}

There are five contributions to the coefficient of $h^2$ in
\eqref{eq:monomialbridge}.

\paragraph{(a)}
The $h^2$ term of one factor in \eqref{eq:factorTaylor}, with all other
terms taken at order zero, gives
\[
{ -\frac1{6}}
c_\rho
\sum_{\nu=1}^{2n}
 \Delta_{\bar\partial}f_{\rho,\nu}
\prod_{\substack{\mu=1\\ \mu\ne\nu}}^{2n}
f_{\rho,\mu}
=0
\]
by \eqref{eq:harmonicfactor}.

\paragraph{(b)}
The product of the $h$ terms of two distinct factors in
\eqref{eq:factorTaylor} gives
\[
{ \frac1{12}}
c_\rho
\sum_{\nu<\mu}
\nabla f_{\rho,\nu}\cdot\nabla f_{\rho,\mu}
\prod_{\substack{\sigma=1\\
\sigma\ne\nu,\mu}}^{2n}
f_{\rho,\sigma}.
\]
By \eqref{eq:harmonicfactor},
\[
- \Delta_{\bar\partial}
\left(
c_\rho\prod_{\nu=1}^{2n}f_{\rho,\nu}
\right)
=
{ c_\rho}
\sum_{\nu<\mu}
\nabla f_{\rho,\nu}\cdot\nabla f_{\rho,\mu}
\prod_{\substack{\sigma=1\\
\sigma\ne\nu,\mu}}^{2n}
f_{\rho,\sigma}.
\]
Summing over $\rho$ therefore gives
\begin{equation}\label{eq:secondorder-b}
-\frac1{12} \Delta_{\bar\partial}Q_{i\bar j}.
\end{equation}

\paragraph{(c)}
The product of the $h$ term in \eqref{eq:factorTaylor} with the
$h$ term in \eqref{eq:potentialTaylor} gives
\begin{equation}\label{eq:secondorder-c}
-{ \frac1{12}}
\nabla|\partial F|^2\cdot\nabla Q_{i\bar j}.
\end{equation}

\paragraph{(d)}
The quadratic term of the scalar exponential in
\eqref{eq:potentialTaylor} gives
\begin{equation}\label{eq:secondorder-d}
\frac1{6}
Q_{i\bar j} \Delta_{\bar\partial}|\partial F|^2.
\end{equation}

\paragraph{(e)}
The last term in \eqref{eq:potentialTaylor}, which is the square of
the linear variation of $|\partial F|^2$, gives
\begin{equation}\label{eq:secondorder-e}
{ \frac1{24}}
Q_{i\bar j}
\left|
\nabla|\partial F|^2
\right|^2.
\end{equation}

Combining \textup{(a)}--\textup{(e)}, the contribution with
$(k,\ell)=(2n-2,2)$ is
\begin{equation}\label{eq:bridgeordertwo}
e^{-|\partial F|^2}
\left[
-\frac1{12} \Delta_{\bar\partial}Q_{i\bar j}
-{ \frac1{12}}
\nabla|\partial F|^2\cdot\nabla Q_{i\bar j}
+\frac1{6}
Q_{i\bar j} \Delta_{\bar\partial}|\partial F|^2
+{ \frac1{24}}
Q_{i\bar j}
\left|
\nabla|\partial F|^2
\right|^2
\right].
\end{equation}

\subsubsection*{Step 5. Combining the two constant contributions}

Recall that the two contributions to the constant term are
\eqref{eq:ellzero-constant}, obtained in Step~3, and
\eqref{eq:bridgeordertwo}, obtained in Step~4.

Since 
$
- \Delta_{\bar\partial}|\partial F|^2
=
{ 4}\operatorname{tr}
\left(
(\partial_z^2F){ \overline{\partial_z^2F}^{\,T}}
\right),$ adding \eqref{eq:ellzero-constant} to
\eqref{eq:bridgeordertwo} gives
\begin{align}
&e^{-|\partial F|^2}
\Big(
-\frac1{12} \Delta_{\bar\partial}Q_{i\bar j}
-{ \frac1{12}}
\nabla|\partial F|^2\cdot\nabla Q_{i\bar j}
+\frac1{12}
Q_{i\bar j} \Delta_{\bar\partial}|\partial F|^2
+{ \frac1{24}}
Q_{i\bar j}
\left|
\nabla|\partial F|^2
\right|^2
\Big).
\label{eq:constant-before-product-rule}
\end{align}
Finally,
\[
\begin{aligned}
 -\Delta_{\bar\partial}
\left(
e^{-|\partial F|^2}Q_{i\bar j}
\right)
=
e^{-|\partial F|^2}
\Big(
& -\Delta_{\bar\partial}Q_{i\bar j}
{ -}\nabla|\partial F|^2\cdot\nabla Q_{i\bar j}
+{ \frac{1}{2}}
Q_{i\bar j}
\left|\nabla|\partial F|^2\right|^2
+
Q_{i\bar j}
 \Delta_{\bar\partial}|\partial F|^2
\Big).
\end{aligned}
\]
Therefore the complete constant-order contribution of
$
\frac{n-1}{n}N^{1,0}+\frac1nN^{0,1}
$
is
\begin{equation}\label{eq:densitysum}
-\frac1{12}
 \Delta_{\bar\partial}
\left(
e^{-|\partial F|^2}Q_{i\bar j}
\right).
\end{equation}

\subsubsection*{Step 6. The differential part of $\ph$}

Recall that
\[
\widehat{p}_{\mathrm{diff}}:=-\frac1n
\sum_{a=1}^n
\left(
z_a\partial_{z_a}
-\bar z_a\partial_{\bar z_a}
\right).
\]
Then the supertrace of
$\widehat{p}_{\mathrm{diff}}e^{-\mathscr{D}_h}$
is obtained by applying the derivatives in $\widehat{p}_{\mathrm{diff}}$ to the
first variable of the integral kernel of $e^{-\mathscr{D}_h}$, then
restricting to the diagonal and integrating the resulting pointwise
supertrace. That is, ignoring the factor $-1/n$,
the corresponding expression is
\begin{equation}\label{eq:Ekernel}
\begin{aligned}
&[\eta\bar\eta]
\int_{\C^n}
\operatorname{str}\Bigg[
\left(
\sum_{a=1}^n\left(z_a\partial_{z_a}-\bar z_a\partial_{\bar z_a}\right)
\right)
K_{\mathscr D_h}(1;z,w;\eta,\bar\eta)
\Bigg]_{w=z}
\,d^{2n}z.
\end{aligned}
\end{equation}

By the same argument as in Step~2, only $(k,\ell)=(2n,0)$ can contribute to the constant term.

For the factor $e^{-|x-y|^2/2h^2}$, note that
\[
\left.
\partial_{x^\alpha}
e^{-|x-y|^2/2h^2}
\right|_{y=x}
=0.
\]
It suffices to act $\widehat{p}_{\mathrm{diff}}$ on
\[
e^{
-hL_F
-\eta\,\partial \phi_i\wedge
-\bar\eta\,(\partial \phi_j\wedge)^*
}.
\]
Note that for a smooth function $f$, we have
\begin{equation}\label{eq:endpointbridge}
\partial_{x^\alpha}
f\left(
(1-u)x+uy+ h\mathcal B_u
\right)
=
(1-u)
(\partial_\alpha f)
\left(
(1-u)x+uy+ h\mathcal B_u
\right),
\end{equation}
which implies that, by Duhamel's principle, the integral kernel of $\widehat{p}_{\mathrm{diff}}e^{-\mathscr{D}_h}$ is
\[
\begin{aligned}
-\int_0^1&(1-\tau)
\operatorname{str}\Big[
e^{-\tau(
hL_F+\eta\,\partial \phi_i\wedge
+\bar\eta\,(\partial \phi_j\wedge)^*)}
\\
&\cdot
\Big(
\sum_{a=1}^n\left(z_a\partial_{z_a}-\bar z_a\partial_{\bar z_a}\right)
\Big)
\left(
hL_F
+\eta\,\partial \phi_i\wedge
+\bar\eta\,(\partial \phi_j\wedge)^*
\right)
e^{-(1-\tau)(
hL_F+\eta\,\partial \phi_i\wedge
+\bar\eta\,(\partial\phi_j\wedge)^*)}
\Big]
\,d\tau.
\end{aligned}
\]
Note that the supertrace vanishes on supercommutators (\Cref{prop:familyheat}). Together with
$
\int_0^1(1-\tau)\,d\tau=\frac12,
$
this implies that the above expression equals
\[
\frac12
\Big(
\sum_{a=1}^n
\left(
z_a\partial_{z_a}-\bar z_a\partial_{\bar z_a}
\right)
\Big)
\left[
\operatorname{str}
\left(
e^{
-hL_F
-\eta\,\partial \phi_i\wedge
-\bar\eta\,(\partial \phi_j\wedge)^*
}
\right)
\right].
\]
Its coefficient of $\eta\bar\eta$ is zero by
\eqref{eq:unweightedzero-one}. Hence the $\widehat{p}_{\mathrm{diff}}$-part of $\ph$
does not contribute a constant term.

\subsubsection*{Step 7. The $g\neq1$ contribution}

For a trace-class operator $T$ commuting with the $\Z_n$-action,
\[
\Str_{\cE}(T|_{\cE})
=\frac1n\sum_{g\in\Z_n}
\Str_{L^2\Omega^{*,*}(\C^n)}(g^*T).
\]
We have finished $g=1$ computation in Steps~3--6; it remains
to estimate $g\ne1$-terms.

Replacing $g$ by $g^{-1}$ only reindexes the finite group sum.
Therefore the $g$-twisted trace is obtained from the kernel evaluated
at $(z,gz)$. Set $z=\sqrt{2}hZ$. Then
\[
e^{-|z-gz|^2/2h^2}
=
e^{-|1-g|^2|Z|^2},
\qquad
(2\pi h^2)^{-n}d^{2n}z
=
\pi^{-n}d^{2n}Z,
\]
and the bridge appearing in the kernel becomes
\[
(1-u)z+ugz+ h\mathcal B_u
=
h\left(
{ \sqrt{2}}(1-u)Z+{ \sqrt{2}}ugZ+\mathcal B_u
\right).
\]
The homogeneity of $F$ and $\phi_i$ gives, for $W\in\C^n$,
\be\label{eq: change of variable}
\begin{aligned}
|\partial F(hW)|^2
&=
h^{2n-2}|\partial F(W)|^2,\\
hL_F(hW)
&=
h^{n-1}L_F(W),\\
(\partial \phi_i)(hW)\wedge
&=
h^{n-1}(\partial \phi_i)(W)\wedge,\\
((\partial \phi_j)(hW)\wedge)^*
&=
h^{n-1}
((\partial \phi_j)(W)\wedge)^*.
\end{aligned}
\ee
Note that the coefficient of $\eta\bar\eta$ contains one term of each
of the last two operators and hence carries the factor
$h^{2n-2}$. Also, homogeneity and non-degeneracy give for some constant $C_K>0,$
\[|\partial F|^2-h\|L_F\|\ge-C_K,\quad \text{for }0<h\le1.\] Moreover, it is clear that every moment
of $\sup_{0\le u\le1}|\mathcal B_u|$ is finite. Thus, for some integer
$D\ge0$ depending only on the polynomial degrees and derivative orders,
\[
\left|
[\eta\bar\eta]\,
\Str_{L^2\Omega^{*,*}(\C^n)}\left(
g^*\ph_{\mathrm{num}} e^{-\mathscr D_h}
\right)
\right|
\le
Ch^{2n-2}
\int_{\C^n}
e^{-|1-g|^2|Z|^2}(1+|Z|)^D\,d^{2n}Z
=
O(h^{2n-2}).
\]

It remains to treat $\ph_{\mathrm{diff}}$. Before restricting to the
twisted diagonal $(z,gz)$, write the kernel in the variables $x=hX$ and
$y=hY$, and then set $Y=gX$.
Differentiation of $e^{-|X-Y|^2/{ 2}}$ produces only a polynomial factor in
$X,Y$, and differentiation of the remaining parts does not
affect the factor $h^{2n-2}$. Hence the same estimate holds for the
differential part. Therefore, for every $g\ne1$,
\[
[\eta\bar\eta]\,
\Str_{L^2\Omega^{*,*}(\C^n)}\left(
g^*\ph e^{-\mathscr D_h}
\right)
=
O_{C^\infty}(h^{2n-2}).
\]

\subsubsection*{Step 8. Integration and the remainder}

Since $F_b$ has no critical point away from the origin, homogeneity and non-degeneracy imply
\[
|\partial F_b(z)|^2\ge c_K|z|^{2n-2},\qquad K\Subset B.
\]
Since $Q_{i\bar j}$ and its derivatives have polynomial growth,
$
e^{-|\partial F_b|^2}Q_{i\bar j}
$
and its  derivatives decay rapidly at infinity. Hence Stokes'
theorem gives
\[
\int_{\C^n}
 \Delta_{\bar\partial}
\left(
e^{-|\partial F_b|^2}Q_{i\bar j}
\right)\,d^{2n}z
=0.
\]
By \eqref{eq:densitysum} and discussions above, this proves
\[
(\mathcal A_0)_{i\bar j}=0.
\]
The $O_{C^\infty}(h^2)$ remainder follows directly from
Lemma~\ref{lem:matrixremainder}.
\end{proof}

\section{The invariance of modified LG torsion}
\label{sec:modified}
In this section, we prove the invariance of modified LG torsion under scalar multiplication and linear
changes of coordinates.

First, we note that:

\begin{lemma}\label{lem:N2reduction}
For every $s>0$,
\begin{equation}\label{eq:N2reduction}
\Str\!\left(\ph\qh(e^{-s\Delta_F}-P_F)\right)
=\frac12\Str\!\left(N^2(e^{-s\Delta_F}-P_F)\right).
\end{equation}
Here $N:=N^{0,1}+N^{1,0}.$ We also have
\begin{equation}\label{eq:centered}
\Str\left[\widehat p\widehat q(e^{-s\Delta_{F}}-P_{F})\right]
   =\frac12\Str\!\left[(N-n)^2e^{-s\Delta_{F}}\right].
\end{equation}
\end{lemma}
\begin{proof}
As in the proof of \Cref{cor:bcov-split}, we can see that $$\Str\Big(\widehat{q}^2(e^{-t\Delta_F}-P_F)\Big)=\Str\Big(\widehat{p}^2(e^{-t\Delta_F}-P_F)\Big)=0.$$ Now note that $2\ph\qh=N^2-\ph^2-\qh^2$.

\vskip 0.1cm
Note also that harmonic forms are all $n$-forms, and as in the proof of \Cref{cor:bcov-split}, we can see that $$\Str\Big(N(e^{-t\Delta_F}-P_F)\Big)=\Str\Big(\widehat{p}(e^{-t\Delta_F}-P_F)\Big)+\Str\Big(\widehat{q}(e^{-t\Delta_F}-P_F)\Big)=0.$$
Thus
\[
 \Str\!\left[\left(\tfrac12(N-n)^2-\tfrac12N^2\right)
       (e^{-s\Delta_{F}}-P_
       F)\right]=0,
\]
which proves \eqref{eq:centered}.
\end{proof}

We first consider the case where $A$ is the scalar matrix $a^{1/n}I$. Note that,
for $a\in\C^*$, the unitary dilation
$
z=|a|^{-1/n}y,
$
followed by a constant unitary rotation of the holomorphic form degree,
conjugates $\Delta_{aF}$ to $|a|^{2/n}\Delta_F$. These maps preserve
$\ph$, $\qh$, and the invariant subspace. Thus changing the heat time in
\eqref{eq:zeta} and differentiating at $w=0$ gives
\begin{align}
\zeta(w;aF)&=|a|^{-2w/n}\zeta(w;F),
\label{eq:zetascaling}\\
\cT^\LG(aF)-\cT^\LG(F)
&=\frac{\zeta(0;F)}{n}\log|a|^2.
\label{eq:scalingunknown}
\end{align}
This motivates us to compute $\zeta(0;F)$.
\begin{proposition}\label{prop:zetazero}
For a smooth degree-$n$ homogeneous polynomial,
\begin{equation}\label{eq:zetazero}
 \zeta(0;F)=\frac{n\chi(X_F)}{12}.
\end{equation}
\end{proposition}

\begin{proof}
\noindent\textbf{Step 1: constancy follows from anomaly.}
Consider the family $G_{a,b}=aF_b$ on $\C^*\times B$.
Let $\mathcal E_z=\sum_\nu z_\nu\partial_{z_\nu}$ be the holomorphic
Euler vector field (distinct from the rotation generator $E$).
We write $G=G_{a,b}$ when $a,b$ are fixed. On forms, homogeneity gives
\[
 [\bar\partial_G,\iota_{\mathcal E_z}]
 =M_{\mathcal E_zG}=nM_G.
\]
Here $[\cdot,\cdot]$ is the supercommutator.

 Since $\partial_aG=F_b=G/a$, the identity implies
\be\label{eq: vanishing Ca} C_a:=P_GM_{\partial_aG}P_G=0.\ee This is because, harmonic forms are rapidly decreasing
by Proposition~\ref{prop:TY_summary} (i) (see also \cite[Lemma~A.3]{TY}),
the contraction $\iota_{\mathcal E_z}$ has polynomial coefficients, and $P_G\bar\partial_G
=\bar\partial_GP_G=0$. Here recall that $P_G$ is the orthogonal projection to harmonic forms. The equation \eqref{eq:rawanomaly} therefore gives
\[
 0=\partial_a\partial_{\bar b^j}\cT^\LG(aF_b).
\]
Apply this to \eqref{eq:scalingunknown}, we have:
\[
 0=\frac{1}{na}\,\partial_{\bar b^j}\zeta(0;F_b).
\]
Since $\zeta(0;F_b)$ is smooth and real, we also have $\partial_{ b^j}\zeta(0;F_b)=0$.
It follows that all of its derivatives in the $B$-directions vanish. {Note that the smooth locus of degree-$n$ homogeneous polynomials is the complement
of the discriminant hypersurface and is therefore connected. We know $\zeta(0;F)$ is constant on this locus, and it hence suffices to compute
it at the Fermat point $F_0=z_1^n+\cdots+z_n^n$.}

\smallskip
\noindent\textbf{Step 2: the exact Fermat product trace.} At $F_0$, both the twisted Laplacian and the diagonal $\Z_n$-action split into their one-variable factors. Therefore we work on  $z^n$ on $\C$.

On functions, the twisted Laplacian is
\[
\Delta_{z^n}^{(0)}
:=
\left.\Delta_{z^n}\right|_{L^2\Omega^{0,0}(\mathbb C)}
=
\Delta_{\bar\partial}
+
2n^2|z|^{2n-2}.
\]
For $g\in\mathbb Z_n$, define
\[
\begin{aligned}
A_g(s)
&:=
\Tr_{L^2(\mathbb C)}
\left(g^*e^{-s\Delta_{z^n}^{(0)}}\right),\\
I_g(s)
&:=
\Str_{L^2\Omega^{*,*}(\mathbb C)}
\left(g^*e^{-s\Delta_{z^n}}\right).
\end{aligned}
\]

A standard argument as in the proof of \Cref{cor:bcov-split} shows that $I_g$ is independent of $s$, and depends only on contribution from harmonic forms. 

Moreover, there is a canonical isomorphism
\[
S_{z^n}:
\mathrm{Jac}(z^n):=
\frac{\C[z]\,dz}{n z^{n-1}\C[z]\,dz}
\xrightarrow{\;\sim\;}
\ker\Delta_{z^n}\subset L^2\Omega^{*,*}(\C),
\]
characterized by  ``Hodge decomposition"
\[
z^k\,dz\longmapsto w_k\in\ker\Delta_{z^n},\quad 0\leq k\leq n-2
\]
where $w_k$ is the harmonic form for which the equation
\[
w_k-z^k\,dz=\bar\partial_{z^n}\nu_k,
\qquad
0\leq k\leq n-2,
\]
is solvable, with $\nu_k$ a smooth function of at most polynomial growth, see Proposition~\ref{prop:TY_summary} (i).

Since $\bar\partial_{z^n}$ commutes with $g^*$ and $g^*z^kdz=g^{k+1}z^kdz$, we have
\[g^*w_k=g^{k+1}w_k.\]
The equation above also implies that $\{w_k\}$ form an orthogonal base of $\ker(\Delta_{z^n})$ since $g^*$ is unitary. Thus
\[
 I_g(s)=\lim_{s'\to\infty}I_g(s')=-\sum_{k=1}^{n-1}g^k
   =\begin{cases}-(n-1),&g=1,\\1,&g\ne1.\end{cases}
\]

Let
\[
 Z_g(u;s):=\Str_{L^2\Omega^{*,*}(\C)}(g^*u^Ne^{-s\Delta_{z^n}}).
\]
Multiplication by $dz\wedge d\bar z$ identifies the restrictions of
$\Delta_{z^n}$ to $\Omega^{0,0}(\mathbb C)$ and
$\Omega^{1,1}(\mathbb C)$. Moreover,
\[
g^*(dz\wedge d\bar z)
=
g\bar g\,dz\wedge d\bar z
=
dz\wedge d\bar z.
\]
Hence the  $(0,0)$ and $(1,1)$ parts each contribute
$A_g(s)$ to $I_g(s)$. Thus
\[ Z_g(u;s)
   =A_g(s)(1-u)^2+uI_g.\]

Since we have $\Str_{L^2\Omega^{*,*}(\C^n)}\!\left(g^*u^Ne^{-s\Delta_{F_0}}\right)=Z_g(u;s)^n$,
\[\Str_{L^2\Omega^{*,*}(\C^n)}\!\left(N^2g^*e^{-s\Delta_{F_0}}\right)=(u\partial_u)^2Z_g(u;s)^n|_{u=1}=2n I_g^{n-1}A_g(s)+n^2I_g^n.\]

Thus by Lemma \ref{lem:N2reduction}
\begin{equation}\label{eq:Fermattrace}
 \Str\!\left(N^2(e^{-s\Delta_{F_0}}-P_{F_0})\right)=2\sum_{g\in\Z_n}I_g^{n-1}A_g(s).
\end{equation}

\noindent\textbf{Step 3: the constant coefficients.}
Let $a_{0,g}$ be the constant coefficient in the small-$s$ expansion
of $A_g(s)$. In this step, $\nabla$ is the real Euclidean gradient. For $g=1$, proceeding as in the proof of \Cref{prop:smalltime}, we have
\begin{align*}
a_{0,1}
 &={ \frac1{2\pi}}
   \int_{\R^{2}}e^{-2n^2|z|^{2n-2}}
       \left(\frac{\Delta_{\bar{\p}}(n^2|z|^{2n-2})}{3}
        +\frac{\left|\nabla(n^2|z|^{2n-2})\right|^2}{{ 6}}\right)d^2z\\
 &=\frac{1}{12\pi}\int_{\R^{2}}e^{-2n^2|z|^{2n-2}}\Delta_{\bar{\p}}(n^2|z|^{2n-2})\,d^2z
 =-\frac{n-1}{12}.
\end{align*}

For
$g=e^{2\pi ij/n}\ne1$, equivariant localization 
gives
\begin{equation}\label{eq:nonidscalar}
 a_{0,g}=\frac{1}{2\pi s}\int_{\R^2}e^{-|1-g|^2|z|^2/(2s)}d^2z=\frac1{|1-g|^2}=\frac1{4\sin^2(\pi j/n)}.
\end{equation}

Equation~\eqref{eq:Fermattrace} now gives
\[
 \zeta(0;F_0)
 =(1-n)^{n-1}\left(-\frac{n-1}{12}\right)
  +\sum_{j=1}^{n-1}\frac1{4\sin^2(\pi j/n)}=\frac{(1-n)^n+n^2-1}{12}
            =\frac{n\chi(X_F)}{12}.
\]
\end{proof}

We now prove invariance of genus one term under linear changes of coordinates. 

We start with the following observation:
\begin{lemma}
\label{lem:GLpluriharmonic}
Let $u:GL(n,\C)\to\R$ be pluriharmonic and satisfy
$u(AU)=u(A)$ for all $U\in U(n)$. Then there are real constants
$c,c_0$ such that
\[
 u(A)=c\log|\det A|^2+c_0.
\]
\end{lemma}

\begin{proof}
Set $\alpha=\partial u$, then
\[\bar\p\alpha=\bar\p\p u=0,\qquad d\alpha=\p^2u+\bar\p\p u=0.
\]
Moreover, right unitary invariance gives $R_U^*\alpha=\alpha$ for $U\in U(n)$,
where $R_U(A)=AU$.

Fix $A\in GL(n,\C)$, $H\in T_AGL(n,\C)\simeq M_n(\C)$, and set
\[f_{A,H}(G):=\alpha_{AG}(HG)-\alpha_A(H).\]
Then $f_{A,H}$ is holomorphic in $G\in GL(n,\C)$ and vanishes on $U(n)$. Since $U(n)$ is the real-analytic maximal totally real submanifold of $GL(n,\C)$, we have
\[f_{A,H}\equiv0,\qquad\text{i.e.}\quad R_G^*\alpha=\alpha,~~\text{for any } G\in GL(n,\C).\]

Let $\ell=\alpha_I$, where $I$ is the identity matrix. Right invariance gives $\alpha_A(H)=\ell(HA^{-1})$. Evaluate $d\alpha=0$ on two
right-invariant vector fields $X^R,Y^R$ given by $X,Y\in\mathfrak{gl}(n,\C)$. Then
\[0=d\alpha(X^R,Y^R)=-\alpha([X^R,Y^R])=\alpha([X,Y]^R)=\ell([X,Y]).\]
Thus,
\[\ell\big|_{\mathfrak{sl}(n,\C)}=0.\]
As a result,
\[ \ell(X)=c\Tr(X),\quad\text{for some }c\in\C.\]
Thus,
\[\alpha=\partial u=c\Tr(A^{-1}dA).\]
Along $A(t)=e^{it}I$, right unitary invariance gives
$0=\frac{d}{dt}u(A(t))=2\Re(c\,in)$.
It follows that $c$ is real.
Finally,
$\partial\log|\det A|^2=\Tr(A^{-1}dA)$ globally. It follows that
\[
\partial\left(u-c\log|\det A|^2\right)
=
\bar\partial\left(u-c\log|\det A|^2\right)
=
0.
\]
Since $GL(n,\mathbb C)$ is connected,
$u-c\log|\det A|^2$ is constant.
\end{proof}

\begin{proposition}
\label{prop:GLinvariance}
For every $A\in GL(n,\C)$, with the fixed Euclidean spectral definition
and the distinguished vacuum $w_0$,
\begin{align}
 \cT^\LG(F\circ A)-\cT^\LG(F)
   &=\frac{\chi(X_F)}{12}\log|\det A|^2,
      \label{eq:TGL}\\
 h_{F\circ A}(w_0(F\circ A),w_0(F\circ A))&=|\det A|^{-2}h_F(w_0(F),w_0(F)).\label{eq:hGL}
\end{align}
It follows that
\begin{equation}\label{eq:TmodGL}
\cT^\LG_{\mathrm{mod}}(F\circ A)
=
\cT^\LG_{\mathrm{mod}}(F).
\end{equation}
\end{proposition}

\begin{proof}
\noindent\textbf{Step 1: variations along the $GL(n,\C)$-orbit vanish on the harmonic space.}
We first examine the infinitesimal deformation generated by a linear
change of coordinates. Every tangent vector at $A\in GL(n,\C)$ is represented by
\[
A(t)=Ae^{tB},
\qquad B\in\mathfrak{gl}(n,\C).
\]
The corresponding variation of the polynomial is
\[
\dot F_{A,B}(z)
:=
\left.\frac{d}{dt}(F\circ Ae^{tB})(z)\right|_{t=0}
=
\sum_j(Bz)_j\,\partial_{z_j}(F\circ A)(z).
\]
In the notation of \eqref{eq:higgs}, the operator corresponding to the
tangent vector $\dot A=AB$ is therefore
\[
C_{A,B}
:=
P_{F\circ A}\,M_{\dot F_{A,B}}\,P_{F\circ A}.
\]
For the holomorphic vector field
\[
v_B=\sum_j(Bz)_j\partial_{z_j},
\]
we have
\[
M_{\dot F_{A,B}}
=
[\bar\partial_{F\circ A},\iota_{v_B}].
\]
Using the same argument as in the proof of \eqref{eq: vanishing Ca}
\[
C_{A,B}=0,
\qquad \forall
A\in GL(n,\mathbb C),\quad
B\in T_A GL(n,\mathbb C).
\]

\noindent\textbf{Step 2: Reduction to the determinant.} Restricting the anomaly formula \eqref{eq:rawanomaly} to the
$GL(n,\mathbb C)$-orbit and using $C_{A,B}=0$, we obtain
\[
\partial_A\bar\partial_A
\cT^\LG(F\circ A)=0.
\]
Thus $A\mapsto\cT^\LG(F\circ A)$ is pluriharmonic.

If $U\in U(n)$, pullback by $U$ is a unitary map of the Euclidean
form spaces and sends the operator associated with $F\circ A$ to the one associated with $F\circ AU$.
The pullback also commutes with $\Z_n$-action and $\widehat{p},\widehat{q}$. Thus, 
 $\cT^\LG(F\circ AU)=\cT^\LG(F\circ A)$. Lemma~\ref{lem:GLpluriharmonic} now gives
\[
 \cT^\LG(F\circ A)-\cT^\LG(F)=c\log|\det A|^2.
\]
Taking $A=rI$, $r>0$, and applying \eqref{eq:scalingunknown} with
$a=r^n$ gives $c=\zeta(0;F)/n$. Proposition~\ref{prop:zetazero}
therefore gives $c=\chi(X_F)/12$, proving \eqref{eq:TGL}.

\smallskip
\noindent\textbf{Step 3: compute the vacuum norm's dependence on $A$.}
Let $\mathrm{dvol}=dz_1\wedge\cdots\wedge dz_n$ and
$\vartheta=\iota_{\mathcal E_z}\mathrm{dvol}$. Here $\mathcal E_z=\sum_\nu z_\nu\partial_{z_\nu}$ is the holomorphic
Euler vector field (distinct from the rotation generator $E$). The linear map $z\mapsto Az$ preserves
the Euler vector field and satisfies
$A^*\mathrm{dvol}=(\det A)\mathrm{dvol}$. Thus
\[
 A^*\vartheta=(\det A)\vartheta,\qquad
 A^*\!\left(\frac{\vartheta}{F}\right)
       =(\det A)\frac{\vartheta}{F\circ A}.
\]
Taking Poincar\'e residues gives, for
$\Phi_A:X_{F\circ A}\to X_F$,
\[
 \Phi_A^*\Omega_F=(\det A)\Omega_{F\circ A}.
\]
Hence
\[
 \|\Omega_{F\circ A}\|_{\mathrm{CY}}^2
       =|\det A|^{-2}\|\Omega_F\|_{\mathrm{CY}}^2.
\]
Then \eqref{eq:LGresiduemetric} gives \eqref{eq:hGL}. Finally, \eqref{eq:TmodGL} follows from \eqref{eq:TGL} and \eqref{eq:hGL}.
\end{proof}

\section{Proof of \Cref{thm: main}}\label{sec: proof main}
By \Cref{cor: harmonic}, it suffices to analyze the asymptotic behavior
of $F_1^{\mathrm{LG}}(f)$ near the boundary of the compactified moduli space.

\subsection{Model near boundary of the compactified moduli}
In this subsection, we study the structure of the compactified moduli space and obtain a local normal form for $F$ near its boundary; see \Cref{lem:axis}.

For this subsection, let
\[
L=\binom{2n-1}{n}-1,
\qquad
U=\mathbb P H^0(\mathbb P^{n-1},\mathcal O(n))\setminus D
=\mathbb P^L\setminus D,
\]
where $D$ is the reduced discriminant hypersurface. Write
\[
\pi:\mathcal X\longrightarrow U
\]
for the universal family of smooth degree-$n$ Calabi--Yau
hypersurfaces in $\mathbb P^{n-1}$. Let
$D^{\mathrm{nod}}\subset D$ denote the locus of hypersurfaces with
exactly one ordinary double point.

The  variety
\[
\mathcal{I}:=\left\{
([F],[x])\in\mathbb P^L\times\mathbb P^{n-1}:
\partial F(x)=0
\right\}
\]
is an irreducible projective bundle over $\mathbb P^{n-1}$: for fixed $x$, the $n$
conditions
\[
\partial_{z_1}F(x)=\cdots=\partial_{z_n}F(x)=0
\]
are independent linear conditions on $F$, then the fiber is isomorphic to $\mathbb P^{L-n}$, and $\dim\mathcal{I}=(n-1)+(L-n)=L-1$. Hence its image is the
irreducible hypersurface $D$. A generic point of $D$ corresponds to a
hypersurface whose only singularity is one ordinary double point, and
\begin{equation}\label{eq:discriminantdegree}
\deg D=n(n-1)^{n-1},
\qquad
\operatorname{codim}_{\mathbb P^L}
\bigl(D\setminus D^{\mathrm{nod}}\bigr)\ge2.
\end{equation}

\begin{proposition}
\label{prop:finiteholonomy}
The positively oriented boundary $\gamma$ of a small disk transverse
to $D^{\mathrm{nod}}$ generates
\[
H_1(U,\Z)
\simeq
\Z_{n(n-1)^{n-1}}.
\]

\end{proposition}

\begin{proof}
Let $M:=\mathbb P^L\setminus(D\setminus D^{\mathrm{nod}})$. The inclusion $M\hookrightarrow\mathbb P^L$ induces
\[
H_i(M;\Z)\simeq H_i(\mathbb P^L;\Z)\quad(i=1,2),\qquad
H_2(M,U;\Z)\simeq H_2(\mathbb P^L,U;\Z),
\]
by general position, since $\mathbb P^L\setminus M$ has real codimension at least $4$.

The Thom isomorphism for $(M,U)$ and the relative homology sequence give
\[
\underbrace{H_2(M;\Z)}_{\Z}
\xrightarrow{\ \cdot n(n-1)^{n-1}\ }
\underbrace{H_2(M,U;\Z)}_{\Z}
\longrightarrow H_1(U;\Z)
\longrightarrow H_1(M;\Z)=0.
\]
The relative generator is represented by a small oriented disk transverse to $D^{\mathrm{nod}}$, whose boundary is $\gamma$. The first map is multiplication by the intersection number of a projective line with $D$, namely
\[
\deg D=n(n-1)^{n-1}.
\]
Hence
\[
H_1(U;\Z)
\simeq
\Z_{n(n-1)^{n-1}}.
\]
\end{proof}

Near a point in $D^{\mathrm{nod}}$, write $(q,b')$ for local chart with $D=\{q=0\}$. We also assume that, for each fixed $b'$, $q$ parametrizes a transverse affine pencil. Write $F_{q,b'}$ for this polynomial lift and suppress fixed $b'$.

\begin{lemma}\label{lem:axis}
There is $A_q\in U(n)$, depending smoothly on $(q,b')$, such that
\begin{equation}\label{eq:polynomial}
 \begin{split}
 (F_q\circ A_q)(w,\xi)
 =a_qw^n+w^{n-2}Q_q(\xi)
    +\sum_{j=3}^n w^{n-j}P_{j,q}(\xi),\qquad
 z=(w,\xi)\in\mathbb C\times\mathbb C^{n-1}=\C^n.
 \end{split}
\end{equation}
Here $Q_q$ is quadratic with uniformly invertible Hessian, each
$P_{j,q}$ is homogeneous of degree $j$, all coefficients are bounded,
and $c|q|\leq|a_q|\leq C|q|$ for fixed positive constants $c,C$.
Moreover,
\begin{equation}\label{eq:gradient-global}
 |\partial F_q(z)|\geq c|q||z|^{n-1}.
\end{equation}
\end{lemma}

\begin{proof}
Fix $b'$ and suppress it from the notation. Choose coordinates such
that the node of $F_0$ is $[1:0:\cdots:0]$, and set
\[
f_q(y)=F_q(1,y),\qquad
\Phi(q,y)=\partial_yf_q(y)
-\frac{nf_q(y)}{1+|y|^2}\,\bar y.
\]
Since
\[
\Phi(0,0)=0,
\qquad
D_y\Phi(0,0)=\partial_y^2f_0(0),
\]
the real implicit-function theorem gives a smooth solution
$y_q=O(|q|)$ of $\Phi(q,y_q)=0$. Set
\[
e_q=\frac{(1,y_q)}{(1+|y_q|^2)^{1/2}}.
\]
The equation $\Phi(q,y_q)=0$ and Euler's identity give
\be\label{eq: gradient}
\partial F_q(e_q)=nF_q(e_q)\,\overline{e_q}.
\ee
Choose $A_q\in U(n)$ smoothly whose first column is $e_q$, and write
$\widetilde F_q=F_q\circ A_q$. Then by \eqref{eq: gradient},
$\partial_\xi\widetilde F_q(1,0)=0$. Then the homogeneous expansion of $\widetilde F_q$
contains no term linear in $\xi$, and takes the form
\[
\widetilde F_q(w,\xi)
=a_qw^n+w^{n-2}Q_q(\xi)
+\sum_{j=3}^n w^{n-j}P_{j,q}(\xi),
\]
where $Q_q$ is quadratic in $\xi$, and $P_{j,q}$ is homogeneous of
degree $j$ in $\xi$. At $q=0$, the corresponding projective singularity is an ordinary
double point. So its transverse Hessian is nondegenerate, 
\[
\det\bigl(\partial_\xi^2Q_0\bigr)\ne0.
\]
Since the coefficients depend smoothly on the parameters, after a further
shrinking of the chart, the matrices $\partial_\xi^2Q_q$ remain
uniformly invertible, and the coefficients of $Q_q$ and
$P_{j,q}$ are uniformly bounded.

Moreover, the transversality implies
\[
a_q=F_q(e_q)
=(1+|y_q|^2)^{-n/2}f_q(y_q)
=\lambda(b')q+O(|q|^2),
\qquad \lambda(b')\ne0,
\]
where the last assertion follows from transversality to the
discriminant. Thus
\[
c|q|\le |a_q|\le C|q|.
\]

On the unit sphere and in a sufficiently small cone
$|\xi|\le\delta$ about the $w$-axis, the normal form gives
\[
|\partial_\xi\widetilde F_q(w,\xi)|\ge c|\xi|,
\qquad
\left|
\partial_w\widetilde F_q(w,\xi)-na_qw^{n-1}
\right|
\le C|\xi|^2.
\]
If $|\xi|\ge|q|$, the first estimate gives $|\partial\widetilde F_q(w,\xi)|\ge c|q|$. If $|\xi|<|q|$, the second gives
\[
|\partial_w\widetilde F_q(w,\xi)|
\ge n|a_q||w|^{n-1}-C|\xi|^2
\ge c'|q|.
\]
Outside this cone, we have a positive lower bound independent of $q$.
Since $A_q$ is unitary and $\partial F_q$ is homogeneous of degree
$n-1$, it follows that
\[
|\partial F_q(z)|\ge c|q|\,|z|^{n-1}.
\]

\end{proof}

For $A\in U(n)$, its pullback on forms is
\[
 A^*\bigl(f(z)\,dz^I\wedge \dd\bar z^J\bigr)
 =(f\circ A)(z)\,\dd(Az)^I\wedge \dd\overline{(Az)}^J.
\]
It is unitary on $L^2\Omega^{*,*}(\mathbb C^n)$, and for $\delta_F=\bar\p_F$ or $\bar\p_F^*$,
\[
 A_q^*\delta_{F_q}(A_q^*)^{-1}=\delta_{F_q\circ A_q},\qquad
 A_q^*\Delta_{F_q}(A_q^*)^{-1}=\Delta_{F_q\circ A_q}.
\]
It commutes with $N$ and the $\Z_n$ action, so the $\ph\qh$-graded heat trace is
unchanged. Hence, from now on, we may assume that\be \text{$F_q$ has the form given on the
right-hand side of \eqref{eq:polynomial}}.\ee

\subsection{Reduction of the problem}
Set
\[
u([F])=F_1^{\LG}(F)-F_1^{\mathrm{BCOV}}(X_F).
\]
By presentation invariance and Corollary~\ref{cor: harmonic}, $u$ is a single-valued real pluriharmonic function on $U$. Now we are going to show that $u$ is a constant.

{Since $U$ is the complement of the discriminant, it is enough to control $u$ near its generic nodal locus. Lemma \ref{lem:extension} below gives the
required removability criterion: on one nonempty open nodal chart, it suffices to prove
\[
\frac{1}{2\pi}\int_0^{2\pi}
|u(re^{i\theta},b')|\,d\theta
=o(r^{-1}).
\]

On a transverse nodal chart, the definition of $F_1^{\LG}$ gives
\[
u([F_q])
=-\frac12\cT^{\LG}(F_q)
-\frac{\chi(X)}{24}
 \log h_{F_q}(w_0,w_0)
-F_1^{\mathrm{BCOV}}(X_{F_q}).
\]
Lemma~\ref{lem:nodalvacuum} gives
\[
\log h_{F_q}(w_0,w_0)=O(1),
\]
whereas the nodal degeneration formula
\cite[Theorem~7.3]{EFM} gives
\begin{equation}\label{eq:cy-input}
F_1^{\mathrm{BCOV}}(X_{F_q})=O(|\log|q||).
\end{equation}
Both terms are $o(|q|^{-1})$. Therefore the only remaining estimate is the following.}

\begin{proposition}\label{prop:reduction of main theorem}
{As $q\to0$,
\[
\cT^{\LG}(F_q)=o(|q|^{-1}),
\]
uniformly in $\arg q$ and for the tangential parameter $b'$ in
compact subsets of the nodal locus.}
\end{proposition}

{We first prove the two auxiliary lemmas used in this reduction, and then turn to Proposition~\ref{prop:reduction of main theorem}.}

\begin{lemma}\label{lem:extension}
Let $u$ be a real, single-valued pluriharmonic function on $U$.
Suppose that on one nonempty open nodal chart, for each fixed $b'$,
\begin{equation}\label{eq:weak-angular}
\frac1{2\pi}\int_0^{2\pi}|u(re^{i\theta},b')|\,d\theta
          =o(r^{-1})\qquad(r \to0).
\end{equation}
Then $u$ is constant on $U$.
\end{lemma}

\begin{proof}

The form $2\partial u$ is closed and holomorphic. Since  $H_1(U,\mathbb Z)\otimes\C=0$, there is a single-valued holomorphic function $g$
on $U$ with $d g=2\partial u$. Adding a real constant to $g$ gives
$\Re g=u$.

On a punctured product chart $(q,b')$, $q\neq0$, expand
$g(q,b')=\sum_{\ell\in\mathbb Z}c_\ell(b')q^\ell$.
The coefficients, defined by contour integrals in $q$, are holomorphic
in $b'$. For every integer $k\geq1$,
\begin{equation}\label{eq:laurent}
 c_{-k}(b')+r^{2k}\overline{c_k(b')}
   =\frac{2r^k}{2\pi}\int_0^{2\pi} u(re^{i\theta},b')e^{ik\theta}\,d\theta.
\end{equation}
The right-hand side is $o(r^{k-1})$ by \eqref{eq:weak-angular}. Therefore $c_{-k}=0$ for all $k\geq1$, and $g$ extends across the chosen open subset of $D$.

Note that $D^{\mathrm{nod}}$ is connected, any two of its points can be joined by a path covered by a finite chain of overlapping product
charts. In the next chart, the vanishing result established in the previous chart
implies that the negative Laurent coefficients vanish on an open set of
tangential parameters. Repeating this argument, we know we can extend $g$ holomorphically across $D^{\mathrm{nod}}$.
Note that $D\setminus D^{\mathrm{nod}}$ has codimension at least two in
$\mathbb P(H^0(\mathbb P^{n-1},\mathcal O(n)))$, so Hartogs extension gives a holomorphic function on all of $\mathbb P(H^0(\mathbb P^{n-1},\mathcal O(n)))$. Compactness makes this function constant.
\end{proof}
We also have:
\begin{lemma}
\label{lem:nodalvacuum}
On every transverse one-node chart $(q,b')$,
\begin{equation}\label{eq:vacuum-bounded}
 0<c\leq h_{F_q}(w_0,w_0)\leq C<\infty
\end{equation}
for small $q$.
\end{lemma}
\begin{proof}
Let $K$ be a compact set of tangential parameters. All constants
below may depend on $K$, but are independent of $b'\in K$ and
of sufficiently small $q\ne0$. By \eqref{eq:LGresiduemetric}, it is enough to bound
\[
\int_{X_{F_q}}i^{(n-2)^2}
\Omega_{F_q}\wedge\overline{\Omega_{F_q}}.
\]

The holomorphic Morse lemma with parameters gives, after shrinking a
neighborhood of $\{0\}\times K$, coordinates
$y=(y_1,\ldots,y_{n-1})$ centered at the node in which
\[
X_{F_q}=\Big\{\sum_{j=1}^{n-1}y_j^2=t(q,b')\Big\},
\qquad
t(q,b')=c_{b'}q+O(q^2),
\qquad c_{b'}\ne0.
\]
The function $c(b')$ is continuous and nonzero on $K$. Compactness
of $K$ gives constants $c_0,C_0>0$ such that
\[
c_0|q|\le |t(q,b')|\le C_0|q|.
\]
Choose $\delta>0$ so that these coordinates are defined on
$B_{2\delta}=\{|y|<2\delta\}$ uniformly in the parameters.
For $t=t(q,b')$, set
\[
Q_t=\Big\{y\in\mathbb C^{n-1}:\sum_{j=1}^{n-1}y_j^2=t\Big\},
\qquad
\rho_t=
\operatorname{Res}_{Q_t}
\frac{dy_1\wedge\cdots\wedge dy_{n-1}}
     {\sum_{j=1}^{n-1}y_j^2-t},
\]
In these coordinates the holomorphic volume form is
\[
\Omega_{F_q}=a(q,b',y)\rho_t,
\qquad\text{for some}\quad
0<c\le |a(q,b',y)|\le C.
\]

The quadratic equation suggests the scale $|t|^{1/2}$, which is comparable to $|q|^{1/2}$. For
\[
S_r:Q_{t/r^2}\longrightarrow Q_t,
\qquad
S_r(x)=rx,
\]
the residue form satisfies
\[
S_r^*\rho_t
=r^{n-3}\rho_{t/r^2}.
\]
Fix $A>0$ sufficiently large that $C_0/A^2<1/2$, and set
\[
r_0=A|q|^{1/2}.
\]
We can assume $|p|$ is small enough so that $r_0<\delta$.

First consider the inner region $|y|\le r_0$.
Under $y=r_0x$, the rescaled parameter satisfies
\[
\frac{c_0}{A^2}
\le\left|\frac{t}{r_0^2}\right|
\le\frac{C_0}{A^2}.
\]
Thus it ranges over a compact set disjoint from zero.
Moreover, on $Q_{t/r_0^2}$,
\[
|x|^2
\geq\Big|\sum_{j=1}^{n-1}x_j^2\Big|
=\left|\frac{t}{r_0^2}\right|
\ge\frac{c_0}{A^2}.
\]
Then the defining gradient $2x$ is uniformly bounded
away from zero. Thus
\[
\int_{Q_{t/r_0^2}\cap\{|x|\le1\}}
i^{(n-2)^2}\rho_{t/r_0^2}
\wedge\overline{\rho_{t/r_0^2}}
\le C.
\]
Using the scaling identity, we obtain
\[
\int_{X_{F_q}\cap B_{r_0}}
i^{(n-2)^2}\Omega_{F_q}\wedge\overline{\Omega_{F_q}}
\le C r_0^{2n-6}
\le C|q|^{n-3}.
\]

Next, let $r_0\le r\le\delta$.
Similarly,
\[
\int_{X_{F_q}\cap\{r\le|y|\le2r\}}
i^{(n-2)^2}\Omega_{F_q}\wedge\overline{\Omega_{F_q}}
\le Cr^{2n-6}.
\]

Choose $J\ge0$ so that
\[
2^Jr_0\le\delta<2^{J+1}r_0,
\qquad r_j=2^jr_0.
\]
The  ball $B_{r_0}$ and the annuli
$\{r_j\le|y|\le2r_j\}$, $0\le j\le J$, cover $B_\delta$.
Since $n\ge4$, we have
\[
\sum_{j=0}^Jr_j^{2n-6}
\le
\frac{r_J^{2n-6}}{1-2^{-(2n-6)}}
\le C\delta^{2n-6}.
\]
Thus
\[
\int_{X_{F_q}\cap B_\delta}
i^{(n-2)^2}\Omega_{F_q}\wedge\overline{\Omega_{F_q}}
\le C\bigl(|q|^{n-3}+\delta^{2n-6}\bigr)
\le C.
\]

Away from $B_\delta$, the family remains smooth at $q=0$,
and the residue forms converge smoothly. This gives a uniform
upper bound for the integral over the complement of $B_\delta$.

The integral over \(X_{F_q}\cap B_\delta^c\) also gives a uniform lower bound.
\end{proof}

\subsection{Proof of \Cref{prop:reduction of main theorem}}

In this subsection, $C,c,\ldots$ denote constants that are, in particular,
independent of $q$. Their values may vary from one occurrence to another.
We fix a sufficiently small constant $q_0>0$ and assume that
\[
0< |q|\leq q_0.
\]

First, we have the following spectral gap.

{\begin{lemma}\label{lem:spectral-gap}
Every nonzero eigenvalue $\lambda$ of $\Delta_{F_q}$ satisfies
\be\label{eq: spectral gap}
\lambda\ge c|q|^{4/n}.
\ee
\end{lemma}}

\begin{proof}
We can check directly that 
\begin{equation}\label{eq: gradient Fq}
\begin{gathered}
 |\partial_zF_q(z)|\ge c|q||z|^{n-1}\quad\text{and}\quad
  |\partial^2 F_q(z)|\le C|z|^{n-2}.
\end{gathered}
\end{equation}
Since $n\ge4$ and $|q|<1$, \eqref{eq: gradient Fq} implies, for a sufficiently
large fixed $R$, as long as $|z|\ge R|q|^{-2/n},$
\[
 2|\partial_zF_q(z)|^2-2|L(F_q''(z))|
 \ge c|q|^2|z|^{2n-2}-C|z|^{n-2}\ge |z|^2.
\]

So, if $ \operatorname{supp}u\subset\{|z|\ge R|q|^{-2/n}\}$, we have
\[
\langle\Delta_{F_q}u,u\rangle\ge\||z|u\|^2.
\]

Apply the IMS localization argument (see \cite[\S 3.1]{CFKS1987}) associated with the cover $|z|<2R|q|^{-2/n}$ and
$|z|> R|q|^{-2/n}$, we have
\begin{equation}\label{eq: gap z square}
 \||z|u\|^2\le\langle\Delta_{F_q}u,u\rangle
                         +C|q|^{-4/n}\|u\|^2.
\end{equation}
The K\"ahler identities give
\begin{equation}\label{eq: gap algebra}
\begin{gathered}
 \Big[\bar{\partial}_{F_q},
       -\frac i2\sum_j\bar z_j\,dz_j\wedge\Big]
       =\omega\wedge,\\
 [\Delta_{F_q},\omega\wedge]=0,
 \qquad [ (\omega\wedge)^*,\omega\wedge]=n-N,
\end{gathered}
\end{equation}
where $\omega=:-\frac {\sqrt{-1}}2\sum_jd\bar z_j\wedge dz_j$ and $[\cdot,\cdot]$ is the super-commutator.
\smallskip

Let $u$ be a degree-$k$ eigenform with eigenvalue $\lambda.$
\begin{itemize}
    \item If $k<n$,
using $(dz_j\wedge)^*=2\iota_{\partial_{z_j}}$, \eqref{eq: gap z square} and \eqref{eq: gap algebra}, we obtain 
\begin{equation}\label{eq: gap pair}
\begin{aligned}
 &\quad(n-k)\|u\|^2
 \le\|\omega\wedge u\|^2\\
 &=\left\langle
       -\frac i2\sum_j\bar z_j\,dz_j\wedge u,
       \bar{\partial}_{F_q}^*(\omega\wedge u)
      \right\rangle+\left\langle
       \bar{\partial}_{F_q}u,
       i\sum_j z_j\,\iota_{\partial_{z_j}}(\omega\wedge u)
      \right\rangle\\
 &\le C\sqrt\lambda\,\||z|u\|\,\|u\|\le C\sqrt{\lambda(\lambda+|q|^{-4/n})}\,\|u\|^2.
\end{aligned}
\end{equation}
If $k>n$, using the adjoint of the first identity in
\eqref{eq: gap algebra} and 
$\|(\omega\wedge)^*u\|^2\ge(k-n)\|u\|^2$, we derive estimate of the form \eqref{eq: gap pair} again. Thus 
\[\lambda\geq c|q|^{4/n},\qquad \text{if }~k\neq n.\]
     \item If $k=n$ and $\lambda>0$, then either $\bar{\partial}_{F_q}u$ or $\bar{\partial}_{F_q}^*u$ is non-zero eigenform.
Thus we also have 
\[\lambda\geq c|q|^{4/n}.\]
\end{itemize}

\end{proof}

{To estimate the radial variation of the heat trace, we treat the $q$-variation near the node as an infinitesimal change of variables as follows.} Now let $0<\varepsilon\ll1$ be a fixed sufficiently small constant.
\begin{lemma}\label{eq: variation lemma1}
There is a holomorphic degree-one vector field $V_q$ on
$\{|\xi|\le 2\varepsilon|w|\}$  such that
\begin{equation}\label{eq: vector bound}
 V_qF_q=|q|\partial_{|q|}F_q,\qquad
 |\nabla^jV_q(z)|\le C_j|z|^{1-j}.
\end{equation}
Here for $q=q_1+\sqrt{-1}q_2$, 
\[
|q|\partial_{|q|}
:=
q_1\partial_{q_1}+q_2\partial_{q_2}.
\]
\end{lemma}

{\begin{proof}
By Lemma \ref{lem:axis}, under some coordinate, $F_q$ is of the form below:
\[
F_q(w,\xi)=a_qw^n+R_q(w,\xi),
\qquad
R_q(w,\xi)
=w^{n-2}Q_q(\xi)
+\sum_{j=3}^n w^{n-j}P_{j,q}(\xi).
\]
Then on
$\{|\xi|\le2\varepsilon|w|\}$,
\[
\int_0^1\partial_\xi^2R_q(w,s\xi)\,ds=w^{n-2}\bigl(\partial_\xi^2Q_q+O(\varepsilon)\bigr),
\]
which is invertible for sufficiently small $\varepsilon$. Define
\[
(V_q)_w=\frac{|q|\partial_{|q|}a_q}{na_q}w
\]
and
\[
(V_q)_\xi
=
\left(\int_0^1\partial_\xi^2R_q(w,s\xi)\,ds\right)^{-1}
\int_0^1\partial_\xi
\left(
|q|\partial_{|q|}R_q-
\frac{|q|\partial_{|q|}a_q}{na_q}w\partial_wR_q
\right)(w,s\xi)\,ds .
\]
Then we have
\[
V_qF_q=|q|\partial_{|q|}F_q.
\]
Since $a_q=cq+O(q^2)$, the coefficients of $V_q$ are uniformly bounded.
The homogeneity of $V_q$ then yields
\[
|\nabla^jV_q(z)|\le C_j|z|^{1-j}.
\]
\end{proof}}

{Since $V_q$ is defined only on the nodal cone
$\{|\xi|\le2\epsilon|w|\}$, we introduce the following cutoff function $\chi$ to obtain
a globally defined deformation that agrees with $V_q$ on the smaller cone.}
Let $\chi$ be a $\Z_n$-invariant cutoff function, such that
\begin{enumerate}[(i)]
    \item $\chi(z)=1$ if $z\in\{z'=(w,\xi)\in \C^n:|\xi|<\varepsilon|w|,\ |z'|>2\}$.
    \item $\chi(z)=0$ if $z\in\{z'=(w,\xi)\in\C^n:|\xi|>2\varepsilon|w|\}\cup\{z'\in\C^n:|z'|<1\}$
    \item It is homogeneous of degree zero for $|z|>2$.
\end{enumerate}

Set
\begin{equation}\label{eq: defn K}
\begin{split}
 K_{i\bar j}
 &=\partial_{z_i}\partial_{\bar{z}_j}
\left(\chi\operatorname{Re}\sum_{\ell=1}^n\bar
z_\ell(V_q)_\ell\right),\\
B(K)&=\sum_{i,j}K_{i\bar j}\left(dz_i\wedge\iota_{\partial_{z_j}}
       +d\bar z_j\wedge\iota_{\partial_{\bar z_i}}
       -\delta_{ij}\right),
\end{split}
\end{equation}
where for a vector field $V$, $(V)_\ell$ denotes its component corresponding to $z_\ell.$
Then it is clear that
\[|B(K)|\le C.\]

{ 
Here and below $[\cdot,\cdot]$ is still the super-commutator.

}

\begin{lemma}\label{eq: variation lemma}
We have
\begin{equation}\label{eq: heat variation Fq}
\begin{split}
 &\frac{|q|}{2}\partial_{|q|}
       \Str \bigl((N-n)^2e^{-t\Delta_{F_q}}\bigr)=t\partial_t\Str 
          \bigl(B(K)(N-n)e^{-t\Delta_{F_q}}\bigr)-\frac{t}{2}\Str ((N-n)^2\mathcal E_qe^{-t\Delta_{F_q}}).
\end{split}
\end{equation}
Here $\mathcal E_q$ has the form
\[\mathcal E_q
=[D_q^+,\bar\partial_{F_q}^*]
+[\bar\partial_{F_q},D_q^-],\]
where $D_q^\pm$ are differential operators that have differential order at most one, degree $\pm1$, and support in
\begin{equation}\label{eq: error support}
 \{|z|\le2\}\cup\{|\xi|\ge\varepsilon|w|\}.
\end{equation}
 Moreover, if we write
$\mathcal E_q=\sum_{|\nu|\le2}e_{\nu,q}\partial^\nu$, then when $|z|\ge2$,
\begin{equation}\label{eq: error coefficient}
 |\partial^\alpha e_{\nu,q}(z)|
 \le C_\alpha|z|^{(n-1)(2-|\nu|)-|\alpha|}.
\end{equation}
\end{lemma}

{
\begin{proof}
On $\{\chi=1\}$, we have
\[
[\mathcal L_{V_q+\bar V_q},\bar\partial_{F_q}]
=
|q|\partial_{|q|}\bar\partial_{F_q},
\qquad
\mathcal L_{V_q+\bar V_q}
+\mathcal L_{V_q+\bar V_q}^*
=
2B(K).
\]
Define
\[
D_q^+
:=
|q|\partial_{|q|}\bar\partial_{F_q}
-
[\mathcal L_{\chi(V_q+\bar V_q)},\bar\partial_{F_q}]
\]
and
\[
D_q^-
:=
(D_q^+)^*
-
\Big[
 \mathcal L_{\chi(V_q+\bar V_q)}^*
 +\mathcal L_{\chi(V_q+\bar V_q)}
 -2B(K),
 \bar\partial_{F_q}^*
\Big].
\]
Then $D_q^\pm$ vanish on $\{\chi=1\}$, have order at most one,
and are supported in
\[
\{|z|\le2\}\cup\{|\xi|\ge\varepsilon|w|\}.
\]
Taking adjoint of $D_q^+$ gives
\[
(D_q^+)^*
=
|q|\partial_{|q|}\bar\partial_{F_q}^*
+
[\mathcal L_{\chi(V_q+\bar V_q)}^*,\bar\partial_{F_q}^*],
\]
and hence
\[
|q|\partial_{|q|}\bar\partial_{F_q}^*
=
[\mathcal L_{\chi(V_q+\bar V_q)},\bar\partial_{F_q}^*]
-
2[B(K),\bar\partial_{F_q}^*]
+
D_q^-.
\]
Differentiating
$\Delta_{F_q}=[\bar\partial_{F_q},\bar\partial_{F_q}^*]$
therefore yields
\begin{equation}\label{eq:qDelta}
\begin{aligned}
|q|\partial_{|q|}\Delta_{F_q}
&=
[\mathcal L_{\chi(V_q+\bar V_q)},\Delta_{F_q}]
-
2[\bar\partial_{F_q},
   [B(K),\bar\partial_{F_q}^*]]
+
\mathcal E_q,
\end{aligned}
\end{equation}
where
\[
\mathcal E_q
=
[D_q^+,\bar\partial_{F_q}^*]
+
[\bar\partial_{F_q},D_q^-].
\]

Set
\[
L_K
:=
\sum_{i,j}K_{i\bar j}\,dz_i\wedge d\bar z_j\wedge,
\qquad
\Lambda
:=
\sum_{\ell}
\iota_{\partial_{\bar z_\ell}}
\iota_{\partial_{z_\ell}}.
\]
Then
\be\label{eq: Bk Lk}
B(K)=[L_K,\Lambda].
\ee
Using
\be\label{eq: Lambda partial}
[\Lambda,\bar\partial_{F_q}]
=\frac12\partial_{F_q}^*,
\qquad
[\Lambda,\bar\partial_{F_q}^*]=0,
\qquad
[\partial_{F_q}^*,\bar\partial_{F_q}^*]=0,
\ee
we obtain
\[
\begin{aligned}
[\Lambda,\Delta_{F_q}]
=0.
\end{aligned}
\]
Thus
\be\label{eq: L commute with heat}
[\Lambda,e^{-t\Delta_{F_q}}]=0.
\ee
Moreover,
\be\label{eq: Lambda comute}
\begin{aligned}
&\bigl[
 \Lambda,
 \bar\partial_{F_q}^*\bar\partial_{F_q}
 -
 \bar\partial_{F_q}\bar\partial_{F_q}^*
\bigr]
=
\frac12\bar\partial_{F_q}^*\partial_{F_q}^*
-
\frac12\partial_{F_q}^*\bar\partial_{F_q}^*
=
\bar\partial_{F_q}^*\partial_{F_q}^*.
\end{aligned}
\ee
Cyclicity, \eqref{eq: Bk Lk}, \eqref{eq: L commute with heat} and \eqref{eq: Lambda comute} give
\be\label{eq: lemma eq1}
\begin{aligned}
&\quad\Str\!\left(
 B(K)
 \bigl(
  \bar\partial_{F_q}^*\bar\partial_{F_q}
  -
  \bar\partial_{F_q}\bar\partial_{F_q}^*
 \bigr)e^{-t\Delta_{F_q}}
\right)
=
\Str\!\left(
 L_K
 \Big[
  \Lambda,
  \bigl(
   \bar\partial_{F_q}^*\bar\partial_{F_q}
   -
   \bar\partial_{F_q}\bar\partial_{F_q}^*
  \bigr)e^{-t\Delta_{F_q}}
 \Big]
\right)
\\
&=
\Str\!\left(
 L_K\bar\partial_{F_q}^*\partial_{F_q}^*
 e^{-t\Delta_{F_q}}
\right).
\end{aligned}
\ee
Using $[\Lambda,N-n]=2\Lambda$, \eqref{eq: Bk Lk} and cyclicity, we have
\be\label{eq:lemma eq2}
\begin{aligned}
&\partial_t
\Str\!\left(
 B(K)(N-n)e^{-t\Delta_{F_q}}
\right)
=
-\Str\!\left(
 [L_K,\Lambda](N-n)\Delta_{F_q}e^{-t\Delta_{F_q}}
\right)
\\
&=
-\Str\!\left(
 L_K
 [\Lambda,(N-n)\Delta_{F_q}e^{-t\Delta_{F_q}}]
\right)
=
-2\Str\!\left(
 L_K\Lambda\Delta_{F_q}e^{-t\Delta_{F_q}}
\right).
\end{aligned}
\ee

By the definition of $K$,
$
\partial_{z_a}K_{i\bar j}
=
\partial_{z_i}K_{a\bar j},
$
so
\be\label{eq: partial LK}
[\partial_{F_q},L_K]=0.
\ee
Also,
the fact that $
\partial_{F_q}^2=0$ and $
[\partial_{F_q},\partial_{F_q}^*]=\Delta_{F_q}
$  implies
\be\label{eq: partial commute with Delta}[\partial_{F_q},e^{-t\Delta_{F_q}}]=0.
\ee
The identity
$
[\Lambda,\partial_{F_q}]
=-\frac12\bar\partial_{F_q}^*
$
gives
\be\label{eq: partial Lambda}
\partial_{F_q}\Lambda
=
\Lambda\partial_{F_q}
+\frac12\bar\partial_{F_q}^*.
\ee
Then by  \eqref{eq: partial LK}-\eqref{eq: partial Lambda},
\be\label{eq: lemma eq3}
\begin{aligned}
0
&=
\Str\!\left(
 \big[
  \partial_{F_q},
  L_K\Lambda\partial_{F_q}^*e^{-t\Delta_{F_q}}
 \big]
\right)
\\
&=
\Str\!\left(
 L_K
 \bigl(
  \partial_{F_q}\Lambda\partial_{F_q}^*
  +
  \Lambda\partial_{F_q}^*\partial_{F_q}
 \bigr)e^{-t\Delta_{F_q}}
\right)
=
\Str\!\left(
 L_K\Lambda
 \bigl(
  \partial_{F_q}\partial_{F_q}^*
  +
  \partial_{F_q}^*\partial_{F_q}
 \bigr)e^{-t\Delta_{F_q}}
\right)
\\
&\quad+
\frac12\Str\!\left(
 L_K\bar\partial_{F_q}^*\partial_{F_q}^*
 e^{-t\Delta_{F_q}}
\right)
=
\Str\!\left(
 L_K\Lambda\Delta_{F_q}e^{-t\Delta_{F_q}}
\right)
+
\frac12\Str\!\left(
 L_K\bar\partial_{F_q}^*\partial_{F_q}^*
 e^{-t\Delta_{F_q}}
\right).
\end{aligned}
\ee
Combining \eqref{eq: lemma eq1}-\eqref{eq: lemma eq3} yields
\begin{equation}\label{eq:BK-time-identity}
\begin{aligned}
&\Str\!\left(
 B(K)
 \bigl(
  \bar\partial_{F_q}^*\bar\partial_{F_q}
  -
  \bar\partial_{F_q}\bar\partial_{F_q}^*
 \bigr)e^{-t\Delta_{F_q}}
\right)
=
\partial_t
\Str\!\left(
 B(K)(N-n)e^{-t\Delta_{F_q}}
\right).
\end{aligned}
\end{equation}

We also have
\[
\begin{aligned}
0
&=
\Str\!\left(
 \left[
  \bar\partial_{F_q},
  \frac12(N-n)^2
  [\bar\partial_{F_q}^*,B(K)]
  e^{-t\Delta_{F_q}}
 \right]
\right)
\\
&=
-\Str\!\left(
 \left(N-n-\frac12\right)
 \bar\partial_{F_q}
 [\bar\partial_{F_q}^*,B(K)]
 e^{-t\Delta_{F_q}}
\right)
+
\Str\!\left(
 \frac12(N-n)^2
 [\bar\partial_{F_q},
  [\bar\partial_{F_q}^*,B(K)]]
 e^{-t\Delta_{F_q}}
\right).
\end{aligned}
\]
Moving the odd operator
$\bar\partial_{F_q}^*e^{-t\Delta_{F_q}}$
around the supertrace gives
\[
\begin{aligned}
&\quad\Str\!\left(
 \left(N-n-\frac12\right)
 \bar\partial_{F_q}B(K)\bar\partial_{F_q}^*
 e^{-t\Delta_{F_q}}
\right)
=
-\Str\!\left(
 B(K)\bar\partial_{F_q}^*
 \left(N-n-\frac12\right)
 \bar\partial_{F_q}e^{-t\Delta_{F_q}}
\right)
\\
&=
-\Str\!\left(
 B(K)\left(N-n+\frac12\right)
 \bar\partial_{F_q}^*\bar\partial_{F_q}
 e^{-t\Delta_{F_q}}
\right).
\end{aligned}
\]
Therefore,
\begin{equation}\label{eq:strBk}
\begin{aligned}
&\quad\Str\!\left(
 \frac12(N-n)^2
 [\bar\partial_{F_q},
  [\bar\partial_{F_q}^*,B(K)]]
 e^{-t\Delta_{F_q}}
\right)
=
\Str\!\left(
 \left(N-n-\frac12\right)
 \bar\partial_{F_q}
 [\bar\partial_{F_q}^*,B(K)]
 e^{-t\Delta_{F_q}}
\right)
\\
&=
\Str\!\left(
 B(K)\left(N-n-\frac12\right)
 \bar\partial_{F_q}\bar\partial_{F_q}^*
 e^{-t\Delta_{F_q}}
\right)+
\Str\!\left(
 B(K)\left(N-n+\frac12\right)
 \bar\partial_{F_q}^*\bar\partial_{F_q}
 e^{-t\Delta_{F_q}}
\right)
\\
&=
\Str\!\left(
 B(K)(N-n)\Delta_{F_q}e^{-t\Delta_{F_q}}
\right)
+
\frac12\Str\!\left(
 B(K)
 \bigl(
  \bar\partial_{F_q}^*\bar\partial_{F_q}
  -
  \bar\partial_{F_q}\bar\partial_{F_q}^*
 \bigr)e^{-t\Delta_{F_q}}
\right)
\\
&=
-\frac12\partial_t
\Str\!\left(
 B(K)(N-n)e^{-t\Delta_{F_q}}
\right),
\end{aligned}
\end{equation}
where the last equality uses \eqref{eq:BK-time-identity}.

Moreover, $[(N-n)^2,\Delta_{F_q}]=0$ and cyclicity give
\begin{equation}\label{eq:N2L}
\begin{aligned}
&\Str\!\left(
 (N-n)^2
 [\mathcal L_{\chi(V_q+\bar V_q)},\Delta_{F_q}]
 e^{-t\Delta_{F_q}}
\right)
=0.
\end{aligned}
\end{equation}

Finally, Duhamel's formula and
$[(N-n)^2,e^{-s\Delta_{F_q}}]=0$ yield
\[
\begin{aligned}
&\frac{|q|}{2}\partial_{|q|}
\Str\!\left((N-n)^2e^{-t\Delta_{F_q}}\right)
=
-\frac{t}{2}
\Str\!\left(
 (N-n)^2
 \bigl(|q|\partial_{|q|}\Delta_{F_q}\bigr)
 e^{-t\Delta_{F_q}}
\right).
\end{aligned}
\]

Substituting \eqref{eq:qDelta} and using
\eqref{eq:strBk}--\eqref{eq:N2L}, we obtain
\[
\begin{aligned}
&\frac{|q|}{2}\partial_{|q|}
\Str\!\left((N-n)^2e^{-t\Delta_{F_q}}\right)
\\
&=
-t\Str\!\left(
 (N-n)^2
 [\bar\partial_{F_q},
  [\bar\partial_{F_q}^*,B(K)]]
 e^{-t\Delta_{F_q}}
\right)
-
\frac{t}{2}\Str\!\left(
 (N-n)^2\mathcal E_qe^{-t\Delta_{F_q}}
\right)
\\
&=
t\partial_t
\Str\!\left(
 B(K)(N-n)e^{-t\Delta_{F_q}}
\right)
-
\frac{t}{2}\Str\!\left(
 (N-n)^2\mathcal E_qe^{-t\Delta_{F_q}}
\right).
\end{aligned}
\]

Lastly, the support and coefficient estimates for $\mathcal E_q$ follow
from \Cref{eq: variation lemma1}, the homogeneity of $\chi$, and the homogeneity of
$F_q$.
\end{proof}}

Motivated by \Cref{eq: variation lemma}, we need to determine the constant term of
$\Str\bigl(B(K)(N-n)e^{-t\Delta_{F_q}}\bigr)$
and estimate
$\int_0^1 \Str\bigl(\tfrac12(N-n)^2\mathcal E_q
e^{-t\Delta_{F_q}}\bigr)\,dt$.
These are established in \Cref{lem: coefficient lemma} and
\Cref{lem: local heat fp}, respectively.
\begin{lemma}\label{lem: coefficient lemma}
For the matrix $K$ constructed above,
\begin{equation}\label{eq: coefficient bounds BK}
 \left|
 \left[
 \Str \bigl(B(K)(N-n)e^{-t\Delta_{F_q}}\bigr)
 \right]_{t^0}
 \right|\le C.
\end{equation}
Moreover,
\begin{equation}\label{eq: zeta zero}
 \zeta_{F_q}(0)
 =
 \left[
 \tfrac12\Str 
 \bigl((N-n)^2e^{-t\Delta_{F_q}}\bigr)
 \right]_{t^0}
\end{equation}
is independent of $q$. Here $[\cdot]_{t^0}$ means taking the constant term. The asymptotic expansions of supertraces above are of the form
{
\[
\sum_{\alpha<0}a_\alpha t^\alpha
      +a_0+O(t^{\delta}),
 \qquad \delta>0,
\]}
where the sum is finite.
\end{lemma}

\begin{proof}
Fix $q\ne0$. Proceeding as in the proof of \Cref{prop:smalltime}, only
\[(k,\ell)=(2n-2,0),\,(2n,0),\,(2n-2,2)\] 
terms contribute through order zero of pointwise supertrace of the operator considered in \eqref{eq: coefficient bounds BK}, the remaining terms are $O_q(t)$.

\begin{enumerate}[(i)]
    \item For $(k,\ell)=(2n-2,0)$, the contribution is
\begin{equation}\label{eq: sgort contri 1}
\begin{aligned}
 &\quad\frac{2(-4)^{n-1}t^{n-2}}{(2\pi)^n}
 \sum_{i,j=1}^n(-1)^{i+j}
 \int_{\mathbb C^n}
 e^{-2t|\partial_zF_q|^2}K_{i\bar j}
 \det\!\bigg[
 \Big(
 \sum_{r=1}^n
 \partial_{z_a}\partial_{z_r}F_q\,
 \overline{\partial_{z_b}\partial_{z_r}F_q}
 \Big)_{\substack{a\ne i\\b\ne j}}
 \bigg]\,dz\\
 &
 =c_{2n-2,0}(q)t^{-n/(n-1)}+O_q(t^{n-2}).
\end{aligned}
\end{equation}
    \item For $(k,\ell)=(2n,0)$, the contribution is
\begin{equation}\label{eq: sgort contri 2}
\begin{aligned}
 &\frac{2(-4)^{n-1}t^n}{3(2\pi)^n}
 \int_{\mathbb C^n}e^{-2t|\partial_zF_q|^2}
 \biggl[
 \Tr\!\left(\partial^2F_q(\partial^2F_q)^*\right)
 \sum_{i,j=1}^n(-1)^{i+j}K_{i\bar j}\\
 &\qquad\times
 \det\!\bigg[
 \Big(
 \sum_{r=1}^n
 \partial_{z_a}\partial_{z_r}F_q\,
 \overline{\partial_{z_b}\partial_{z_r}F_q}
 \Big)_{\substack{a\ne i\\b\ne j}}
 \bigg]
 -\Tr(K)|\det\partial^2F_q|^2
 \biggr]\,dz\\
 &\qquad
 =c_{2n,0}(q)+O_q(t^n).
\end{aligned}
\end{equation}
     \item For $(k,\ell)=(2n-2,2)$, integration by parts gives
\begin{equation}\label{eq: sgort contri 3}
\begin{aligned}
 &\frac{2(-4)^{n-1}}{(2\pi)^n}
 \sum_{i,j=1}^n(-1)^{i+j}
 \int_{\mathbb C^n}e^{-2t|\partial_zF_q|^2}
 \biggl[
 -\frac{t^{n-1}}{12}\Delta_{\bar\p}K_{i\bar j}
 -\frac{t^n}{3}K_{i\bar j}
 \Tr\!\left(\partial^2F_q(\partial^2F_q)^*\right)
 \biggr]\\
 &\qquad\times
 \det\!\bigg[
 \Big(
 \sum_{r=1}^n
 \partial_{z_a}\partial_{z_r}F_q\,
 \overline{\partial_{z_b}\partial_{z_r}F_q}
 \Big)_{\substack{a\ne i\\b\ne j}}
 \bigg]\,dz=c_{2n-2,2}(q)+O_q(t^{n-1}),
\end{aligned}
\end{equation}
where we used
$
 -\Delta_{\bar\partial}(4|\partial_zF_q|^2)
 =
 8\Tr\!\left(\partial^2F_q(\partial^2F_q)^*\right).
$
\end{enumerate}

To verify the powers in
\eqref{eq: sgort contri 1}--\eqref{eq: sgort contri 3},
we use the homogeneity of $F_q$ and the degree-zero homogeneity
of $K$ for $|z|\ge2$.
A term indexed by $(k,\ell)$ has 
$t^{k+\ell/2-n}$ and spatial degree $k(n-2)-\ell$;
spatial integration therefore gives the power
\[
 t^{\,k+\ell/2-n
       -\frac{2n+k(n-2)-\ell}{2(n-1)}}
 =
 t^{\,\frac{n(k+\ell-2n)}{2(n-1)}}.
\]
Thus the three powers are, respectively,
$t^{-n/(n-1)}$, $t^0$, and $t^0$.
While the compact region $|z|\le2$ contributes the errors
$O(t^{n-2})$, $O(t^n)$, and $O(t^{n-1})$
in the three formulas, uniformly in $q$.

The terms containing
$\Tr(\partial^2F_q(\partial^2F_q)^*)$
in \eqref{eq: sgort contri 2} and
\eqref{eq: sgort contri 3} cancel.
For the remaining determinant term, the area formula gives
\begin{equation}\label{eq: coe degree}
 t^n\int_{\mathbb C^n}
 e^{-2t|\partial_zF_q|^2}
 |\det\partial^2F_q|^2\,dz
 =
 (n-1)^n(\pi/2)^n.
\end{equation}
Since $\|K\|_\infty\le C$, this term is uniformly bounded.

For the remaining term in
\eqref{eq: sgort contri 3}, we have
\[
 \operatorname{supp}(\Delta_{\bar\p}K)
 \subset
 \{1\le|z|\le2\}
 \cup
 \{|z|\ge2,\ \varepsilon|w|\le|\xi|
  \le2\varepsilon|w|\}.
\]
On the second set,
$|\partial_zF_q|\ge c|z|^{n-1}$ and
$|\Delta_{\bar\p}K|\le C|z|^{-2}$.
Each minor in \eqref{eq: sgort contri 3}
is bounded by $C|z|^{2(n-1)(n-2)}$, so
\begin{equation}\label{eq: lap K bound}
\begin{aligned}
 &t^{n-1}\int_{\mathbb C^n}
 e^{-2t|\partial_zF_q|^2}
 |\Delta_{\bar\p}K|\,
 |z|^{2(n-1)(n-2)}\,dz\\
 &\qquad\le
 Ct^{n-1}
 +Ct^{n-1}\int_{|z|\ge2}
 e^{-ct|z|^{2n-2}}
 |z|^{2(n-1)(n-2)-2}\,dz
 \le C
 \qquad(0<t\le1).\end{aligned}
\end{equation}
Here the $Ct^{n-1}$ term in the second line comes from integration over $|z|\leq 2$.
Equations \eqref{eq: coe degree} and
\eqref{eq: lap K bound} therefore bound
$c_{2n,0}(q)+c_{2n-2,2}(q)$ uniformly in $q$.
For $g\in\Z_n\setminus\{1\}$, the only fixed point is zero,
where $K$ vanishes; hence the corresponding contribution is
$O_q(t^M)$ for every $M$.

Finally, take $K=I$, so that $B(K)=N-n$ and
$\Delta_{\bar\p}K=0$.
The cancellation above and \eqref{eq: coe degree} give
\[
 \left[
 \tfrac12\Str_{L^2\Omega^{*,*}(\C^n)}\bigl((N-n)^2e^{-t\Delta_{F_q}}\bigr)
 \right]_{t^0}
 =
 \frac{(-1)^n n(n-1)^n}{12}.
\]
For $g\ne1$, the constant coefficient is the free
fixed-point coefficient, and therefore
\begin{equation}\label{eq: zeta zero-value}
\begin{aligned}
 &\left[
 \tfrac12\Str 
 \bigl((N-n)^2e^{-t\Delta_{F_q}}\bigr)
 \right]_{t^0}=
 \frac{(-1)^n(n-1)^n}{12}
 +
 \frac1{2n}
 \sum_{\substack{g\in\Z_n\\g\ne1}}
 \frac{\str\bigl(g^*(N-n)^2\bigr)}
      {|\det_{\mathbb R}(1-g)|},
\end{aligned}
\end{equation}
which is independent of $q$.
\end{proof}

Let $\FP\int_0^1 h(t)dt$ denote the Hadamard finite part at $t=0$.
More precisely, if
\be\label{eq: only rational power}
 h(t)=\sum_{\alpha<-1}a_\alpha t^\alpha
      +a_{-1}t^{-1}+O(t^{-1+\delta}),
 \qquad \delta>0,
\ee
where the sum is finite, we define
\[
 \FP\int_0^1 h(t)\,dt
 :=
 \lim_{\epsilon\to0^+}
 \Big[
   \int_\epsilon^1 h(t)\,dt
   +\sum_{\alpha<-1}
      \frac{a_\alpha}{\alpha+1}\epsilon^{\alpha+1}
   +a_{-1}\log\epsilon
 \Big].
\]

Let $H=L^2\Omega^\bullet(\mathbb C^n)$. For operators on $H$, we use $\|\cdot\|$, 
$\|\cdot\|_{\mathrm{HS}}$, and $\|\cdot\|_{\mathrm{tr}}$ to denote, respectively, the operator norm, the Hilbert--Schmidt norm, and the trace norm. More precisely,
\begin{itemize}
    \item for $L\in\mathcal B(H)$, 
\[
\|L\|=\sup_{\|u\|_{H}=1}\|Lu\|_{H};
\]
    \item for a Hilbert--Schmidt operator $L$, choose any
orthonormal basis $\{e_j\}_{j\ge1}$ of $H$,
\[
\|L\|_{\mathrm{HS}}^2
=\sum_{j=1}^{\infty}\|Le_j\|_{H}^2=\Tr(L^*L);
\]    
    \item for a trace-class operator $L$,
\[
\|L\|_{\mathrm{tr}}
=\Tr\bigl((L^*L)^{1/2}\bigr)
=\sum_{j=1}^{\infty}s_j(L),
\]
Here $s_j(L)$ are the singular
values of $L$, counted with multiplicity.
\end{itemize}

\begin{lemma}\label{lem: local heat fp}
Uniformly in $q$,
\begin{equation}\label{lem: local fp}
 \left|\FP\int_0^1
 \Str \bigl(\tfrac12(N-n)^2\mathcal E_q
                         e^{-t\Delta_{F_q}}\bigr)\,dt\right|
 \le C.
\end{equation}
If $E_q$ is a differential operator of order at most one with support and
coefficient bounds as in Lemma~\ref{eq: variation lemma},
then
\begin{equation}\label{eq: HS norm}
 \|E_qe^{-\Delta_{F_q}/4}\|_{\mathrm{HS}}\le C.
\end{equation}
\end{lemma}

\begin{proof}
Choose a smooth scalar-invariant cutoff $0\le\rho\le1$ satisfying
\begin{equation}\label{eq:localized-rho}
\begin{gathered}
 \rho=0
 \quad\text{on }
 \{|z|\le6\}\cup\{|\xi|\ge(\varepsilon/4)|w|\},\\
 \rho=1
 \quad\text{on }
 \{|z|\ge12,\ |\xi|\le(\varepsilon/8)|w|\},\\
 \rho(rz)=\rho(z)
 \quad (|z|\ge12,\ r\ge1),\\
 |\nabla_\xi\rho|
 \le C\varepsilon^{-2}|w|^{-2}|\xi|
 \quad (w\ne0).
\end{gathered}
\end{equation}
Fix $0<b\ll\varepsilon^2$, decrease $q_0$ so that
$|a_q|\le b/2$, and set
\begin{equation}\label{eq:localized-function-modification}
\begin{gathered}
 \widetilde F_q=F_q+b\rho w^n,\\
 \widetilde\Delta_{F_q}
 =
 \tfrac12\left[
 d+d(2\operatorname{Re}\widetilde F_q)\wedge,\,
 \bigl(d+d(2\operatorname{Re}\widetilde F_q)\wedge\bigr)^*
 \right].
\end{gathered}
\end{equation}

We can verify that the $1/2$-regular tameness conditions in \cite{DY} hold uniformly:
\begin{equation}\label{eq: reg tame}
\begin{gathered}
 \lim_{R\to\infty}
 \inf_{\substack{|q|\le q_0\\ |z|\ge R}}
 |\nabla(2\operatorname{Re}\widetilde F_q)(z)|
 =+\infty,\\
 \frac{|\nabla^m(2\operatorname{Re}\widetilde F_q)(z)|}
 {|\nabla(2\operatorname{Re}\widetilde F_q)(z)|^{(m+1)/2}}
 \le C_m|z|^{-(m-1)(n+1)/2}
 \le C_m
 \quad (|z|\ge12,\ m\ge1).
\end{gathered}
\end{equation}
The polynomial tameness condition in \cite{DY} also holds:
\begin{equation}\label{v6-polynomial-tame}
 \int_{\{z:\,
 |\nabla(2\operatorname{Re}\widetilde F_q)(z)|^2\le\lambda\}}
 \bigl(\lambda-
 |\nabla(2\operatorname{Re}\widetilde F_q)(z)|^2\bigr)^n\,dz
 \le C\lambda^{\,n+n/(n-1)}
 \quad(\lambda\ge1).
\end{equation}
Thus, Proceeding as in Dai--Yan \cite[Theorem~4.1 and Lemmas~3.5, 4.8]{DY},
together with homogeneity
outside $|z|=12$,  gives
\begin{equation}\label{eq:localized-modified-bounds}
\begin{gathered}
 \left|\FP\int_0^1
 \Str \bigl(\tfrac12(N-n)^2\mathcal E_qe^{-t\widetilde\Delta_{F_q}}\bigr)\,dt\right|
 \le C,\\
 \|E_qe^{-\widetilde\Delta_{F_q}/4}\|_{\mathrm{HS}}\le C.
\end{gathered}
\end{equation}

Note that the two Laplacians $\widetilde\Delta_{F_q}$ and $\Delta_{F_q}$ agree in a
neighborhood of $\operatorname{supp}\mathcal E_q$.
Moreover, for every
$z\in\operatorname{supp}(\widetilde\Delta_{F_q}-\Delta_{F_q})$,
we have
\[
\operatorname{dist}\big(z,\operatorname{supp}\mathcal E_q\big)
\geq c\langle z\rangle,
\]
where $\langle z\rangle=(1+|z|^2)^{1/2}$.
Finite propagation speed and local elliptic estimates give,
for the corresponding heat kernels $K_{F_q}$ and
$\widetilde K_{F_q}$, on $\operatorname{supp}\mathcal E_q$, for some large $L$,
\begin{equation}\label{eq:localized-kernel-difference}
 \left|
 \left.
 \partial_x^\alpha\partial_y^\beta
 (K_{F_q}-\widetilde K_{F_q})(t;x,y)
 \right|_{x=z,\,y=gz}
 \right|
 \le C
             e^{-c\langle z\rangle^2/t},
 \qquad
 \substack{g\in\Z_n,\ 0<t\le1,\\
            |\alpha|+|\beta|\le2}.
\end{equation}

\end{proof}

\begin{proof}[Proof of \Cref{prop:reduction of main theorem}]
We first estimate the $\mathcal E_q$-term in \eqref{eq: heat variation Fq} for
$t\ge1$. For any operator $E_q$ of order at most one with support and
coefficient bounds as in Lemma~\ref{eq: variation lemma}, choose bump function $\chi_1$ such that
$\chi_1=1$ near its support, with support in a slightly larger
region of the form \eqref{eq: error support}. By
\eqref{eq: HS norm},
\begin{equation}\label{eq: sandwich}
 \|e^{-\Delta_{F_q}/4}E_qe^{-\Delta_{F_q}/4}\|_{\mathrm{tr}}
 \le\|e^{-\Delta_{F_q}/4}\chi_1\|_{\mathrm{HS}}
       \|E_qe^{-\Delta_{F_q}/4}\|_{\mathrm{HS}}\le C.
\end{equation}
The spectral theorem and
\eqref{eq: spectral gap} imply, for $t\ge1$, 
\begin{equation}\label{eq: diff heat}
 \|{\bar{\partial}}_{F_q}e^{-(t-1/2)\Delta_{F_q}}\|+
 \|{\bar{\partial}}_{F_q}^*e^{-(t-1/2)\Delta_{F_q}}\|
       \le Ct^{-1/2}e^{-ct|q|^{4/n}}.
\end{equation}
 Cyclicity gives
\begin{equation}\label{eq: error trace fac}
\begin{split}
 &\Str (\tfrac12(N-n)^2\mathcal E_qe^{-t\Delta_{F_q}})=\Str \bigl((N-n-\tfrac12)D_q^+{\bar{\partial}}_F^*e^{-t\Delta_{F_q}}\bigr)-\Str \bigl((N-n+\tfrac12)D_q^-{\bar{\partial}}_Fe^{-t\Delta_{F_q}}\bigr).
\end{split}
\end{equation}
Note that
\be\ba\label{eq: trace norm ineq}
&\quad\left|
\operatorname{Str}\!\left(
(N-n\mp\tfrac12)D_q^\pm
\bar\partial_{F_q}^*e^{-t\Delta_{F_q}}
\right)
\right|\\
&=
\left|
\operatorname{Str}\!\left(
e^{-\Delta_{F_q}/4}
(N-n\mp\tfrac12)D_q^\pm
e^{-\Delta_{F_q}/4}
\bar\partial_{F_q}^*
e^{-(t-1/2)\Delta_{F_q}}
\right)
\right|\\
&\le
\left\|
e^{-\Delta_{F_q}/4}
(N-n\mp\tfrac12)D_q^\pm
e^{-\Delta_{F_q}/4}
\right\|_{\mathrm{tr}}
\left\|
\bar\partial_{F_q}^*
e^{-(t-1/2)\Delta_{F_q}}
\right\|.
\ea\ee

Applying \eqref{eq: sandwich} with $E_q=(N-n\mp\tfrac12)D_q^\pm$ and using
\eqref{eq: diff heat}-\eqref{eq: trace norm ineq},
\[ |\Str (\tfrac12(N-n)^2\mathcal E_qe^{-t\Delta_{F_q}})|
     \le Ct^{-1/2}e^{-ct|q|^{4/n}},\]
     which implies
\begin{equation}\label{eq: error large time}
\begin{split}
 \int_1^\infty|\Str (\tfrac12(N-n)^2\mathcal E_qe^{-t\Delta_{F_q}})|\dd t
     &\le C|q|^{-2/n}.
\end{split}
\end{equation}

Since harmonic forms have degree $n$, for each fixed $q\ne0$,
\[
 \Str \bigl(B(K)(N-n)e^{-t\Delta_{F_q}}\bigr)\longrightarrow0
 \qquad(t\to\infty).
\]
By Lemma~\ref{lem: coefficient lemma},
$\zeta_{F_q}(0)$ is independent of $q$, and
\begin{equation}\label{eq: finite part tor}
\begin{split}
 \mathscr{T}^{\mathrm{LG}}(F_q)
 ={}&-\frac12\FP\int_0^\infty
       \Str \bigl((N-n)^2e^{-t\Delta_{F_q}}\bigr)\frac{\dd t}{t}-\gamma\zeta_{F_q}(0).
\end{split}
\end{equation}
Here $\gamma$ is the Euler constant. Integrate \eqref{eq: heat variation Fq} over all times and use
\eqref{eq: finite part tor},
\begin{equation}
\begin{split}
 |q|\partial_{|q|} \mathscr{T}^{\mathrm{LG}}(F_q)
   ={}&\left[\Str \bigl(B(K)(N-n)e^{-t\Delta_{F_q}}\bigr)\right]_{t^0}+\FP\int_0^\infty
              \Str (\tfrac12(N-n)^2\mathcal E_qe^{-t\Delta_{F_q}})\dd t.
\end{split}
\end{equation}
Estimates
\eqref{eq: coefficient bounds BK}, \eqref{lem: local fp}, and
\eqref{eq: error large time} now give
\begin{equation}\label{eq: derive bound torsion}
 \left||q|\partial_{|q|} \mathscr{T}^{\mathrm{LG}}(F_q)\right|\le C(1+|q|^{-2/n}).
\end{equation}
The torsion is bounded on $|q|=q_0$ for some fixed $q_0$. Integrating
\eqref{eq: derive bound torsion} along the ray yields
\begin{equation}
 | \mathscr{T}^{\mathrm{LG}}(F_q)|
 \le C+C\int_{|q|}^{q_0}(1+s^{-2/n})\frac{\dd s}{s}
 \le C|q|^{-2/n}.
\end{equation}

\end{proof}

\appendix
\section{Quillen metrics associated with LG model}

BCOV torsion is closely related to a Quillen metric on a certain
line bundle, and its anomaly formula is related to the curvature
formula for this line bundle. In this appendix, we develop an analogous
construction for the LG model.

Decompose $\cH$ into the integer eigenspaces
$\mathcal H_a=\ker(\ph-a)$, $\cH=\bigoplus_a\mathcal H_a$, and let
$h_a=h_{F_b}|_{\mathcal H_a}$ be the restricted Hermitian metric.
Define the graded determinant line and its $L^2$ metric by
\begin{equation}\label{eq:LGdetline}
 \Lambda_{\LG}=\bigotimes_a(\det\mathcal H_a)^{(-1)^n a},
 \qquad
 h_{\det,\LG}=\bigotimes_a(\det h_a)^{(-1)^n a}.
\end{equation}
Then
\begin{equation}\label{eq:detcurvature}
R(h_{\det,\LG})=(-1)^n\sum_aa\Tr(R^{\mathcal H}|_{\mathcal{H}_a})=\Str_{\mathcal H}(\ph R^{\mathcal H}).
\end{equation}

\begin{proposition}
\label{prop:rawflat}
The metric
\begin{equation}\label{eq:rawQmetric}
h_{Q,\LG}:=e^{\cT^{\LG}}h_{\det,\LG}
\end{equation}
has zero Chern curvature.
\end{proposition}

\begin{proof}
This follows from
\eqref{eq:detcurvature} and \eqref{eq:rawanomaly}.
\end{proof}

To avoid a fractional power in the metric, define
\begin{equation}\label{eq:completedline}
 \mathscr C_{\LG}:=\Lambda_{\LG}^{\otimes12}\otimes\mathcal L_0^{\otimes\chi(X)},\qquad
 h_{\mathscr C}:=
 e^{12\cT_{\mathrm{mod}}^{\LG}}h_{\det,\LG}^{12}
 =h_{Q,\LG}^{12}\,h_F(w_0,w_0)^{\chi(X)}.
\end{equation}
Recall that $\mathcal{L}_0$ is defined in \eqref{eq: defn of L0}.

Then by Theorem \ref{thm:modifiedanomalyearly}, we have:
\begin{proposition}[Curvature of the completed line]
\label{prop:completedcurvature}
We have
\begin{equation}\label{eq:completedcurvature}
 R(h_{\mathscr C})=\chi(X)\omega_{\mathrm{WP}}^{\LG}.
\end{equation}
In the distinguished vacuum frame this is equivalent to the modified anomaly
formula \eqref{eq:main}.
\end{proposition}
\section{The  genus-one BCOV invariant}\label{app: defn BCOV}
We recall the definition of
\cite[Definition~5.7]{EFM}.
Let $X$ be a compact connected K\"ahler manifold with
$K_X\simeq\mathcal O_X$, and write $d=\dim_{\mathbb C}X$.
For the hypersurfaces considered in this paper, $d=n-2$.

Fix a K\"ahler form $\omega$. Let
\[
\Delta_{\bar\partial}
=
\bar\partial\bar\partial^*
+
\bar\partial^*\bar\partial
\]
act on $\Omega^{*,*}(X)$, and denote by $\Pi$ the orthogonal
projection onto its kernel.
For $\operatorname{Re}s>d$, define
\[
\zeta_{\mathrm{BCOV}}(s)
=
\frac{1}{\Gamma(s)}
\int_0^\infty
t^{s-1}
\operatorname{Str}\!\left(
N^{1,0}N^{0,1}\bigl(e^{-t\Delta_{\bar\partial}}-\Pi\bigr)
\right)\,dt,
\]
where $\operatorname{Str}$ is taken with respect to the canonical
$\mathbb Z_2$-grading.
The small-time expansion of the heat trace, together with its
exponential decay for large time, extends
$\zeta_{\mathrm{BCOV}}(s)$ meromorphically to the complex plane and
shows that it is regular at $s=0$.
The BCOV torsion is defined by
\begin{equation}
\label{eq:app-bcov-torsion}
T_{\mathrm{BCOV}}(X,\omega)
=
\exp\bigl(-\zeta_{\mathrm{BCOV}}'(0)\bigr).
\end{equation}

The definition of the BCOV invariant also contains a correction
comparing the Calabi--Yau volume form with the K\"ahler volume form.
Fix $0\ne\eta\in H^0(X,K_X)$ and normalize it by
\[
\|\eta\|_{L^2}^{\,2}
=
\frac{1}{(2\pi)^d}
\int_X i^{d^2}\eta\wedge\overline{\eta}.
\]
The resulting ratio of volume forms is
\[
\varphi_\omega
=
\frac{i^{d^2}\eta\wedge\overline{\eta}}
{\|\eta\|_{L^2}^{\,2}\,(\omega^d/d!)}.
\]
The associated correction factor is
\begin{equation}
\label{eq:app-bcov-a}
A(X,\omega)
=
\exp\left(
-\frac{1}{12}
\int_X
\log\varphi_\omega\,
c_d(T^{1,0}X,\omega)
\right),
\end{equation}
where $c_d(T^{1,0}X,\omega)$ is the normalized $d$-th
Chern--Weil form representative of $c_d(T^{1,0}X)$.

Let $\alpha_1,\ldots,\alpha_{b_k}$ be the harmonic representatives
of a basis of
$H^k(X,\mathbb Z)/\mathrm{torsion}$.
Set
\[
V_k(X,\omega)
=
\det\left(
\int_X \alpha_i\wedge *_\omega\alpha_j
\right)_{i,j=1}^{b_k}.
\]

Define
\begin{equation}
\label{eq:app-bcov-b}
B(X,\omega)
=
\prod_{k=1}^{2d}
V_k(X,\omega)^{(-1)^{k+1}k/2}.
\end{equation}

The BCOV invariant is
\begin{equation}
\label{eq:app-bcov-invariant}
\tau_{\mathrm{BCOV}}(X)
=
\frac{A(X,\omega)}{B(X,\omega)}
T_{\mathrm{BCOV}}(X,\omega).
\end{equation}
It is independent of $\eta$, the integral bases, and
the auxiliary K\"ahler metric
(see \cite[Proposition~5.8]{EFM}).
We use the logarithmic normalization
\begin{equation}
\label{eq:app-bcov-genus-one}
F_1^{\mathrm{BCOV}}(X)
=
-\frac12\log\tau_{\mathrm{BCOV}}(X).
\end{equation}

\section{Another proof of \Cref{prop:smalltime}}
\label{app:getzler-small-time}

We prove \Cref{prop:smalltime} without using Gaussian bridge. Fix $K\Subset B$ and write
$F=F_b$. All estimates are uniform on $K$, together with
all parameter derivatives. 
For fixed $i,j$, as in \cref{sec:smalltime}, introduce odd variables $\eta,\bar\eta$,
 and set
\[
\mathscr D_h=h^2\Delta_{\bar\partial}+\mathcal M_h,
 \qquad
 \mathcal M_h=|\partial F|^2+hL_F
 +\eta\,\partial\phi_i\wedge
       +\bar\eta\,(\partial\phi_j\wedge)^*.
\]
 Let $K_{\mathscr D_h}(t;x,y)$ denote the heat kernel of $\mathscr D_h$, and write $\operatorname{str}$
and $\operatorname{Str}$ for pointwise and operator supertraces. 
Recall that we have decompositions
\begin{align*}
 \widehat p
   &=\widehat p_{\mathrm{num}}+\widehat p_{\mathrm{diff}},\\
 \widehat p_{\mathrm{num}}
   &=\frac{n-1}{n}N^{1,0}+\frac1nN^{0,1},
 &\widehat p_{\mathrm{diff}}
   &=-\frac1n\sum_{a=1}^n
      (z_a\partial_{z_a}-\bar z_a\partial_{\bar z_a}),
\end{align*}
and
\[\Str(\widehat{p}e^{-\mathscr D_h})=\frac{1}{n}\sum_{g\in\Z_n}\Str_{L^2\Omega^{*,*}(\C^n)}(g^*\widehat{p}e^{-\mathscr D_h}).\]

We compute $\ph_{\mathrm{num}}$-contribution with $g=1$ in \Cref{subsec: getzler filtration}, the $\ph_{\mathrm{diff}}$-contribution with $g=1$ in \Cref{subsec: diff contri}, and $g\neq1$-contribution in \Cref{subsec: g neq 1}.

\subsection{The heat-kernel expansion}\label{subsec: heat kernel}
Let $ [h^j\eta\bar\eta]K_{\mathscr D_h}(1;z,z)$ denote the coefficient of $h^j\eta\bar\eta$ in the expansion of $K_{\mathscr D_h}(1;z,z)$, and let $ [\eta\bar\eta]K_{\mathscr D_h}(1;z,z)$ denote the coefficient of $\eta\bar\eta$ in the expansion of $K_{\mathscr D_h}(1;z,z)$. These notations are also used after taking the pointwise
or operator supertrace.

Proceeding as in \cite[\S3, Theorems 4.1--4.2]{DY}, one has an asymptotic expansion of the form
\begin{align}
 [\eta\bar\eta] K_{\mathscr D_h}(1;x,y)
 &\sim \frac{e^{-|x-y|^2/(2h^2)}}{(2\pi h^2)^n}
 \exp\!\left(-\int_0^1
 |\partial F((1-u)x+uy)|^2\,du\right)\cdot
 \sum_{j\geq0}
 h^{j}\Theta_{j}(x,y).
 \label{eq:getzler-kernel-expansion}
\end{align}

Moreover, $\{\Theta_j(z,z)\}$ are algebraic combinations of derivatives of $F,\bar{F}$, $\phi_i$, $\bar\phi_j$.

 We now compute $\Theta_j(z,z)$ explicitly. One may either use Getzler
rescaling, as in \cite{DY}, or proceed as follows.

Fix $0\neq z_0\in \C^n$, and consider change of variable $z\to z_0+hX$. Then the coordinates above,
\[\mathscr{D}_h=H_{0}+\mathcal{M}_h(z_0+hX),\]
where $H_0=-\frac12\sum_{\alpha=1}^{2n}\partial_{X^\alpha}^2.$
Consider
\begin{equation*}
 \qquad H_{h,z_0}:=H_0+\mathcal M_h(z_0),\quad\text{and}\quad
 \mathcal R_{h,z_0}(X)=\mathcal M_h(z_0+hX)-\mathcal M_h(z_0)
\end{equation*}
Then 
\be
\mathcal R_{h,z_0}(X) =\sum_{1\leq|\alpha|\leq2n-2}
 \frac{h^{|\alpha|}X^\alpha}{\alpha!}
 \partial^\alpha\mathcal M_h(z)|_{z=z_0},\label{eq:Taylor}
\ee
Moreover, it is straightforward to check that
\be\ba
 h^{2n}K_{{\mathscr{D}}_h}(1;z+hX,z+hY)
 =K_{H_{h,z}+\mathcal R_{h,z}}(1;X,Y).
 \label{eq:scaled-kernel}
\ea\end{equation}
Next,
for any integer $N\geq1$, by Duhamel's principal,
\begin{equation}\ba
 e^{-(H_{h,z_0}+\mathcal R_{h,z_0})}
 &=e^{-H_{h,z_0}}
 +\sum_{q=1}^N(-1)^q
 \int_{0<u_1<\cdots<u_q<1}
 e^{-u_1H_{h,z_0}}\mathcal R_{h,z_0}
 e^{-(u_2-u_1)H_{h,z_0}}\cdots\\
 &
 e^{-(1-u_q)H_{h,z_0}}\,du_1\cdots du_q
 +\mathcal E_{N+1},\label{eq:finite-Duhamel}\ea\ee
 where
\be \ba\mathcal E_{N+1}
 =(-1)^{N+1}
 \int_{0<u_1<\cdots<u_{N+1}<1}
 &e^{-u_1H_{h,z_0}}\mathcal R_{h,z_0}\cdots
 e^{-(u_{N+1}-u_N)H_{h,z_0}}\mathcal R_{h,z_0}\\
 &\cdot
 e^{-(1-u_{N+1})(H_{h,z_0}+\mathcal R_{h,z_0})}
 \,du_1\cdots du_{N+1}.
\ea\ee
Since $H_0$ commutes with  $\mathcal M_h(z_0)$,
\begin{equation}
 K_{H_{h,z_0}}(u;X,Y)
 =G_u(X,Y)e^{-u\mathcal M_h(z_0)},
 \label{eq:free-kernel}
\end{equation}
where \[G_u(X,Y)=(2\pi u)^{-n}e^{-|X-Y|^2/(2u)}\]
and $K_{H_{h,z_0}}$ is the integral kernel of $e^{-tH_{h,z_0}}$.

At the diagonal, \eqref{eq:scaled-kernel} gives
\[
 h^{2n}K_{{\mathscr{D}}_h}(1;z_0,z_0)
 =K_{H_{h,z_0}+\mathcal R_{h,z_0}}(1;0,0),
\]
where $K_{H_{h,z}+\mathcal R_{h,z}}$ is the integral kernel of $e^{-t(H_{h,z}+\mathcal R_{h,z})}$. 

By \eqref{eq:finite-Duhamel} and \eqref{eq:free-kernel},
\be\ba\label{eq: expan Dh using duhamel}
 &\quad h^{2n}K_{{\mathscr{D}}_h}(1;z_0,z_0)
 =(2\pi)^{-n}e^{-\mathcal M_h(z_0)}
   +\mathcal{E}_{N+1}(0,0)\\
 &\quad+\sum_{q=1}^N(-1)^q
 \int_{0<u_1<\cdots<u_q<1}
 \int_{(\mathbb C^{n})^q}
 \Bigg(\prod_{\nu=0}^q
 G_{u_{\nu+1}-u_\nu}(Z_\nu,Z_{\nu+1})\Bigg)\\
 &\qquad\qquad\cdot
 e^{-u_1\mathcal M_h(z_0)}
 \mathcal R_{h,z_0}(Z_1)
 e^{-(u_2-u_1)\mathcal M_h(z_0)}
 \cdots
 \mathcal R_{h,z_0}(Z_q)
 e^{-(1-u_q)\mathcal M_h(z_0)}
 \,dZ_1\cdots dZ_q\,du_1\cdots du_q.
\ea\ee
Here, for each term in the summation corresponding to $q\geq1$, we use the
convention
\[
u_0=0,\qquad u_{q+1}=1,\qquad Z_0=Z_{q+1}=0.
\]
Moreover, \be\ba \label{eq: expression error} 
&\mathcal{E}_{N+1}(0,0)=(-1)^{N+1}
 \int_{0<u_1<\cdots<u_{N+1}<1}
 \int_{(\mathbb C^{n})^{N+1}}
 \Bigg(\prod_{\nu=0}^{N}
 G_{u_{\nu+1}-u_\nu}(Z_\nu,Z_{\nu+1})\Bigg)
 e^{-u_1\mathcal M_h(z_0)}
 \mathcal R_{h,z_0}(Z_1)\\ &\quad\cdot
 e^{-(u_2-u_1)\mathcal M_h(z_0)}\cdots
 \mathcal R_{h,z_0}(Z_{N+1})
K_{H_{h,z_0}+\mathcal{R}_{h,z_0}}(1-u_{N+1};Z_{N+1},0)
 \,dZ_1\cdots dZ_{N+1}\,du_1\cdots du_{N+1}\ea\ee
By \eqref{eq:Taylor},
\be\ba\label{eq: sub1}
 &e^{-u_1\mathcal M_h(z_0)}
 \mathcal R_{h,z_0}(Z_1)
 e^{-(u_2-u_1)\mathcal M_h(z_0)}
 \cdots
 \mathcal R_{h,z_0}(Z_q)
 e^{-(1-u_q)\mathcal M_h(z_0)}
 \\
 &\quad=
 \sum_{\substack{
       1\leq|\alpha_1|,\ldots,|\alpha_q|\leq2n-2}}
 \frac{
 h^{|\alpha_1|+\cdots+|\alpha_q|}
 Z_1^{\alpha_1}\cdots Z_q^{\alpha_q}}
 {\alpha_1!\cdots\alpha_q!}
 \\
 &\qquad\qquad\cdot
 e^{-u_1\mathcal M_h(z_0)}
 \bigl(\partial^{\alpha_1}\mathcal M_h\bigr)(z_0)
 e^{-(u_2-u_1)\mathcal M_h(z_0)}
 \cdots
 \bigl(\partial^{\alpha_q}\mathcal M_h\bigr)(z_0)
 e^{-(1-u_q)\mathcal M_h(z_0)}.
\ea\ee

Note for all $t\in(0,1],$
\be\label{eq: sub2}
 e^{-t\mathcal M_h(z_0)}
 =e^{-t|\partial F(z_0)|^2}
   \sum_{m=0}^{\infty}\frac{(-t)^m}{m!}
   \Bigl(
       hL_F(z_0)
       +\eta\,\partial\phi_i(z_0)\wedge
       +\bar\eta\,(\partial\phi_j(z_0)\wedge)^*
   \Bigr)^m.
\ee
Moreover
\be\label{eq: sub3}
 \prod_{\nu=0}^q
 e^{-(u_{\nu+1}-u_\nu)|\partial F(z_0)|^2}
 =
 e^{-|\partial F(z_0)|^2
       \sum_{\nu=0}^q(u_{\nu+1}-u_\nu)}
 =
 e^{-|\partial F(z_0)|^2}.
\ee
Also, by a straightforward computation, we see that if $s_0+\cdots+s_q=t\leq2$, $Z_0=X$,
$Z_{q+1}=Y$, and $v_j=s_0+\cdots+s_{j-1}$, then
\begin{equation}\label{eq: integral G}
 \int\prod_{\nu=0}^qG_{s_\nu}(Z_\nu,Z_{\nu+1})
 \prod_{\nu=1}^q\langle Z_\nu\rangle^{m_\nu}
 \,dZ_1\cdots dZ_q
 \leq C\big(\langle X\rangle\langle Y\rangle\big)^{m_1+\cdots+m_q}G_t(X,Y),
\end{equation}
where recall that $\langle X\rangle:=\sqrt{1+|X|^2}$.

By \eqref{eq: expression error} and \eqref{eq: integral G}, it is straightforward to see that, for some
$c>0$ and  sufficiently large $N$,
\be\label{eq: error bound}
\mathcal{E}_{N+1}(0,0)
=
o\big(h^{2n+2}e^{-c|\p F(z_0)|^2}\big).
\ee
 Substituting
\eqref{eq: sub1}--\eqref{eq: sub3} into
\eqref{eq: expan Dh using duhamel}, noting \eqref{eq: integral G} and \eqref{eq: error bound} and carrying out the integrations, we see
that $\Theta_j(z,z),j\leq2n$ can be written as a sum of \be\label{eq: akl} a_{k,\ell}\quad\mbox{with}\quad
k+\ell=j,\ee where each $a_{k,\ell}$ is a finite linear combination of
expressions of the form
\begin{equation}\label{eq:getzler-coefficient-monomial}
\left(
\prod_{\nu=1}^{r}
\partial^{\delta_\nu}|\partial F|^2
\right)
[\eta\bar\eta]\big(T_1\cdots T_{k+2}\big),
\end{equation}
evaluated at $z$. The ordered list $(T_1,\ldots,T_{k+2})$
is a permutation of
\[
\partial^{\alpha_1}L_F,\ldots,\partial^{\alpha_k}L_F,
\quad
\eta\,\partial^\beta(\partial\phi_i\wedge),
\quad
\bar\eta\,\partial^\gamma\bigl((\partial\phi_j\wedge)^*\bigr),
\]
where the multiindices refer to derivatives in $\R^{2n}\cong\C^n$ with
\[
 r\geq0,
 \qquad
 |\delta_\nu|\geq1,
 \qquad
 \ell=
\sum_{\mu=1}^{k}|\alpha_\mu|+|\beta|+|\gamma|
 +\sum_{\nu=1}^{r}|\delta_\nu|.
\]

\subsection{Getzler filtration}\label{subsec: getzler filtration}
Set
\[
(e_1,\ldots,e_{2n})
=
(dz_1,\ldots,dz_n,d\bar z_1,\ldots,d\bar z_n),
\]
and  $e_a^\vee$ be its dual frame. Then consider
\[
\varepsilon_a:=e_a\wedge,
\qquad
\iota_a:=\iota_{e_a^\vee}.
\]
Then
\[
\varepsilon_a\varepsilon_b+\varepsilon_b\varepsilon_a=0,
\qquad
\iota_a\iota_b+\iota_b\iota_a=0,
\qquad
\iota_a\varepsilon_b+\varepsilon_b\iota_a
=\delta_{ab}\operatorname{Id}.
\]

For $r\geq0$, define the Getzler filtration by
\[
\mathcal F^r
:=
\operatorname{span}
\left\{
\prod_{a=1}^{2n}
\varepsilon_a^{\alpha_a}\iota_a^{\beta_a}
\;\middle|\;
\alpha,\beta\in\{0,1\}^{2n},
\quad
\sum_{a=1}^{2n}(\alpha_a+\beta_a)\leq r
\right\}.
\]
We have
\[
\operatorname{str}
\left(
\prod_{a=1}^{2n}
\varepsilon_a^{\alpha_a}\iota_a^{\beta_a}
\right)
=
\begin{cases}
1, & \alpha_a=\beta_a=1
     \quad\text{for every }a,\\
0, & \text{otherwise}.
\end{cases}
\]
Consequently,
\[
\operatorname{str}\big|_{\mathcal F^{4n-1}}=0.
\]

We have for $a_{k,\ell}$ in \eqref{eq: akl},
\[
a_{k,\ell}\in\mathcal{F}^{2k+2},
 \qquad
\widehat p_{\mathrm{num}}a_{k,\ell}\in\mathcal{F}^{2k+4}.
\]
Note that each summand of $\partial^\alpha L_F$ changes the holomorphic
form degree by $1$ or $-1$, while one $\partial^\alpha(\partial\phi_i\wedge)$ and one $\partial^\alpha
\bigl((\partial\phi_j\wedge)^*\bigr)$
gives total holomorphic degree zero. Note also that the reflection $Y\mapsto-Y$
preserves the heat kernel and the diagonal corresponds to $Y=0$.
It follows that
\[
 \begin{aligned}
 \operatorname{str}(\widehat p_{\mathrm{num}}a_{k,\ell})
   &=0
   &&\text{if }k<2n-2\text{ or }k\text{ is odd},\\
 a_{k,\ell}
   &=0
   &&\text{if }\ell\text{ is odd}.
 \end{aligned}
\]
Thus a nonzero contribution  to supertrace has order $h^{k+\ell-2n}$, with
$k\geq2n-2$ and $k,\ell\in2\mathbb Z_{\geq0}$. In particular,
\begin{equation}
 \begin{aligned}
 k+\ell<2n
   &\quad\Longrightarrow\quad(k,\ell)=(2n-2,0),\\
 k+\ell=2n
   &\quad\Longrightarrow\quad
     (k,\ell)=(2n,0)\ \text{or}\ (2n-2,2).
 \end{aligned}
 \label{eq:getzler-orders}
\end{equation}

Using the same argument as in Step~3 of the proof of
\Cref{prop:smalltime}, which does not involve the Brownian-bridge argument,
the contribution of $(k,\ell)=(2n,0)$ to  $ [h^{2n}\eta\bar\eta]
   \operatorname{str}
     (\widehat p_{\mathrm{num}}e^{-\mathcal M_h})$ is
\begin{equation}\label{eq:getzler-frozen-density}
-\frac{e^{-|\partial F|^2}Q_{i\bar j}}{12(2\pi)^n}
\Delta_{\bar\partial}|\partial F|^2.
\end{equation}

To find $(k,\ell)=(2n-2,2)$ contribution to  the constant term of $ [h^{2n}\eta\bar\eta]
   \operatorname{str}
     (\widehat p_{\mathrm{num}}e^{-\mathcal M_h})$. Note
the identity
\[
e^{-uH_0}Y^\alpha
=(Y^\alpha+u\partial_{Y^\alpha})e^{-uH_0}
\]
gives
\begin{align*}
 \bigl(
   e^{-uH_0}Y^\alpha Y^\beta e^{-(1-u)H_0}
 \bigr)(0,0)
   &=(2\pi)^{-n}u(1-u)\delta_{\alpha\beta},\\
 \bigl(
   e^{-uH_0}Y^\alpha e^{-(v-u)H_0}
   Y^\beta e^{-(1-v)H_0}
 \bigr)(0,0)
   &=(2\pi)^{-n}u(1-v)\delta_{\alpha\beta},
   \quad 0<u<v<1.
\end{align*}
Thus by Duhamel's principal, the $(2n-2,2)$ contribution is 
\begin{align*}
 &-\frac1{2(2\pi)^n}[h^{2n-2}\eta\bar\eta]
   \sum_\alpha\int_0^1u(1-u)\operatorname{str}\Bigl(
     \widehat p_{\mathrm{num}}e^{-u\mathcal M_h}
     (\partial_\alpha^2\mathcal M_h)
     e^{-(1-u)\mathcal M_h}
   \Bigr)\,du\\
 &\quad+\frac1{(2\pi)^n}[h^{2n-2}\eta\bar\eta]
   \sum_\alpha\int_{0<u<v<1}u(1-v)\operatorname{str}\Bigl(
     \widehat p_{\mathrm{num}}e^{-u\mathcal M_h}
     (\partial_\alpha\mathcal M_h)\\
 &\hspace{9em}\cdot e^{-(v-u)\mathcal M_h}
     (\partial_\alpha\mathcal M_h)e^{-(1-v)\mathcal M_h}
   \Bigr)\,du\,dv.
\end{align*}
Using
\[
 \frac12\int_0^1u(1-u)\,du=\frac1{12},
 \qquad
 \int_{0<u<v<1}u(1-v)\,du\,dv=\frac1{24},
\]
we obtain the $(2n-2,2)$ contribution is 
\begin{equation}
 \frac1{(2\pi)^n}[h^{2n-2}\eta\bar\eta]\operatorname{str}
 \left[
   \widehat p_{\mathrm{num}}e^{-\mathcal M_h}
   \left(
     \frac16\Delta_{\bar\partial}\mathcal M_h
     +\frac1{24}\sum_\alpha(\partial_\alpha\mathcal M_h)^2
   \right)
 \right].
 \label{eq:getzler-spatial}
\end{equation}
By cyclicity and Duhamel's principal,
\[
 [h^{2n-2}\eta\bar\eta]\operatorname{str}\left[
   \widehat p_{\mathrm{num}}
   \left(
     \Delta_{\bar\partial}e^{-\mathcal M_h}
     +e^{-\mathcal M_h}\Delta_{\bar\partial}\mathcal M_h
     +\frac12e^{-\mathcal M_h}
       \sum_\alpha(\partial_\alpha\mathcal M_h)^2
   \right)
 \right]=0.
\]
Note that
\[
 \Delta_{\bar\partial}L_F=0,
 \qquad
 \Delta_{\bar\partial}(\partial\phi_i\wedge)=0,
 \qquad
 \Delta_{\bar\partial}(\partial\phi_j\wedge)^*=0,
\]
so
\[
 \Delta_{\bar\partial}\mathcal M_h
   =\Delta_{\bar\partial}|\partial F|^2.
\]
Hence \eqref{eq:getzler-spatial} equals
\begin{align*}
 &\frac1{12(2\pi)^n}[h^{2n-2}\eta\bar\eta]\operatorname{str}
   \left[
     \widehat p_{\mathrm{num}}
     \left(
       e^{-\mathcal M_h}\Delta_{\bar\partial}\mathcal M_h
       -\Delta_{\bar\partial}e^{-\mathcal M_h}
     \right)
   \right]\\
 &\quad=\frac1{12(2\pi)^n}
   \left(
     e^{-|\partial F|^2}Q_{i\bar j}
       \Delta_{\bar\partial}|\partial F|^2
     -\Delta_{\bar\partial}
       \left(e^{-|\partial F|^2}Q_{i\bar j}\right)
   \right).
\end{align*}
Adding \eqref{eq:getzler-frozen-density} gives
\begin{equation}
 [h^0\eta\bar\eta]\operatorname{str}
   \bigl(\widehat p_{\mathrm{num}}K_{\mathscr D_h}(1;z,z)\bigr)
 =-\frac1{12(2\pi)^n}\Delta_{\bar\partial}
   \left(e^{-|\partial F|^2}Q_{i\bar j}\right).
 \label{eq:getzler-density}
\end{equation}

\subsection{The $\widehat p_{\mathrm{diff}}$ contribution}\label{subsec: diff contri}

Recall that
\[
 \widehat p_{\mathrm{diff}}
 =-\frac1n\sum_{a=1}^n
 \left(
 z_a\partial_{z_a}
 -\bar z_a\partial_{\bar z_a}
 \right).
\]
Its corresponding contribution is
\begin{align}
 -\frac1n\int_{\mathbb C^n}
 [\eta\bar\eta]\,
 \str\Bigg(
 \left.
 \sum_{a=1}^n
 \left(
 x_a\partial_{x_a}
 -\bar x_a\partial_{\bar x_a}
 \right)
 K_{\mathscr D_h}(1;x,y)
 \right|_{x=y=z}
 \Bigg)\,d^{2n}z.
 \label{eq:identity-diff-trace}
\end{align}

Fix $z$ and set
\[
 x=z+hX,\qquad y=z+hY,\qquad
 \widetilde K_{h,z}(X,Y)
   :=h^{2n}K_{\mathscr D_h}(1;z+hX,z+hY).
\]
Then
\begin{equation}
 \left.
   \partial_{x^\alpha}K_{\mathscr D_h}(1;x,y)
 \right|_{x=y=z}
 =
 h^{-2n-1}
 \left.
   \partial_{X^\alpha}\widetilde K_{h,z}(X,Y)
 \right|_{X=Y=0}.
 \label{eq:getzler-endpoint-scaling}
\end{equation}

In particular,
\begin{align*}
 &\left.
 \left(
 x_a\partial_{x_a}
 -\bar x_a\partial_{\bar x_a}
 \right)
 K_{\mathscr D_h}(1;x,y)
 \right|_{x=y=z}\\
 &\quad=
 h^{-2n-1}
 \left.
 \left(
 (z_a+hX_a)\partial_{X_a}
 -(\bar z_a+h\bar X_a)\partial_{\bar X_a}
 \right)
 \widetilde K_{h,z}(X,Y)
 \right|_{X=Y=0}\\
 &\quad=
 h^{-2n-1}
 \left.
 \left(
 z_a\partial_{X_a}
 -\bar z_a\partial_{\bar X_a}
 \right)
 \widetilde K_{h,z}(X,Y)
 \right|_{X=Y=0}.
\end{align*}
Thus
\begin{align}
 &[\eta\bar\eta]\,
 \Str\left(\widehat p_{\mathrm{diff}}e^{-\mathscr D_h}\right)
 \notag\\
 &\quad=
 -\frac{h^{-2n-1}}n
 \int_{\mathbb C^n}
 \str\Bigg(
 \left.
 \sum_{a=1}^n
 \left(
 z_a\partial_{X_a}
 -\bar z_a\partial_{\bar X_a}
 \right)
 [\eta\bar\eta]\widetilde K_{h,z}(X,Y)
 \right|_{X=Y=0}
 \Bigg)\,d^{2n}z.
 \label{eq:identity-diff-scaled-trace}
\end{align}

Note that\be\label{eq:derivative 0}
\partial_{X^\alpha}\Big((2\pi)^{-n}e^{-|X-Y|^2/2}\Big)|_{X=Y=0}=0,\ee
proceeding as in the computation of $a_{k,\ell}$,
\begin{equation}
 \left(
 z_a\partial_{X_a}
 -\bar z_a\partial_{\bar X_a}
 \right)[\eta\bar\eta]\widetilde{K}_{h,z}(X,Y)|_{X=Y=0}\sim
 \frac{e^{-|\partial F(z)|^2}}{(2\pi )^n}
\sum_{k,\ell\geq0}h^{k+\ell}b_{k,\ell}(z).
 \label{eq:getzler-kernel-expansion0}
\end{equation}
Here $b_{k,\ell}$ are algebraic combinations of derivatives of $F,\bar{F}$, $\phi_i$, $\bar\phi_j$. Moreover, it is of the form of \eqref{eq:getzler-coefficient-monomial}.

Moreover, it follows from the same argument as in \Cref{subsec: getzler filtration},
\[
 \operatorname{str}\left(
     b_{k,\ell}
 \right)\neq0
 \quad\Longrightarrow\quad
 \begin{cases}
   2k+2\geq4n,\\
   k\in2\mathbb N_0,\\
   \ell\in2\mathbb N_0+1.
 \end{cases}
\]

Hence the only term that can contribute through order $h^0$ is
$b_{2n,1}$. Moreover, by \eqref{eq:derivative 0} and
\eqref{eq:getzler-kernel-expansion}, it is straightforward to verify that
\begin{equation}\label{eq: b is derivative}
\str\big(b_{2n,1}(z)\big)
=
\str\left(
\left.D_x\Theta_{2n}(x,y)\right|_{x=y=z}
\right)
=
[h^{0}\eta\bar\eta]\operatorname{str}
\left(
\left.
D_xK_{\mathscr D_h}(1;x,y)
\right|_{x=y=z}
\right),
\end{equation}
where $D_x$ is a differential operator consisting only of first-order
derivatives in the first variable.

Note that the heat kernel of $H_0$ satisfies
\begin{align*}
 &\int_{\mathbb R^{2n}}
   e^{-uH_0}(X,W)\,W^\beta\,
   e^{-(1-u)H_0}(W,Y)\,dW\\
 &\qquad=
 \bigl((1-u)X^\beta+uY^\beta\bigr)e^{-H_0}(X,Y),
\end{align*}
and thus
\[
 \left.
 \partial_{X^\alpha}
 \int_{\mathbb R^{2n}}
   e^{-uH_0}(X,W)\,W^\beta\,
   e^{-(1-u)H_0}(W,Y)\,dW
 \right|_{X=Y=0}
 =
 \frac{1-u}{(2\pi)^n}\delta_{\alpha\beta}.
\]
Duhamel's formula gives
\be\ba \label{eq:getzler-endpoint-constant}
 &[h^0\eta\bar\eta]\str
\left(\left.
\partial_{x^\alpha}K_{\mathscr D_h}(1;x,y)
   \right|_{x=y=z}\right)\\
 &\qquad=
 -\frac1{(2\pi)^n}[h^{2n}\eta\bar\eta]
 \int_0^1(1-u)\operatorname{str}\left(
   e^{-u\mathcal M_h(z)}
   (\partial_\alpha\mathcal M_h)(z)
   e^{-(1-u)\mathcal M_h(z)}
 \right)\,du.
\ea\ee

Since supertrace vanishes on supercommutator,
\begin{align*}
 &\operatorname{str}\left(
   e^{-u\mathcal M_h(z)}
   (\partial_\alpha\mathcal M_h)(z)
   e^{-(1-u)\mathcal M_h(z)}
 \right)\\
 &\qquad=
 \operatorname{str}\left(
   (\partial_\alpha\mathcal M_h)(z)e^{-\mathcal M_h(z)}
 \right)
 =
 -\partial_\alpha\operatorname{str}
       \left(e^{-\mathcal M_h(z)}\right).
\end{align*}
Consequently, \eqref{eq:unweightedzero-one} implies
\[[h^0\eta\bar\eta]\operatorname{str}
\left(\left.
\partial_{x^\alpha}K_{\mathscr D_h}(1;x,y)
   \right|_{x=y=z}\right)=
 \frac1{2(2\pi)^n}\partial_\alpha
[h^{2n}\eta\bar\eta]\operatorname{str}\left(e^{-\mathcal M_h(z)}\right)
 =0.
\]
The $\widehat p_{\mathrm{diff}}$-term hence
gives no contribution of order $h^0$.

\subsection{The $g\neq1$ contribution}\label{subsec: g neq 1}
We next consider the $g\neq1$ contribution. For
$g\in\mathbb Z_n\setminus\{1\}$, tracing the construction of the heat kernel
expansion in \cite{DY}, one sees that for some $c>0$, $M>0$,
\begin{equation}
\left|
[\eta\bar\eta]\,
\widehat p K_{\mathscr D_h}(1;gz,z)
\right|
\leq C h^{-2n}(|z|+h)^{2n-2}
\left(1+\frac{|z|}{h}\right)^M
e^{-c|gz-z|^2/h^2},
\end{equation}
Here the factor $(|z|+h)^{2n-2}$ comes from
$\partial\phi_i$ and $\partial^*\phi_j$.

It follows that
\[
[\eta\bar\eta]\,
\Str_{L^2\Omega^{*,*}(\mathbb C^n)}
\left(
g^*\widehat p e^{-\mathscr D_h}
\right)
=
O_{C^\infty(K)}(h^{2n-2})
=
O_{C^\infty(K)}(h^2).
\]

Thus, to obtain the constant term, it remains only to integrate the local
constant term in \eqref{eq:getzler-density}. Since
\[
|\partial F_b(z)|^2\geq c_K|z|^{2n-2}
\quad\text{and}\quad
Q_{i\bar j}\in\mathbb C[z,\bar z],
\]
Stokes' theorem gives
\[
\int_{\mathbb C^n}
\Delta_{\bar\partial}
\left(
e^{-|\partial F_b|^2}Q_{i\bar j}
\right)
d^{2n}z
=
0.
\]
Therefore,
\[
\operatorname*{FP}_{h\to0}
\operatorname{Str}
\left(
\widehat p\,e^{-D_h}
\right)^{(1,1)}
=
0.
\]

\bibliographystyle{plain}
\bibliography{references}

@article{Tang2017Torsion,
  author        = {Tang, X.},
  title         = {{$tt^*$} Geometry, Singularity Torsion and Anomaly Formulas},
  year          = {2017},
  journal = {arXiv: 1710.03915v2},
  url           = {https://arxiv.org/abs/1710.03915}
}

@book{CFKS1987,
  author    = {Cycon, H. L. and Froese, R. G. and
               Kirsch, W. and Simon, B.},
  title     = {Schr{\"o}dinger Operators:
               With Application to Quantum Mechanics and Global Geometry},
  series    = {Texts and Monographs in Physics},
  publisher = {Springer-Verlag},
  address   = {Berlin},
  year      = {1987}
}

@article{Zhang2022Extension,
  author  = {Zhang, Y.},
  title   = {An extension of {BCOV} invariant},
  journal = {International Mathematics Research Notices},
  volume  = {2022},
  number  = {8},
  pages   = {5973--6018},
  year    = {2022},
  doi     = {10.1093/imrn/rnaa265},
  eprint  = {1905.03964},
  archivePrefix = {arXiv},
  primaryClass  = {math.DG}
}

@article{Givental2001Quantization,
  author  = {Givental, A. B.},
  title   = {{Gromov--Witten invariants and quantization of
             quadratic Hamiltonians}},
  journal = {Moscow Mathematical Journal},
  volume  = {1},
  number  = {4},
  pages   = {551--568},
  year    = {2001}
}

@article{Teleman2012,
  author  = {Teleman, C.},
  title   = {The structure of $2D$ semi-simple field theories},
  journal = {Inventiones Mathematicae},
  volume  = {188},
  number  = {3},
  pages   = {525--588},
  year    = {2012},
  doi     = {10.1007/s00222-011-0352-5}
}

@article{Zhang2023Blowup,
  author  = {Zhang, Y.},
  title   = {{BCOV} invariant and blow-up},
  journal = {Compositio Mathematica},
  volume  = {159},
  number  = {4},
  pages   = {780--829},
  year    = {2023},
  doi     = {10.1112/S0010437X23007042},
  eprint  = {2003.03805},
  archivePrefix = {arXiv},
  primaryClass  = {math.DG}
}

@article{FuZhang2023BCOV,
  author  = {Fu, L. and Zhang, Y.},
  title   = {Motivic integration and birational invariance of
             {BCOV} invariants},
  journal = {Selecta Mathematica. New Series},
  volume  = {29},
  pages   = {25},
  year    = {2023},
  doi     = {10.1007/s00029-023-00832-3},
  eprint  = {2007.04835},
  archivePrefix = {arXiv},
  primaryClass  = {math.AG}
}

@article{FanLanYang2021,
  author  = {Fan, H. and Lan, T. and Yang, Z.},
  title   = {{LG/CY} Correspondence Between {$tt^*$} Geometries},
  journal = {Communications in Mathematical Research},
  volume  = {37},
  number  = {3},
  pages   = {297--349},
  year    = {2021},
  doi     = {10.4208/cmr.2020-0050},
  eprint  = {2011.14658},
  archivePrefix = {arXiv},
  primaryClass  = {math.AG}
}

@article{FanLanYang2022,
  author = {Fan, H. and Lan, T. and Yang, Z.},
  title  = {Constructing the {LG/CY} Isomorphism Between
            {$tt^*$} Geometries},
  year   = {2022},
 journal = {To appear in Frontier of Mathematics, arXiv: 2210.16747},
}

@article{WangYan2025Graphs,
  author = {Wang, M. and Yan, J.},
  title  = {Feynman Graph Integrals on {K}\"ahler Manifolds},
  year   = {2025},
  journal = {arXiv: 2507.09170},
}

@article{WangYan2026Renormalizations,
  author = {Wang, M. and Yan, J.},
  title  = {Renormalizations in Holomorphic Field Theories
            on {K}\"ahler Manifolds},
  year   = {2026},
journal = {arXiv: 2608.00546},

}

@article{LiWang2026BCOV,
  author = {Li, S. and Wang, M.},
  title  = {On the Exact Quantization of {BCOV} Theory
            on {Calabi--Yau} Manifolds},
  year   = {2026},
  journal = {arXiv: 2610.21927},
  archivePrefix = {arXiv},
  primaryClass  = {math-ph},
  url = {https://arxiv.org/abs/2610.21927}
}

@article{Getzler1997,
  author  = {Getzler, Ezra},
  title   = {Intersection theory on
             {$\overline{\mathcal{M}}_{1,4}$}
             and elliptic {Gromov--Witten} invariants},
  journal = {Journal of the American Mathematical Society},
  volume  = {10},
  number  = {4},
  pages   = {973--998},
  year    = {1997},
  doi     = {10.1090/S0894-0347-97-00246-4}
}

@article{li2019l2,
  title={On the {$L^2$}-{Hodge theory of Landau-Ginzburg models}},
  author={Li, S. and Wen, H.},
  journal={Advances in Mathematics},
  volume={396},
  pages={108165},
  year={2022},
  doi={10.1016/j.aim.2021.108165},
  publisher={Elsevier}
}

@article{cecotti1991topological,
  title={Topological--anti-topological fusion},
  author={Cecotti, S. and Vafa, C.},
  journal={Nuclear Physics B},
  volume={367},
  number={2},
  pages={359--461},
  year={1991},
  publisher={Elsevier}
}

@article{SaitoFM,
 title={From primitive forms to {F}robenius manifolds},
  author={Saito, K. and Takahashi, A.},
  year={2008},
  publisher={Kyoto University, Research Institute for Mathematical Sciences},
  journal={Proc. Sympos. Pure Math.},
  volume={78},
  pages={31--48}
}

@inproceedings{Saitohigher,
  title={The higher residue pairings {$\mathcal{K}_F$} for a family of hypersurface singular points},
  author={Saito, K.},
  booktitle={Proceedings of Symposia in pure mathematics},
  volume={40},
  pages={441--463},
  year={1983},
  organization={Amer Mathematical SOC 201 Charles ST, Providence, RI 02940-2213}
}

@article{1981Primitive,
  title={Primitive forms for a universal unfolding of a function with an isolated critical point},
  author={Saito, K.},
  journal={Faculty of Science, The University of Tokyo},
  number={3},
  pages={775-792},
  year={1982}
}

@article{chang2015witten,
  title={Witten’s top {C}hern class via cosection localization},
  author={Chang, H.-L. and Li, J. and Li, W.-P.},
  journal={Inventiones mathematicae},
  volume={200},
  number={3},
  pages={1015--1063},
  year={2015},
  publisher={Springer}
}

@article{Greene1990DualityI,
  title={Duality in {Calabi-Yau} Moduli Space},
  author={Greene, B. and Plesser, M.},
  journal={Nuclear Physics B},
  year={1990},
  volume={338},
  number={1},
  pages={15-37}
}

@article{Witten1993,
  title={Algebraic Geometry Associated with Matrix Models of Two-Dimensional Gravity},
  author={Witten, E.},
  journal={Topological methods in modern mathematics},
  pages={235-269},
  year={1993}
}

@article{COGP1991,
  author = {Candelas, P. and de la Ossa, X. C. and Green, P. S. and Parkes, L.},
  title = {A pair of {Calabi--Yau} manifolds as an exactly soluble superconformal theory},
  journal = {Nuclear Physics B},
  volume = {359},
  number = {1},
  pages = {21--74},
  year = {1991},
  doi = {10.1016/0550-3213(91)90292-6},
  url = {https://www.sciencedirect.com/science/article/pii/0550321391902926}
}

@article{Witten1993Phases,
  author = {Witten, E.},
  title = {Phases of {$N=2$} theories in two dimensions},
  journal = {Nuclear Physics B},
  volume = {403},
  number = {1--2},
  pages = {159--222},
  year = {1993},
  doi = {10.1016/0550-3213(93)90033-L},
  eprint = {hep-th/9301042},
  archivePrefix = {arXiv},
  url = {https://arxiv.org/abs/hep-th/9301042}
}

@incollection{Givental1998Mirror,
  author = {Givental, A.},
  title = {A mirror theorem for toric complete intersections},
  booktitle = {Topological Field Theory, Primitive Forms and Related Topics},
  editor = {Kashiwara, M. and Matsuo, A. and Saito, K. and Satake, I.},
  series = {Progress in Mathematics},
  volume = {160},
  pages = {141--175},
  publisher = {Birkh\"auser},
  address = {Boston, MA},
  year = {1998},
  doi = {10.1007/978-1-4612-0705-4_5},
  eprint = {alg-geom/9701016},
  archivePrefix = {arXiv},
  url = {https://arxiv.org/abs/alg-geom/9701016}
}

@article{LianLiuYau1997,
  author = {Lian, B. H. and Liu, K. and Yau, S.-T.},
  title = {Mirror principle. {I}},
  journal = {Asian Journal of Mathematics},
  volume = {1},
  number = {4},
  pages = {729--763},
  year = {1997},
  doi = {10.4310/AJM.1997.v1.n4.a5},
  eprint = {alg-geom/9712011},
  archivePrefix = {arXiv},
  url = {https://arxiv.org/abs/alg-geom/9712011}
}

@article{FJR2013,
  author = {Fan, H. and Jarvis, T. J. and Ruan, Y.},
  title = {The {Witten} equation, mirror symmetry, and quantum singularity theory},
  journal = {Annals of Mathematics},
  series = {2},
  volume = {178},
  number = {1},
  pages = {1--106},
  year = {2013},
  doi = {10.4007/annals.2013.178.1.1},
  eprint = {0712.4021},
  archivePrefix = {arXiv},
  primaryClass = {math.AG},
  url = {https://annals.math.princeton.edu/2013/178-1/p01}
}

@article{ChiodoRuan2010,
  author = {Chiodo, A. and Ruan, Y.},
  title = {{Landau--Ginzburg/Calabi--Yau} correspondence for quintic three-folds via symplectic transformations},
  journal = {Inventiones Mathematicae},
  volume = {182},
  number = {1},
  pages = {117--165},
  year = {2010},
  doi = {10.1007/s00222-010-0260-0},
  eprint = {0812.4660},
  archivePrefix = {arXiv},
  primaryClass = {math.AG},
  url = {https://arxiv.org/abs/0812.4660}
}

@article{ChiodoRuan2011State,
  author = {Chiodo, A. and Ruan, Y.},
  title = {{LG/CY} correspondence: the state space isomorphism},
  journal = {Advances in Mathematics},
  volume = {227},
  number = {6},
  pages = {2157--2188},
  year = {2011},
  doi = {10.1016/j.aim.2011.04.011},
  eprint = {0908.0908},
  archivePrefix = {arXiv},
  primaryClass = {math.AG},
  url = {https://arxiv.org/abs/0908.0908}
}

@article{Zinger2009,
  author = {Zinger, A.},
  title = {The reduced genus 1 {Gromov--Witten} invariants of {Calabi--Yau} hypersurfaces},
  journal = {Journal of the American Mathematical Society},
  volume = {22},
  number = {3},
  pages = {691--737},
  year = {2009},
  doi = {10.1090/S0894-0347-08-00625-5},
  eprint = {0705.2397},
  archivePrefix = {arXiv},
  primaryClass = {math.AG},
  url = {https://arxiv.org/abs/0705.2397}
}

@article{GuoRoss2019LG,
  author = {Guo, S. and Ross, D.},
  title = {Genus-one mirror symmetry in the {Landau--Ginzburg} model},
  journal = {Algebraic Geometry},
  volume = {6},
  number = {3},
  pages = {260--301},
  year = {2019},
  doi = {10.14231/AG-2019-013},
  eprint = {1611.08876},
  archivePrefix = {arXiv},
  primaryClass = {math.AG},
  url = {https://arxiv.org/abs/1611.08876}
}

@article{GuoRoss2019Global,
  author = {Guo, S. and Ross, D.},
  title = {The genus-one global mirror theorem for the quintic {$3$}-fold},
  journal = {Compositio Mathematica},
  volume = {155},
  number = {5},
  pages = {995--1024},
  year = {2019},
  doi = {10.1112/S0010437X19007231},
  eprint = {1703.06955},
  archivePrefix = {arXiv},
  primaryClass = {math.AG},
  url = {https://arxiv.org/abs/1703.06955}
}

@article{ChangLiLiLiu2019MSP,
  author = {Chang, H.-L. and Li, J. and Li, W.-P. and Liu, C.-C. M.},
  title = {Mixed-{Spin-P} fields of {Fermat} polynomials},
  journal = {Cambridge Journal of Mathematics},
  volume = {7},
  number = {3},
  pages = {319--364},
  year = {2019},
  doi = {10.4310/CJM.2019.v7.n3.a3},
  eprint = {1505.07532},
  archivePrefix = {arXiv},
  primaryClass = {math.AG},
  url = {https://arxiv.org/abs/1505.07532}
}

@article{ChangGuoLi2021,
  author = {Chang, H.-L. and Guo, S. and Li, J.},
  title = {Polynomial structure of {Gromov--Witten} potential of quintic {$3$}-folds},
  journal = {Annals of Mathematics},
  series = {2},
  volume = {194},
  number = {3},
  pages = {585--645},
  year = {2021},
  doi = {10.4007/annals.2021.194.3.1},
  eprint = {1809.11058},
  archivePrefix = {arXiv},
  primaryClass = {math.AG},
  url = {https://annals.math.princeton.edu/2021/194-3/p01}
}

@article{RaySinger1973,
  author = {Ray, D. B. and Singer, I. M.},
  title = {Analytic torsion for complex manifolds},
  journal = {Annals of Mathematics},
  volume = {98},
  number = {1},
  pages = {154--177},
  year = {1973},
  doi = {10.2307/1970909},
  url = {https://annals.math.princeton.edu/1973/98-1/p05}
}

@article{BCOV1993,
  author = {Bershadsky, M. and Cecotti, S. and Ooguri, H. and Vafa, C.},
  title = {Holomorphic anomalies in topological field theories},
  journal = {Nuclear Physics B},
  volume = {405},
  number = {2--3},
  pages = {279--304},
  year = {1993},
  doi = {10.1016/0550-3213(93)90548-4},
  eprint = {hep-th/9302103},
  archivePrefix = {arXiv},
  note = {With an appendix by Katz, S.},
  url = {https://arxiv.org/abs/hep-th/9302103}
}

@article{EFM,
  author = {Eriksson, D. and {Freixas i Montplet}, G. and Mourougane, C.},
  title = {{BCOV} invariants of {Calabi--Yau} manifolds and degenerations of {Hodge} structures},
  journal = {Duke Mathematical Journal},
  volume = {170},
  number = {3},
  pages = {379--454},
  year = {2021},
  doi = {10.1215/00127094-2020-0045},
  eprint = {1809.05452},
  archivePrefix = {arXiv},
  url = {https://arxiv.org/abs/1809.05452}
}

@article{EFMMirror,
  author = {Eriksson, D. and {Freixas i Montplet}, G. and Mourougane, C.},
  title = {On genus one mirror symmetry in higher dimensions and the {BCOV} conjectures},
  journal = {Forum of Mathematics, Pi},
  volume = {10},
  pages = {e19},
  year = {2022},
  doi = {10.1017/fmp.2022.13},
  eprint = {1911.06734},
  archivePrefix = {arXiv},
  url = {https://arxiv.org/abs/1911.06734}
}

@misc{LiLiSaito,
  author = {Li, C. and Li, S. and Saito, K.},
  title = {Primitive forms via polyvector fields},
  year = {2013},
  eprint = {1311.1659},
  archivePrefix = {arXiv},
  primaryClass = {math.AG},
  note = {arXiv:1311.1659},
  url = {https://arxiv.org/abs/1311.1659}
}

@misc{CostelloLiBCOV,
  author = {Costello, K. J. and Li, S.},
  title = {Quantum {BCOV} theory on {Calabi--Yau} manifolds and the higher genus {B}-model},
  year = {2012},
  eprint = {1201.4501},
  archivePrefix = {arXiv},
  primaryClass = {math.QA},
  note = {arXiv:1201.4501},
  url = {https://arxiv.org/abs/1201.4501}
}

@article{HeLiShenWebb2022,
  author = {He, W. and Li, S. and Shen, Y. and Webb, R.},
  title = {{Landau--Ginzburg} mirror symmetry conjecture},
  journal = {Journal of the European Mathematical Society},
  volume = {24},
  number = {8},
  pages = {2915--2978},
  year = {2022},
  doi = {10.4171/JEMS/1155},
  eprint = {1503.01757},
  archivePrefix = {arXiv},
  url = {https://ems.press/journals/jems/articles/3324287}
}

@article{BK,
  author  = {Bismut, J.-M. and K{\"o}hler, K.},
  title   = {Higher analytic torsion forms for direct images
             and anomaly formulas},
  journal = {Journal of Algebraic Geometry},
  volume  = {1},
  number  = {4},
  pages   = {647--684},
  year    = {1992}
}

@article{GueneysuMatrix,
  author  = {G{\"u}neysu, B.},
  title   = {Multiplicative matrix-valued functionals and the continuity properties of semigroups corresponding to partial differential operators with matrix-valued coefficients},
  journal = {Journal of Mathematical Analysis and Applications},
  volume  = {380},
  number  = {2},
  pages   = {709--725},
  year    = {2011},
  doi     = {10.1016/j.jmaa.2011.02.038},
}

@article{DY,
  author        = {Dai, X. and Yan, J.},
  title         = {Witten deformation on non-compact manifolds: heat kernel expansion and local index theorem},
  journal       = {Mathematische Zeitschrift},
  volume        = {303},
  year          = {2023},
  doi           = {10.1007/s00209-022-03150-0},
}

@article{BGS2,
  author  = {Bismut, J.-M. and Gillet, H. and Soul{\'e}, C.},
  title   = {Analytic Torsion and Holomorphic Determinant Bundles.{ II}.
             Direct Images and {Bott--Chern} Forms},
  journal = {Communications in Mathematical Physics},
  volume  = {115},
  number  = {1},
  pages   = {79--126},
  year    = {1988},
  doi     = {10.1007/BF01238854}
}

@book{zhang2001lectures,
  title={Lectures on {Chern-Weil} theory and {Witten} deformations},
  author={Zhang, W.},
  volume={4},
  year={2001},
  publisher={World Scientific}
}

@article{FF,
  author  = {Fan, H. and Fang, H.},
  title   = {Torsion Type Invariants of Singularities},
  journal = {Vietnam Journal of Mathematics},
  volume  = {49},
  number  = {2},
  pages   = {381--432},
  year    = {2021},
  doi     = {10.1007/s10013-021-00510-x}
}

@article{Fan,
  author        = {Fan, H.},
  title         = {Schr{\"o}dinger equations, deformation theory and {$tt^*$}-geometry},
  year          = {2011},
  journal        = {arXiv:1107.1290},
  archivePrefix = {}
}

@article{TY,
  author        = {Tang, X. and Yan, J.},
  title         = {{CY/LG} correspondence for {Weil--Petersson} metrics and {$tt^*$} structures},
  journal       = {Journal of Geometry and Physics},
  volume        = {216},
  pages         = {105545},
  year          = {2025},
  eprint        = {2205.05791},
  archivePrefix = {arXiv}
}

@article{FLY,
  author        = {Fang, H. and Lu, Z. and Yoshikawa, K.-I.},
  title         = {Analytic torsion for {Calabi--Yau} threefolds},
  journal       = {Journal of Differential Geometry},
  volume        = {80},
  pages         = {175--259},
  year          = {2008},
}

@article{BCOV,
  author  = {Bershadsky, M. and Cecotti, S. and Ooguri, H. and Vafa, C.},
  title   = {Kodaira--{S}pencer theory of gravity and exact results for quantum string amplitudes},
  journal = {Communications in Mathematical Physics},
  volume  = {165},
  pages   = {311--427},
  year    = {1994}
}

\end{document}